\pdfoutput=1
\documentclass[12pt]{article}
\usepackage[hmargin=1in,top=1in,bottom=1.25in]{geometry}

\usepackage{style/qstyle}
\usepackage[normalem]{ulem}

\renewcommand{\emph}[1]{\textbf{#1}}
\renewcommand{\mathbf}[1]{\textit{\textbf{#1}}}

\newcommand{\ecnote}[1]{{\color{orange}\small[E: #1]}}

\title{Existence and nonexistence of commutativity gadgets\\ for entangled CSPs}

\author[1]{Eric Culf}
\author[2,3]{Josse van Dobben de Bruyn}
\author[4]{Matthijs Vernooij}
\author[5]{Peter Zeman}

\affil[1]{\small{Department of Applied Mathematics and Institute for Quantum Computing, University of Waterloo, Canada, \texttt{eculf@uwaterloo.ca}}}
\affil[2]{Department of Applied Mathematics and Computer Science, Technical University of Denmark, Denmark}
\affil[3]{Current address: Department of Applied Mathematics, Faculty of Mathematics and Physics, Charles University, Czech Republic, \texttt{josse.van-dobben-de-bruyn@matfyz.cuni.cz}}
\affil[4]{Delft Institute of Applied Mathematics, TU Delft, The Netherlands, \texttt{m.n.a.vernooij@tudelft.nl}}
\affil[5]{Department of Algebra, Faculty of Mathematics and Physics, Charles University, Czech Republic, \texttt{peter.zeman@matfyz.cuni.cz}}

\date{}

\begin{document}

\maketitle

\begin{abstract}
    Commutativity gadgets allow NP-hardness proofs for classical constraint satisfaction problems (CSPs) to be carried over to undecidability proofs for the corresponding entangled CSPs.
    This has been done, for instance, for NP-complete boolean CSPs and 3-colouring in the work of Culf and Mastel.
    For many CSPs over larger alphabets, including $k$-colouring when $k \geq 4$, it is not known whether or not commutativity gadgets exist, or if the entangled CSP is decidable.
    In this paper, we initiate a systematic study of commutativity gadgets and prove the first known obstruction to their existence.
    We do this by extending the definition of the quantum automorphism group of a graph to the quantum endomorphism monoid of a CSP, and showing that a CSP with non-classical quantum endomorphism monoid does not admit a commutativity gadget.
    As a consequence of this, we get that no commutativity gadget exists for $k$-colouring when $k \geq 4$. However, we construct a commutativity gadget for an alternate way of presenting $k$-colouring as a nonlocal game, the oracular setting.
    
    Furthermore, we prove an easy to check sufficient condition for the quantum endomorphism monoid to be non-classical, extending a result of Schmidt for the quantum automorphism group of a graph, and use this to give examples of CSPs that do not admit a commutativity gadget. We also show that existence of commutativity gadgets and oracular commutativity gadgets is equivalent for graphs with no four-cycle; and that the odd cycles and the odd graphs have a commutative quantum endomorphism monoid, leaving open the possibility that they might admit a commutativity gadget.
\end{abstract}

\newpage

\setcounter{tocdepth}{2}
\tableofcontents

\newpage

\section{Preface: the no-go theorem for graph colouring}\label{sec:preface}

The recent groundbreaking work known as MIP*=RE has shown that it is undecidable to approximate the quantum value of a general nonlocal game~\cite{JNV+21}. A fruitful line of research aiming to understand this result better asks which more restricted classes of games preserve this general undecidability~\cite{CM24,TV25,MSSV25}. One of the pathways to undecidability is as follows: for problems whose classical version is $\tsf{NP}$-complete, undecidability can be guaranteed by the existence of an object called a commutativity gadget~\cite{CM24}. In this work, we study the existence of such gadgets. While our work is able to capture a wide range of problems, the key idea becomes clear from the simple case of $k$-colouring of graphs. Therefore, to begin, we devote this section to work out our results for this case. Afterwards, we introduce our general setting of constraint satisfaction problems and prove the full extent of our results.

A \emph{$k$-colouring} of a graph~$H$ is an assignment of one of $k$ colours to each of the vertices of~$H$ such that no two adjacent vertices are assigned the same colour. Equivalently, a $k$-colouring is a graph homomorphism from $H$ to the complete graph $K_k$ --- a map on the vertices that maps adjacent vertices to adjacent vertices. 

Graph colouring naturally gives rise to the following computational problem: given a graph $H$, decide if it admits a $k$-colouring. In this setting, graph colouring is a specific example of a constraint satisfaction problem (CSP). This problem is $\tsf{NP}$-complete for $k\geq 3$. Since it is a CSP, $k$-colouring is $\tsf{NP}$-complete in a particularly nice way: any boolean satisfiability formula can be expressed as a graph colouring instance, where each clause is separately rewritten as part of the graph to be coloured. 

Alternately, the $k$-colouring problem can be presented as a two-prover one-round interactive proof system, or \emph{nonlocal game}. In a nonlocal game, two players, Alice and Bob, attempt to satisfy a known winning condition without communicating. To play the game associated to a $k$-colouring instance $H$, the referee samples two random vertices $x$ and $y$, and sends $x$ to Alice and $y$ to Bob. Each player responds with a vertex of $K_k$, $a$ and $b$ respectively. To win, the players must respond with the same vertex if $x=y$ and unequal vertices if $x$ and $y$ are adjacent. If a strategy wins with probability $1$, we say it is perfect. This characterisation of the $k$-colouring is equivalent to the standard one in the sense that any perfect deterministic strategy (Alice and Bob's answers are simply functions of their inputs) is given by a graph homomorphism --- Alice and Bob's functions are equal and are a homomorphism $H\rightarrow K_k$. 

One can also extend the class of strategies allowed in this game to find different notions of graph homomorphism. If the players can communicate, it is clear that they can always win. However, they can access stronger types of correlation while still being forbidden from communicating. First, if the players have access to shared randomness, it does not change the class of four-colourable graphs, as any such probabilistic strategy is simply a convex combination of deterministic strategies. However, there is a significant difference if the players share quantum entanglement; there exist graphs that are quantum $k$-colourable but not $k$-colourable~\cite{CMN+07}.

In an entangled strategy, Alice and Bob's systems are correlated via a quantum state $\ket{\psi}_{AB}$. Then, to sample their answers, Alice and Bob measure their systems with projective measurements $\{P^x_a\}_a$ and $\{Q^y_b\}_b$, respectively. If their strategy is perfect, then the state $\ket{\psi}$ is maximally entangled and $Q^x_a=(P^x_a)^T$, since the game is synchronous (meaning that the players must give the same answer if they received the same question). In this setting, perfect winning has the following simple characterisation: $P^x_aP^y_a=0$ if $x$ is adjacent to $y$ in $H$. We call a collection of projective measurements satisfying this condition a \emph{quantum $k$-colouring} or a \emph{quantum graph homomorphism} from~$H$ to~$K_k$. If the projections pairwise commute, this reduces to a direct sum of classical colourings.

Analogously to the classical setting, the question of whether a given graph $H$ admits a quantum $k$-colouring is a natural computational problem. What is the complexity of this decision problem? In the case of three-colouring, the answer is known: it is in fact undecidable to determine whether a graph admits a quantum three-colouring~\cite{CM24}. This reduction is possible by means of a tool used to lift classical reductions between constraint satisfaction problems (CSPs) to quantum-sound reductions --- a \emph{commutativity gadget}. The goal of a commutativity gadget is to constrain the measurements associated to a pair of vertices to commute, without imposing any other constraints. This leads to the following definition:

\begin{definition}
    A commutativity gadget for $k$-colouring is a graph $G$ together with two distinguished vertices $x,y\in G$ such that:
    \begin{itemize}
        \item For any quantum $k$-colouring $\{P_a^u\}_{u,a}$ of $G$, the assignments of $x$ and $y$, the projective measurements $\{P_a^x\}_a$ and $\{Q_b^y\}_b$, commute.
        \item For any two (possibly equal) colours $c,c'$ there exists a quantum $k$-colouring of $G$ such that the colour of $x$ is $c$ and the colour of $y$ is $c'$. 
    \end{itemize}
\end{definition}

The flagship example of a commutativity gadget, due to~\cite{Ji13}, is for $3$-colouring. Here, the gadget is strikingly simple: two copies of the triangle graph $K_3$ with edges connecting corresponding vertices, representable as a triangular prism. The distinguished vertices can be taken as any two elements not connected by an edge. This serves as a commutativity gadget as any quantum three-colouring is commutative and hence classical, but there are many classical three-colourings allowing the distinguished variables to be assigned any colour.

Given any graph $H$ with vertices $r,s\in H$, we can use this commutativity gadget to force the measurements for $r$ and $s$ to commute in any quantum $k$-colouring. To do so, we create a new graph $H'$ by adding the triangular prism to $H$ with the distinguished vertices identified with $r$ and $s$. Then, every quantum $3$-colouring of $H'$ is the extension of a quantum $3$-colouring of $H$ where the measurements corresponding to $r$ and $s$ commute; and every such quantum $3$-colouring of $H$ extends to a quantum $3$-colouring of $H'$. Suppose one is able to classically reduce a system of constraints to a $k$-colouring problem. When considering a quantum version of this reduction, it is necessary that the vertices involved in a single constraint behave classically in isolation; otherwise, new, unintended quantum solutions of this constraint may become possible. By adding commutativity gadgets between these vertices, we can avoid this problem.

One may ask whether a similar construction, and hence a similar technique to show hardness, extends to quantum graph colouring for larger $k$, for example $4$. Surprisingly, the answer is no! This contrasts greatly with the classical scenario, where there is an extraordinarily simple reduction. Add an apex vertex to the graph; the original graph is three-colourable if and only if the new graph is four-colourable. However, for quantum four-colouring, we show that no commutativity gadget, following the definition sketched above, may exist.

The nonexistence of commutativity gadgets for a CSP arises from the symmetries of the underlying relations --- in the case of four-colouring, these are the symmetries of the graph $K_4$. In particular, it is known that there is a non-classical quantum homomorphism $\{\Pi^a_b\}_{a,b\in K_4}\colon K_4\rightarrow K_4$~\cite{Wan98}; in contrast, every quantum homomorphism $K_3\rightarrow K_3$ is classical. Using this noncommutative quantum homomorphism, we can show that $K_4$ does not admit a commutativity gadget.

\begin{theorem}\label{thm:colouring-nogo}
    $K_4$ does not admit a commutativity gadget.
\end{theorem}

\begin{proof}
    Suppose on the contrary that $K_4$ admits a commutativity gadget $G$ with distinguished variables $x,y\in G$. First, since $\{\Pi^a_b\}_{a,b}$ is non-classical, there exist $a_1,b_1,a_2,b_2\in K_4$ such that $\Pi^{a_1}_{b_1}$ and $\Pi^{a_2}_{b_2}$ do not commute. Next, since the relations on the gadget cannot constrain the values $x$ and $y$ can be assigned to, so there must exist a quantum homomorphism $\{P^z_a\}_{z\in G,a\in K_4}:G\rightarrow K_4$ such that $P^x_a=\delta_{a,a_1}$ and $P^y_a=\delta_{a,a_2}$. Consider the composition of $P$ with $\Pi$, which is the quantum homomorphism $G\rightarrow K_4$ defined as
    \begin{align*}
        (\Pi\circ P)^z_a=\sum_{b\in K_4}P^z_b\otimes\Pi^b_a. 
    \end{align*}
    This is in fact a well-defined quantum homomorphism. Furthermore, we find that
    \begin{align*}
        &(\Pi\circ P)^x_{b_1}=\sum_{b\in K_4}P^x_b\otimes\Pi^b_{b_1}=1\otimes\Pi^{a_1}_{b_1},\\
        &(\Pi\circ P)^y_{b_2}=\sum_{b\in K_4}P^y_b\otimes\Pi^b_{b_2}=1\otimes\Pi^{a_2}_{b_2},
    \end{align*}
    meaning that $(\Pi\circ P)^x_{b_1}$ and $(\Pi\circ P)^y_{b_2}$ do not commute. Therefore, there exists a quantum four-colouring of~$G$ that does not constrain the projectors for $x$ and $y$ to be classical. So $G$ cannot be a commutativity gadget for~$K_4$.
\end{proof}

What does this say about the hardness of quantum four-colouring? We know that unlike for classical three-colouring, we cannot use the classical reductions to directly rewrite the constraints of any CSP in terms of it. However, this does not preclude the possibility that entangled four-colouring is undecidable. Nevertheless, since the classical reduction does not quantise, we would need a natively quantum reduction to achieve this.

In the remainder of the paper, we study the notion of a commutativity gadget in more generality, which allows us to broadly extend this result. We also study complementary directions, notably the existence of graph colouring commutativity gadgets in a different model of entangled CSP, and a sufficient criterion for the existence of non-classical homomorphisms. 

\section{Introduction: general constraint satisfaction problems}

\emph{Constraint satisfaction problems (CSPs)} capture some of the most fundamental computational problems, including linear equations, graph colourings, and variants and generalisations of satisfiability.
They are specified by a fixed \emph{template}, which is a relational structure with finite relational signature.
A~\emph{signature} is a tuple of relational symbols each of which has an associated arity,
and a \emph{relational structure}~$A$ in that signature consists of an alphabet~$A$ and an interpretation of each relational symbol in the signature as relations over tuples of length given by the arity, say $R^A\subseteq A^2$ and $S^A\subseteq A^3$. A canonical example of a relational structure is a \emph{graph}, which can be presented as a relational structure with a single symmetric irreflexive relation of arity~$2$. An \emph{instance} of the CSP over $A$, denoted $\CSP(A)$, consists of a finite set of variables and a list of constraints about those variables given by relational symbols from the signature, \textit{e.g.}, $R(x_1,x_2),S(x_4,x_2,x_3),R(x_1,x_4)$.
The goal is to decide whether the variables can be assigned values from $A$ so that all the constraints are simultaneously satisfied.
An equivalent formalization of $\CSP(A)$, which is the one we use in this work, is as the membership problem for the class of all finite relational structures $X$ admitting a \emph{homomorphism} to $A$ --- a map $X\rightarrow A$ that preserves all the relations.

The strongest results are known for finite-alphabet CSPs, \textit{i.e.} the variables of an instance can take only finitely many values.
The first main result in this direction is \emph{Schaefer's dichotomy theorem}~\cite{Sch78} for Boolean CSPs, \textit{i.e.} the variables can take only two values. It classifies the complexity of all boolean CSPs as either in $\tsf{P}$ or $\tsf{NP}$-complete. Many important CSPs are Boolean, for example 3SAT, written over the boolean alphabet $\Z_2=\{0,1\}$, which has eight relations $R_{\mathbf{a}}^{\mathrm{3SAT}}=\Z_2^3\backslash\{\mathbf{a}\}$ for all $\mathbf{a}\in\Z_2^3$, corresponding to the formulae $x\lor y\lor z$, $x\lor y\lor\lnot z$, \textit{etc.}; or LIN, written here over the boolean alphabet $\{+1,-1\}$, which has two relations $R_{\pm}^{\mathrm{LIN}}=\set*{\mathbf{a}\in\{+1,-1\}^3}{a_1a_2a_3=\pm 1}$, corresponding to the linear relations.

If we fix the relational structure to be a simple undirected graph $H$, the problem $\CSP(H)$ becomes equivalent to the problem of deciding the existence of a graph homomorphism $K\to H$.
The \emph{Hell--Ne\v{s}et\v{r}il theorem}~\cite{HN90} gives a dichotomy for this case: testing whether $K\to H$ is solvable in polynomial time if  $H$ is bipartite, otherwise the problem is $\tsf{NP}$-complete.
Decades of research on finite-alphabet CSPs culminated in the landmark result known as the \emph{CSP dichotomy theorem}~\cite{Bul17,Zhu17}, which states that every finite-alphabet CSP is either in $\tsf{P}$ or $\tsf{NP}$-complete. This theorem was proved only recently and unifies both of the aforementioned results.
The dichotomy is driven by algebraic properties of $A$: the problem $\CSP(A)$ is in $\tsf{P}$ if and only if $A$ admits a \emph{weak near unanimity polymorphism}, which is a certain kind of map $A^n\to A$ that is compatible with all the relations of the relational structure $A$.

The origins of entangled CSPs trace back to the same period as classical CSPs.
The Mermin--Peres magic square~\cite{Mer90,Per90a} is a paradigmatic example of a quantum advantage in CSPs. In the language of relational structures, the Mermin--Peres magic square game is the linear constraint system corresponding to homomorphisms $X\rightarrow\mathrm{LIN}$, where the alphabet of the relational structure~$X$ is $X=\{x_1,\dots, x_9\}$ with relations $R_+^{X}=\{(x_1,x_2,x_3), (x_4,x_5,x_6), (x_7,x_8,x_9)\}$ and $R_-^{X}=\{(x_1,x_4,x_7), (x_2,x_5,x_8), (x_3,x_6,x_9)\}$. This can be visualised as having the variables arrayed on a $3\times 3$ grid where the row products are $+1$ and the column products are $-1$. However, the system of equations has no solution over the alphabet $\{-1,+1\}$, \textit{i.e.} the space of homomorphisms $X\rightarrow\mathrm{LIN}$ is empty. However, Mermin showed that the system has an \emph{operator solution} consisting of linear operators on a four-dimensional Hilbert space.
Mermin used this construction to prove the Bell--Kochen--Specker theorem on the impossibility of explaining quantum mechanics via hidden variables.

To implement the operator solution, the Mermin--Peres magic square can be rephrased as a \emph{nonlocal game}. In general, nonlocal games, which are a development of Bell inequalities~\cite{Bel64} from the physics literature, provide a mathematical framework to study quantum entanglement~\cite{CHTW04}. In a nonlocal game, two players usually called Alice and Bob, who may not communicate but may come up with a strategy in advance, try to correctly respond to questions given to each by a referee. There are a variety of (often but not always equivalent) ways to present a CSP as a nonlocal game. In the \emph{constraint-constraint game} of $X\rightarrow A$, both Alice and Bob get an element of a relation of $X$ (an equation) and they must respond with an element from the corresponding relation of $A$ (a satisfying assignment). They win if their assignments agree on the variables which appear in both equations.
Alice and Bob can use a \emph{classical strategy} (deterministic or probabilistic), or a \emph{quantum strategy}, in which the players determine their answers by performing independent measurements on a shared entangled state.
The goal of Alice and Bob is to maximize the winning probability. There is a perfect classical strategy for the game if and only if there is a (classical) homomorphism $X\rightarrow A$. 
The power of entanglement is revealed when the maximum winning probability achievable by a quantum strategy is strictly larger than that of a classical strategy. Another presentation as a nonlocal game relevant to this work is the \emph{constraint-variable game}, where Alice gets a tuple as above but Bob gets a single variable from that tuple, and, to win, the assignment to the variable that Bob responds with must agree with the satisfying assignment that Alice responds with. A final presentation, the \emph{assignment game}, valid only for arity-2 CSPs, is the one where Alice and Bob each receive one of the two variables from an element of a relation and need to respond with assignments that together are satisfying.

Aravind~\cite{Ara04} showed that the operator solution for the Mermin--Peres magic square corresponds to a perfect quantum strategy of the associated constraint-constraint nonlocal game, which has no perfect classical strategy.
In a similar fashion, one can associate a nonlocal game to the CSP of any relational structure $A$, which gives rise to its \emph{entangled CSP}. Cleve and Mittal~\cite{CM14a} showed that perfect quantum strategies for boolean constraint systems correspond to operator assignments in a similar way; and Abramsky, Barbosa, de Silva, and Zapata~\cite{ABdSZ17} extended this to general relational structure homomorphisms.

As for classical CSPs, one can ask about the decision complexity of entangled CSPs. In general, the complexity of deciding if there is a perfect quantum strategy for a nonlocal game is much harder than for classical strategies. In fact, due to Ji, Natarajan, Vidick, Wright, and Yuen~\cite{JNV+21}, it is \emph{undecidable} in the following strong way: there is a polynomial-time reduction from the halting problem to the problem of deciding whether, given a succinct description of a nonlocal game, there is a perfect quantum strategy or every quantum strategy wins with probability less than $\frac{1}{2}$, given that one of the two holds. The \emph{halting problem} is the problem of deciding whether a given Turing machine halts, and is the prototypical example of an undecidable problem~\cite{Tur36}; and \emph{succinctly presented nonlocal games} are those where the referee has efficient probabilistic algorithms to sample questions and decide correctness of answers, but may not be able to efficiently access the table of values for these functions. 

The undecidability of nonlocal games extends, in some cases, to the level of entangled CSPs. Culf and Mastel~\cite{CM24} show that, for any boolean relational structure $A$ such that $\CSP(A)$ is $\tsf{NP}$-complete, there is a polynomial-time reduction from the halting problem to the problem of deciding whether, for a succinctly-presented instance $X$ of $\CSP(A)$, there is a perfect quantum strategy for the game $X\rightarrow A$ (in either the constraint-constraint or constraint-variable model) or whether every quantum strategy wins with probability $<s$ (for some fixed $s<1$ depending on $A$), given that one of the two holds. Conversely, for every boolean relational structure $A$ whose classical CSP is in $\tsf{P}$ except for LIN, Atserias, Kolaitis, and Severini~\cite{AKS17} show that, give a relational structure $X$, the existence of a perfect quantum strategy for the $X\rightarrow A$ game can be decided in polynomial time. Hence, for succinctly presented relational structures, this can be decided in exponential time, which is computable. In the remaining boolean case, linear constraint systems LIN, the situation remains unclear: Slofstra~\cite{Slo19} shows that deciding whether there exists a perfect quantum strategy is undecidable, but the decidability in the gapped case, as studied in~\cite{CM24}, is unknown. Beyond boolean CSPs, even less is known. For example, due to \cite{CM24}, the complexity of entangled $3$-colouring is known to be undecidable in the same way as for boolean CSPs above, in all three of the constraint-constraint, constraint-variable, and assignment models; but nothing is known for entangled $k$-colouring with $k>3$. On the other hand, due to Bulatov and \v{Z}ivn\'y~\cite{BZ25}, the entangled CSP problem for any CSP of bounded width --- classically solvable in polynomial time by the local consistency checking algorithm~\cite{BK09}, such as homomorphism to a bipartite graph --- is also solvable in polynomial time.

The current proof techniques to show undecidability of entangled CSPs are quantisations of the primitive positive reductions (or gadget reductions) used to show $\tsf{NP}$-completeness of classical CSPs. They rely on being able to express constraints in the original constraint systems as conjunctions of constraints in the constraint system being reduced to. To extend this reduction to the entangled setting requires two components. First, we need to be able to split these conjunctions of constraints up into the constituent constraints so that they can be asked separately. This was shown to be generically possible by Mastel and Slofstra~\cite{MS24}. Second, we need the variables in the subconstraints within these conjunctions to commute, in order to permit only classical satisfying assignments at this level. This, however, does not seem to be generically possible. To achieve this commutation, prior work relies on \emph{commutativity gadgets}, which are constraint systems which do not constrain the classical values a pair of variables can take, but force them to commute in any quantum assignment. The idea of a commutativity gadget was introduced by Ji~\cite{Ji13}, who constructed a commutativity gadget for $3$-colouring and $1$-in-$3$-SAT in order to show gapless reductions between entangled CSPs. \cite{CM24} extended this to commutativity gadgets for all boolean CSPs and to gapped reductions.

The aim of this work is to study the existence of commutativity gadgets, and hence the possibility of generically quantising classical CSP reductions. As a starting point, let us consider the case of $k$-colouring. It has been shown in \cite{Ji13} that a commutativity gadget exists if $k=3$. In this work, as presented in~\cref{sec:preface}, we show that it cannot exist for $k\geq 4$ by establishing a link to the quantum permutation group $S_k^+$. It is well known that $S_k^+$ is non-classical for $k\geq 4$ \cite{Wan98}, and we show that these quantum symmetries prevent the existence of commutativity gadgets.

We explore this line of reasoning for arbitrary CSPs. This requires us to quantise the endomorphism monoids of relational structures. We distinguish between \emph{quantum endomorphism monoids} and \emph{oracular quantum endomorphism monoids} (precisely defined in Definition \ref{def:quantum-endomorphism-monoid}), which reflect the difference between quantum automorphism groups as defined by Banica~\cite{Ban05} and Bichon~\cite{Bic03}, respectively; or, in the language of nonlocal games, the difference between strategies for the assignment game, and the constraint-constraint and constraint-variable games, respectively. Using this new framework, we reach the following no-go theorem, extending~\cref{thm:colouring-nogo}:
\begin{theorem*} [See Theorem \ref{thm:nogo}]
    \interlinepenalty=99
    Let $A$ be a relational structure. 
    \begin{enumerate}[(a)]
        \item If $A$ admits a non-classical quantum endomorphism, then no commutativity gadget exists for~$A$.
        \item If $A$ admits a non-classical oracular quantum endomorphism, then no oracular commutativity gadget exists for $A$.
    \end{enumerate}
    \interlinepenalty=0
\end{theorem*}

To rule out the existence of (oracular) commutativity gadgets using this theorem, one needs to be able to understand the appropriate quantum endomorphism monoid
of the relational structure $A$. Quantum automorphism groups have already proven to be objects of interest in the case of graphs \cite{Sch18, Sch20, RS22a, BB07}. Clearly, the no-go theorem implies that any graph with a non-classical quantum automorphism group does not allow for a commutativity gadget. Furthermore, some of the methods used in the analysis of quantum automorphism groups can be translated to quantum endomorphism monoids as well. In particular, by introducing a new technical condition called weak adjacency congruence (WAC), we are able to prove the following criterion, inspired by \cite{Sch20}, for the existence of non-classical quantum endomorphisms for a graph:

\begin{theorem*}[See Theorem \ref{thm:endo-schmidt}]
    Let $G$ be a graph. Suppose that $f$ and $g$ are non-trivial endomorphisms.
    \begin{enumerate}[(a)]
        \item If $f$ and $g$ have disjoint supports and are WAC, then $G$ admits non-classical quantum endomorphisms.
        \item If $f$ and $g$ have disconnected supports, then $G$ admits non-classical oracular quantum endomorphisms.
    \end{enumerate}
\end{theorem*}

Combining the above theorems gives one a tool to conclude the non-existence of commutativity gadgets for graphs based on purely classical properties of these graphs. We use it to find a graph that admits neither a commutativity gadget nor an oracular commutativity gadget. Currently, we are not aware of similar general criteria that can be used to prove that all quantum endomorphisms are classical or a graph admits a commutativity gadget. Instead, we analyse several concrete classes of graphs to establish examples that exhibit these properties.

While we have shown that $k$-colouring does not admit a commutativity gadget, we are able to construct an oracular commutativity gadget for it. This implies that oracular entangled $k$-colouring is undecidable for $k\geq 4$ using some results implicit in~\cite{CM24}: see \cref{sec:comm-gadgets-complexity}. For other kinds of graphs, the odd cycles and the so-called odd graphs, we are able to show that all quantum endomorphisms are classical, by proving that all quantum endomorphisms are in fact quantum automorphisms. We also show that for any graph without a length four cycle, any oracular commutativity gadget is a commutativity gadget.

\paragraph{Outline}

In \cref{sec:prelims}, we give a detailed overview of the background needed for this article, reviewing the necessary concepts from graph theory, algebra, nonlocal games, and complexity theory. In the next section, \cref{sec:quantum-homs}, we develop the concepts of quantum endomorphism monoids and oracular quantum endomorphism monoids, and capture them in a single categorical framework. \cref{sec:comm-gadgets} introduces the different notions of commutativity gadgets and the relations between them, details some of their implications for CSPs, and proves the no-go theorem. Lastly, \cref{sec:examples} investigates graph CSPs using the previously developed tools. In the appendix, we show in \cref{sec:an-attempt} that a natural generalisation of the commutativity gadget for the 3-cycle does not yield a commutativity gadget for larger odd cycles; and in \cref{sec:bipartite} that the entangled CSP of a bipartite graph is in $\tsf{P}$ in an alternate way to~\cite{BZ25}.

\paragraph{Open questions}
While we address many different aspects of commutativity gadgets, this work is the first systematic investigation of commutativity gadgets of arbitrary CSPs, and many questions remain. Here we introduce a few of them, and they will be restated more precisely in appropriate locations in the main text.
\begin{enumerate}
	\item One can define commutativity gadgets to capture approximate homomorphisms in different ways: robust commutativity gadgets satisfy the minimum requirements needed to capture this, while algebraic commutativity gadgets, based on the game algebra, do not but are easier to work with. We know that algebraic commutativity gadgets are automatically robust commutativity gadgets, but we do not know if the reverse implication is true. In other words, are there (robust) commutativity gadgets that are not algebraic commutativity gadgets?
	\item Intuitively, commutativity gadgets impose commutation relations between two PVMs without imposing any relations on the classical assignments to those PVMs. This means that for every possible classical assignment to these two PVMs, there must be a quantum assignment to the commutativity gadget that agrees with it. For every known CSP with a commutativity gadget, this quantum assignment can be chosen to be classical. Are there examples of CSPs that only admit commutativity gadgets where for some classical values of the PVMs no classical assignment exists? A reduction using such a commutativity gadget would not necessarily be a valid reduction for a classical CSP.
	\item We show that oracular entangled $k$-colouring is undecidable for $k\geq 4$. Does the same hold for (non-oracular) entangled $k$-colouring?
	\item Using our criterion for the existence of non-classical quantum endomorphisms, we are able to construct an example of an NP-complete CSP with alphabet size four that does not have an oracular commutativity gadget. When the alphabet size is two for an NP-complete CSP, it is known that an oracular commutativity gadget always exists~\cite{CM24}. However, for alphabet size three this is still an open problem: does there exist an NP-complete CSP with alphabet size three that does not allow for an (oracular) commutativity gadget?
	\item For odd cycles and odd graphs, we have shown that all quantum endomorphisms are classical. Consequently, our no-go theorem does not apply, but this does not automatically give us a construction of a commutativity gadget. A natural, but unsuccessful, candidate of a commutativity gadget can be found in \cref{sec:an-attempt}. Therefore, the question remains: can one construct a commutativity gadget for these graphs?
\end{enumerate}

\paragraph{Related work} There has been a wide range of work studying quantum homomorphisms of relational structures. This idea was first introduced in the context of graphs by Man\v{c}inska and Roberson~\cite{MR16}, via the assignment nonlocal game, and extended to the setting of general relational structures by Abramsky, Barbosa, de Silva, and Zapata~\cite{ABdSZ17}, via the constraint-variable game. Ciardo~\cite{Cia24}, and Bulatov and \v{Z}ivn\'y~\cite{BZ25} studied connections between the complexity of classical CSPs and quantum satisfiability. The connection between quantum isomorphisms of graphs and quantum groups was first investigated by Lupini, Man\v{c}inska, and Roberson~\cite{LMR20}, who considered the non-oracular setting. An extension of oracular quantum automorphism groups to directed hypergraphs --- a different way to extend graphs than to relational structures --- was studied by Faro\ss{}~\cite{Far24}.

The notion of the quantum core of a graph was initially studied by Ortiz and Paulsen~\cite{OP16}. They show that every graph has a minimal idempotent quantum endomorphism, which they term the quantum core of the graph. This is analogous to the classical case, wherein the core is the image of a minimal idempotent graph endomorphism. We make use of the special case where the idempotent quantum endomorphism is the identity, which is equivalent to all quantum endomorphisms being automorphisms, which quantises the idea of a graph being a core (\textit{i.e.} its own core).

The complexity of specific classes of nonlocal games has been an area of recent active study. Man\v{c}inska, Spaas, Spirig, and Vernooij~\cite{MSSV25} have shown that the entangled independent set game, which is classically easy, is undecidable. Also, Taller and Vidick~\cite{TV25} have shown that entangled linear system games are undecidable, but without perfect completeness.

\paragraph{Acknowledgements}%
EC was supported by a CGS D scholarship from Canada's NSERC.
JvDdB was supported by the Carlsberg Foundation Young Researcher Fellowship CF21-0682 ``Quantum Graph Theory''. MV was supported by the NWO Vidi grant VI.Vidi.192.018 `Non-commutative harmonic analysis and rigidity of operator algebras'.
PZ was funded by the European Union (ERC, POCOCOP, 101071674). Views and opinions expressed are however those of the authors only and do not necessarily reflect those of the European Union or the European Research Council Executive Agency. Neither the European Union nor the granting authority can be held responsible for them.

\section{Preliminaries}\label{sec:prelims}

\subsection{Notation}

We denote the cyclic group modulo $n$ by $\Z_n$ and denote its equivalence classes by the standard representatives $\{0,\ldots,n-1\}$. We write $[n]=\{1,\ldots,n\}\subseteq\N$.

Given a set $X$, we write elements of $X^n$ as $\mathbf{x}=(x_1,\ldots,x_n)$. For strings $x\in X^n$, write $|x|=n$ for the length of $x$.

We only work with probability distributions on finite sets, so we can represent a probability distribution $\pi$ on a finite set $X$ by a function $\pi:X\rightarrow[0,1]$ such that $\sum_{x\in X}\pi(x)=1$. The \emph{uniform distribution} on a finite set $X$, defined as $x\mapsto\frac{1}{|X|}$, is denoted $\mbb{u}_X$. The subscript is dropped when the set $X$ is clear from context.

For a Hilbert space $H$, write $B(H)$ for the algebra of bounded linear operators $H\rightarrow H$. A \emph{positive operator-valued measure (POVM)} is a finite set of positive operators $\{P_i\}_{i\in I}$ in $B(H)$ such that $\sum_{i\in I}P_i=1$. A \emph{projection} is an element $P\in B(H)$ such that $P=P^\ast=P^2$. A \emph{projection-valued measure (PVM)} is a POVM $\{P_i\}_{i\in I}$ such that $P_i$ is a projection for all $i$.

In Hilbert spaces, we use bra-ket notation. In $\C^2$, we denote a canonical orthonormal basis, called the \emph{computational basis}, by $\{\ket{0},\ket{1}\}$. We write the \emph{Hadamard basis} $\{\ket{+},\ket{-}\}$, where $\ket{+}=\frac{1}{\sqrt{2}}\parens*{\ket{0}+\ket{1}}$ and $\ket{-}=\frac{1}{\sqrt{2}}\parens*{\ket{0}-\ket{1}}$.

We write the \emph{commutator} of two elements $x,y$ in an algebra as $[x,y]=xy-yx$.

\subsection{Graphs}

We work with finite simple undirected loop-free graphs. In this case a \emph{graph} $G=(V(G),E(G))$ where $V(G)$ is a finite set, called the \emph{vertices}, and $E(G)\subseteq\set*{\{x,y\}}{x,y\in V(G)\land x\neq y}$, called the \emph{edges}. The edges induce a symmetric irreflexive relation on $V(G)$ , which we denote by $x\sim_G y$ iff $\{x,y\}\in E(G)$. We often write $x\in G$ to mean $x\in V(G)$. $G$ is \emph{bipartite} if there exists a partition $V(G)=A\cup B$ such that $x\nsim_Gy$ if $x,y\in A$ or $x,y\in B$.

A \emph{walk} of length $n$ in $G$ is a tuple $\mathbf{x}=(x_1,\ldots,x_{n+1})\in V(G)^{n+1}$ such that $x_i\sim_G x_{i+1}$ for all $i\in[n]$. We say $\mathbf{x}$ is a walk from $x$ to $y$ if $x_1=x$ and $x_{n+1}=y$. The \emph{distance} between $x$ and $y$ is the length of the shortest walk from $x$ to $y$; this is denoted $d(x,y)$ and it is a metric on $G$. The \emph{diameter} of $G$ is the maximum distance between two vertices. We say $\mathbf{x}$ is \emph{cyclic} if $x_1=x_{n+1}$. We say $\mathbf{x}$ is a \emph{path} if $x_i\neq x_j$ for all $i\neq j$; a \emph{cyclic path}, or \emph{cycle}, is a walk of length $n>2$ where the above holds except that $x_1=x_{n+1}$. The \emph{girth} of a graph is the length of the shortest cyclic path. We also use the \emph{even and odd girths} which are with respect to paths of even and odd length, respectively; and the \emph{walk girth}, which is the girth with respect to cyclic walks of length $>0$. The \emph{odd walk girth} is the girth with respect to cyclic walks of odd length; however, a notion of even walk girth is not sensible, as this would be $2$ for every graph that has an edge.

The \emph{complement} of a graph $G$ is the graph $\overline{G}$ with the same vertices and edges $x\sim_{\overline{G}}y$ if and only if $x\neq y$ and $x\nsim_G y$.

A \emph{graph product} is a way to define a graph on the vertex set $V(G)\times V(H)$ for graphs $G$ and $H$. We work with the \emph{box product} (or Cartesian product), denoted $G\Boxprod H$ with edges $(x,y)\sim_{G\Boxprod H}(x',y')$ iff $x=x'$ and $y\sim_H y'$ or $x\sim_G x'$ and $y=y'$; and with the \emph{categorical product} (or tensor product), denoted $G\times H$ with edges $(x,y)\sim_{G\times H}(x',y')$ iff $x\sim_G x'$ and $y\sim_H y'$.

We work with some standard examples of graphs. The \emph{complete graph} on $n$ vertices is the graph $K_n$ with $V(K_n)=\Z_n$ and $x\sim_{K_n} y$ for all $x\neq y$. The \emph{$n$-cycle} is the graph $C_n$ with $V(C_n)=\Z_n$ and $x\sim_{C_n}y$ iff $x=y\pm 1$. The \emph{path} of length $n$ is the graph $P_n$ with $V(P_n)=[n+1]$ and $x\sim_{P_n}y$ iff $x=y\pm 1$. For $n\geq k$, the \emph{Kneser graph} is the graph $K(n,k)$ with $V(K(n,k))=\set*{u\subseteq[n]}{\abs{u}=k}$ and $u\sim_{K(n,k)}v$ iff $u\cap v=\varnothing$. Some special graphs are Kneser graphs: $K_n=K(n,1)$, the \emph{Petersen graph} is $P=K(5,2)$, and the \emph{odd graphs} are $O_n=K(2n-1,n-1)$. The odd graph $O_n$ has odd girth $2n-1$.

\subsection{Category theory}

A \emph{category} $\mc{C}$ consists of a class of \emph{objects} $\ob(\mc{C})$ and for every pair $a,b\in\ob(\mc{C})$, a class of \emph{morphisms} $\mor_{\mc{C}}(a,b)$, where the morphisms have a composition operation $\circ$ such that for all $f\in\mor_{\mc{C}}(a,b)$ and $g\in\mor_{\mc{C}}(b,c)$, $g\circ f\in\mor_{\mc{C}}(a,c)$, satisfying associativity $h\circ(g\circ f)=(h\circ g)\circ f$ for all $f\in\mor_{\mc{C}}(a,b)$,$g\in\mor_{\mc{C}}(b,c)$, and $h\in\mor_{\mc{C}}(c,d)$ and existence of an identity $\id_a\in\mor_{\mc{C}}(a,a)$ such that $f\circ\id_a=\id_b\circ f=f$ for all $f\in\mor_{\mc{C}}(a,b)$. Where clear from context, we write $a\in\mc{C}$ to mean $a\in\ob(\mc{C})$, $\mor(a,b)$ to mean $\mor_{\mc{C}}(a,b)$, and $f:a\rightarrow b$ to mean $f\in\mor_{\mc{C}}(a,b)$.

An \emph{endomorphism} is a morphism $f:a\rightarrow a$, and we denote the set of endomorphisms by $\enmo(a)=\mor(a,a)$. An \emph{isomorphism} is a morphism $f:a\rightarrow b$ with an inverse morphism $g:b\rightarrow a$ such that $f\circ g=\id_b$ and $g\circ f=\id_a$, and we denote the set of isomorphisms by $\iso(a,b)$. An \emph{automorphism} is an isomorphism from an object to itself, and we write $\aut(a)=\iso(a,a)$.

A \emph{(covariant) functor} $F$ from a category $\mc{C}$ to a category $\mc{D}$ is a map, denoted $F:\mc{C}\rightarrow\mc{D}$ that to each object $a\in\mc{C}$ assigns an object $F(a)\in\mc{C}$, and to each morphism $f:a\rightarrow b$ assigns a morphism $F(f):F(a)\rightarrow F(b)$, such that $F(\id_a)=\id_{F(a)}$ and $F(g\circ f)=F(g)\circ F(f)$. Two categories $\mc{C}$ and $\mc{D}$ are \emph{isomorphic}, denoted $\mc{C}\cong\mc{D}$, if there exist functor $F:\mc{C}\rightarrow\mc{D}$ and $G:\mc{D}\rightarrow\mc{C}$ such that $F\circ G=\Id_{\mc{D}}$ and $G\circ F=\Id_{\mc{C}}$, where $\Id$ denotes the identity functor.

\subsection{Relational structures}\label{sec:defs-rel-structs}
A \emph{relation} over a set $A$ is a subset $R\subseteq A^n$ for some $n\in\N$, called the \emph{arity} of $R$ and denoted~$\ar(R)$. A \emph{relational structure} is a pair $(A,\mc{R})$, where $A$ is a finite set, called the \emph{alphabet}, and $\mc{R}$ is a collection of relations over $A$. A \emph{signature} is a set $\sigma$ such that each $R\in\sigma$ is associated an arity $\ar(R)\in\N$. The elements of $\sigma$ are called \emph{relational symbols}. A \emph{relational structure over~$\sigma$} is a relational structure $(A,\sigma^A)$ with relations $R^A\subseteq A^{\ar(R)}$ for all $R\in\sigma$. Any relational structure $(A,\mc{R})$ can be expressed as a relational structure over a signature, simply by assigning one relational symbol to each element of $\mc{R}$. When the signature or set of relations is clear, we suppress reference to it and refer to the relational structure by its alphabet only.

Let $A$ and $B$ be relational structures over the same signature $\sigma$. Then a \emph{homomorphism} from $A$ to $B$ is a function $f:A\rightarrow B$ such that for all $R\in\sigma$ and $\mathbf{a}\in R^A$, $$f(\mathbf{a})\coloneqq(f(a_1),\ldots,f(a_{\ar(R)}))\in R^B.$$

We denote the \emph{category of relational structures} over a signature $\sigma$ by $\mc{R}(\sigma)$. Here the objects are the relational structures over $\sigma$ and the morphisms $\mor(A,B)$ are the homomorphisms~${f:A\rightarrow B}$.

Graphs provide a standard example of relational structures. Given a graph $G$, it can be presented as the relational structure over $\sigma=\{E\}$ with arity $\ar(E)=2$, with alphabet $V(G)$ and relation $E^G=\set*{(x,y)}{x\sim_G y}$. We implicitly treat graphs as relational structures using this presentation.

Relational structures are useful to encode constraint satisfaction problems. An alternate way to express this other than relational structure homomorphisms is in terms of constraint systems, and there is a bijective correspondence between these two pictures. A \emph{constraint system} $S$ over an alphabet $\Sigma$ consists of a set of variables $X$ and constraints $(\mathbf{v}_i,C_i)$ for $i\in[m]$, where $\mathbf{v}_i\in X^{n_i}$ and $C_i\subseteq \Sigma^{n_i}$. A \emph{satisfying assignment} to $S$ is a function $f:X\rightarrow\Sigma$ such that $f(\mathbf{v}_i)\in C_i$ for all $i$. $S$ induces a pair of relational structures over the same signature as follows. Let $\sigma$ be the partition of $[m]$ given by the equivalence relation $i\sim j$ if $C_i=C_j$. Then, let $A$ be the relational structure over $\sigma$ with alphabet $\Sigma$ and relations $R^A=C_i$ for any $i\in R$; and let $B$ be the relational structure with alphabet $X$ and relations $R^B=\set*{\mathbf{v}_i}{i\in R}$. Then, it follows directly that this correspondence between constraint systems and pairs of relational structures over the same signature is bijective, and that $f$ is a satisfying assignment to $S$ if and only if it is a homomorphism $B\rightarrow A$.

\subsection{\texorpdfstring{$\ast$}{\textasteriskcentered{}}-algebras and \texorpdfstring{$C^\ast$}{\textit{C}*}-algebras}

A unital \emph{$\ast$-algebra} is an algebra $\mc{A}$ over $\C$ (or $\R$) equipped with an antilinear involution $x\mapsto x^\ast$ such that $1^\ast=1$ and $(xy)^\ast=y^\ast x^\ast$. Denote the (real) subspace of hermitian elements by $\mc{A}_h=\set*{x\in\mc{A}}{x^\ast=x}$. A \emph{$\ast$-positive cone} over $\mc{A}$ is a set $\mc{A}_+\subseteq\mc{A}_h$ such that $1\in\mc{A}_+$, $x+y\in\mc{A}_+$ for all $x,y\in\mc{A}_+$, and $axa^\ast\in\mc{A}_+$ for all $x\in\mc{A}_+$ and $a\in\mc{A}$. A $\ast$-positive cone induces an order on $\mc{A}_h$ via $x\geq y$ iff $x-y\in\mc{A}_+$. We say $\mc{A}_+$ is \emph{archimedean} if for all $x\in\mc{A}$, there exists $R\in\R$ such that $x^\ast x\leq R1$; a $\ast$-algebra equipped with an archimedean $\ast$-positive cone is called a \emph{semi-pre-$C^\ast$-algebra}, following the notation of~\cite{Oza13}. A standard choice of $\ast$-positive cone is the \emph{sum-of-squares cone} given by $\mc{A}_+=\set*{\sum_{i=1}^nx_i^\ast x_i}{n\in\N,\;x_i\in\mc{A}}$. Given a set of generators $S$ and relations $R$, we denote the \emph{$\ast$-algebra generated $S$ subject to $R$} as $\gen*{S}{R}=F(S)/\ang*{R}$, where $F(S)$ is the \emph{free algebra} with generators $S$ and $\ang*{R}$ is the two-sided ideal of $F(S)$ generated by $R$.

We say a semi-pre-$C^\ast$-algebra is \emph{hereditary} if $\mc{A}_+\cap(-\mc{A}_+)=\{0\}$. These arise via quotients by hereditary ideals: see~\cite{HMPS19}. Importantly, if $\mc{A}$ is hereditary then when $x_i\in\mc{A}_+$ and $\sum_ix_i=0$, we must have $x_i=0$ for all $i$.

A semi-pre-$C^\ast$-algebra is a \emph{pre-$C^\ast$-algebra} if $x^\ast x\leq\varepsilon 1$ for all $\varepsilon>0$ implies that $x=0$ in $\mc{A}$. Every pre-$C^\ast$-algebra is hereditary. If $\mc{A}$ is a pre-$C^\ast$-algebra, we can define a norm on it by
\begin{align*}
    \norm{x}=\sup\set*{\norm{\pi(x)}}{\pi:\mc{A}\rightarrow B(H)\text{ is a $\ast$-homomorphism}}=\inf\set*{\lambda\in\R_{\geq0}}{x^\ast x\leq\lambda^21}.
\end{align*}
This norm satisfies the properties $\norm{xy}\leq\norm{x}\norm{y}$ and $\norm{x^\ast x}=\norm{xx^\ast}=\norm{x}^2$. A \emph{$C^\ast$-algebra} is a pre-$C^\ast$-algebra that is complete with respect to this norm. In a $C^\ast$-algebra $\mc{A}$, there is a unique $\ast$-positive cone given by $\mc{A}_+=\set*{x^\ast x}{x\in\mc{A}}$.

Given a semi-pre-$C^\ast$-algebra $\mc{A}$, we can convert it to a pre-$C^\ast$-algebra by taking the quotient by the infinitesimal ideal $I(\mc{A})=\set*{x\in\mc{A}}{x^\ast x\leq\varepsilon 1\;\forall\,\varepsilon\geq 0}$. We call the completion of $\mc{A}/I(\mc{A})$ the \emph{universal $C^\ast$-algebra} of $\mc{A}$, and denote it $C_u^\ast(\mc{A})$. Given a set of generators $S$ and a set of relations $R$, the \emph{$C^\ast$-algebra generated by $S$ subject to $R$} is denoted $C^\ast\!\gen*{S}{R}=C^\ast_u(\gen{S}{R})$.

A \emph{von Neumann algebra} is a unital $\ast$-subalgebra $\mc{M}\subseteq B(H)$ for a Hilbert space $H$ that is closed under the \emph{weak operator topology}, defined via convergence as $(x_\lambda)\rightarrow x$ iff $\braket{\psi}{x_\lambda}{\phi}\rightarrow\braket{\psi}{x}{\phi}$ for all $\ket{\psi},\ket{\phi}\in H$. A  von Neumann algebra is \emph{finite} if there exists a faithful tracial state (see below) on $\mc{M}$. A finite von Neumann algebra is \emph{Connes-embeddable} if there is an injective $\ast$-homomorphism $\mc{M}\rightarrow\mc{R}^\omega$, where $\mc{R}^\omega$ is an ultrapower of the hyperfinite $\mathrm{II}_1$ factor.

A \emph{state} on a semi-pre-$C^\ast$-algebra $\mc{A}$ is a linear map $\rho:\mc{A}\rightarrow\C$ such that $\rho(1)=1$ and $\rho(x)\geq 0$ for all $x\in\mc{A}_+$. We say $\rho$ is \emph{tracial} if $\rho(xy)=\rho(yx)$ for all $x,y\in\mc{A}$; and we say $\rho$ is \emph{faithful} if $\rho(x^\ast x)=0$ implies $x=0$ for all $x\in\mc{A}_+$ The existence of a faithful state automatically guarantees that $\mc{A}$ is a pre-$C^*$-algebra, since the infinitesimal ideal is trivial. Tracial states are usually denoted by $\tau$. Any state induces a seminorm $\norm{a}_\rho=\rho(a^\ast a)$ called the \emph{$\rho$-norm}. Every state $\rho$ induces a \emph{GNS representation}, which is a $\ast$-representation $\pi:\mc{A}\rightarrow B(H)$ such that there is a unit vector $\ket{\psi}\in H$ satisfying $\rho(x)=\braket{\psi}{\pi(x)}{\psi}$. A state is called \emph{finite-dimensional} if it admits a finite-dimensional GNS representation; a state is called \emph{Connes-embeddable} if it admits a GNS representation where the von Neumann algebra generated by $\pi(\mc{A})$ is Connes-embeddable. A \emph{character} is a state that is a $\ast$-representation.

The \emph{tensor product} of $\ast$-algebras $\mc{A}$ and $\mc{B}$ is denoted $\mc{A}\otimes\mc{B}$, or sometimes $\mc{A}\odot\mc{B}$ in order to emphasize that we are not taking a completion. We implicitly use the isomorphism $\C\otimes\mc{A}\cong\mc{A}\otimes\C\cong\mc{A}$. For $C^\ast$-algebras $\mc{A}$ and $\mc{B}$, there are a variety of (generally inequivalent) ways to complete $\mc{A}\odot\mc{B}$ to a $C^\ast$-algebra. Notably, we make use of the \emph{minimal tensor product} (or min-tensor product), which is the completion $\mc{A}\otimes_{\min}\mc{B}$ of $\mc{A}\odot\mc{B}$ with respect to the norm $\norm{x}_{\min}=\norm{(\pi_A\otimes\pi_B)(x)}$ for $\pi_A$ and $\pi_B$ faithful $\ast$-representations of $\mc{A}$ and $\mc{B}$, respectively; and the \emph{maximal tensor product} (or max-tensor product), which is the completion $\mc{A}\otimes_{\max}\mc{B}$ of $\mc{A}\odot\mc{B}$ with respect to the norm $\norm{x}_{\max}=\sup\set*{\norm{\pi(x)}}{\pi\text{ $\ast$-representation of }\mc{A}\odot\mc{B}}$. For von Neumann algebras $\mc{M}\subseteq B(H)$ and $\mc{N}\subseteq B(K)$, however, there is a unique way to complete $\mc{M}\odot\mc{N}$ to a von Neumann algebra: $\mc{M}\overline{\otimes}\mc{N}$ is the closure of $\mc{M}\odot\mc{N}$ in $B(H\otimes K)$.

\subsection{Quantum spaces and quantum groups}

By Gelfand duality, compact spaces are in bijection with commutative $C^\ast$-algebras, via the map $X\mapsto C(X)$, where $C(X)$ is the $C^\ast$-algebra of continuous functions $X\rightarrow\C$ with the norm $\norm{f}=\max\set*{|f(x)|}{x\in X}$. By analogy, we can see a noncommutative $C^\ast$-algebra as the algebra of continuous functions on a virtual \emph{compact quantum space} $X$, and denote it $C(X)$. A compact quantum space $Y$ is a \emph{quantum subspace} of $X$ if there is a surjective $\ast$-homomorphism $C(X)\rightarrow C(Y)$. A natural way to construct a quantum subspace is with respect to a set of relations $R\subseteq C(X)$: then, the surjection is simply the quotient by the two-sided ideal generated by the relations.

A \emph{compact quantum group (CQG)} is a quantum space $G$ equipped with a $\ast$-homomorphism $\Delta:C(G)\rightarrow C(G)\otimes_{\min}C(G)$, called the \emph{coproduct}, satisfying the relation $(\id\otimes\Delta)\Delta=(\Delta\otimes\id)\Delta$, called \emph{coassociativity}, and such that $(\id\otimes C(G))\Delta(C(G))$ and $(C(G)\otimes\id)\Delta(C(G))$ are dense in $C(G)\otimes_{\min}C(G)$. When working with compact quantum groups, we assume all tensor products of $C^\ast$-algebras are min-tensor-products, and suppress the subscript. A \emph{quantum subgroup} $H$ of a CQG $G$ is a CQG that is a quantum subspace and $\Delta_G\pi=(\pi\otimes\pi)\Delta_H$ where $\pi$ is the quotient map.

In this work, we need only compact matrix quantum groups. A \emph{compact matrix quantum group (CMQG)} is a CQG where $C(G)$ is generated by generators $p_{ij}$ for $1\leq i,j\leq n$ and $\Delta(p_{ij})=\sum_{k=1}^np_{ik}\otimes p_{kj}$. Conversely, if $\Delta$ defined in this way is a $\ast$-homomorphism and the matrices $u=(p_{ij})$ and $u^T=(p_{ji})$ are invertible, then it automatically satisfies the two conditions on the coproduct of a CQG, inducing a CMQG structure. Further, we can define a \emph{counit} $\ast$-homomorphism $\epsilon:C(G)\rightarrow \C$ by $\epsilon(p_{ij})=\delta_{ij}$ and a \emph{antipode} $\ast$-anti-automorphism $S:C(G)\rightarrow C(G)$ by $S(p_{ij})=p_{ji}^\ast$. These satisfy the conditions of a \emph{Hopf algebra}, \textit{i.e.} $(\id\otimes\epsilon)\Delta=(\epsilon\otimes\id)\Delta=\id$ and $\mu(\id\otimes S)\Delta=\mu(S\otimes\id)\Delta=1\epsilon$, where $\mu:C(G)\otimes C(G)\rightarrow C(G)$ is the product $\ast$-homomorphism $\mu(x\otimes y)=xy$.

The \emph{quantum permutation group} of $n$ elements is the CMQG $S_n^+$ with algebra of functions generated by $p_{ij}$ with $1\leq i,j\leq n$ subject to the relations $p_{ij}^2=p_{ij}^\ast=p_{ij}$ for all $i,j$, and $\sum_{j}p_{ij}=\sum_j p_{ji}=1$ for all $i$. In the same way, we can define the quantum permutation group of any finite set $X$ as $S_{X}^+\cong S_{|X|}^+$ Given a graph $G$, the \emph{quantum automorphism group} of $G$ is the quantum subgroup $\aut^+(G)$ of $S_G^+$ with respect to the relations $p_{uv}p_{u'v'}=0$ when $u\sim_G u'\land v\nsim_G v'$ or $u\nsim_G u'\land v\sim_G v'$. As a shorthand, write $\Aut^+(G)=C(\aut^+(G))$. An alternate version of the quantum automorphism group is the \emph{oracular quantum automorphism group}, which is the quantum subgroup $\aut^{o+}(G)$ of $\aut^+(G)$ by the relations $[p_{uv},p_{u'v'}]=0$ when $u\sim_G u'$ or $v\sim_G v'$. We similarly write $$\Aut^{o+}(G)=C(\aut^{o+}(G))$$

Prior work has used varied notation for quantum automorphism groups. For example, $\aut^+(G)$ has been denoted $\mathrm{Qut}(G)$~\cite{LMR20,RS22a,DKRSZ25,ADKRZ25}, $S_{|G|}^G$~\cite{SW19}, $\mathrm{QAut}(G)$~\cite{JSW20}, $G_{aut}^+(G)$~\cite{Sch20}, and $\mathbb{Q}\mathrm{ut}(G)$~\cite{BGMSW25}; $\aut^{o+}(G)$ has been denoted $G_{aut}^{\ast}(G)$~\cite{Sch20} and $S_{|G|}^{\mc{E}}$~\cite{SW19} where $\mc{E}$ is the adjacency matrix of $G$; $\Aut^+(G)$ has been denoted $H(G)$~\cite{Ban05} and $A(G)$~\cite{BB07}; and $\Aut^{o+}(G)$ has been denoted $A_{aut}(G)$~\cite{Bic03}.

\subsection{Nonlocal games}

A $2$-player \emph{nonlocal game} $\ttt{G}$ consists of finite sets $X,Y,A,B$ called Alice's questions, Bob's questions, Alice's answers, and Bob's answers; a probability distribution $\pi:X\times Y\rightarrow[0,1]$, called the question distribution; and a function $V:A\times B\times X\times Y\rightarrow\{0,1\}$, $(a,b,x,y)\mapsto V(a,b|x,y)$, called the predicate. A strategy for a nonlocal game $\ttt{G}$ is given by a \emph{correlation}, which is a function $p:A\times B\times X\times Y$, $(a,b,x,y)\mapsto p(a,b|x,y)$, such that $(a,b)\mapsto p(a,b|x,y)$ is a probability distribution for all $x,y$. The \emph{value} of a nonlocal game $\ttt{G}$ with respect to a correlation $p$ is
$$\mfk{v}(\ttt{G},p)=\sum_{(x,y,a,b)\in X\times Y\times A\times B}\pi(x,y)V(a,b|x,y)p(a,b|x,y).$$
We say a correlation $p$ is \emph{perfect} for $\ttt{G}$ if $\mfk{v}(\ttt{G},p)=1$.

A correlation is \emph{deterministic} if there are functions $f:X\rightarrow A$, $g:Y\rightarrow B$ such that $p(a,b|x,y)=\delta_{f(x),a}\delta_{g(y),b}$. A correlation is \emph{classical} if there exists a random variable $\Lambda$ on a set $L$ and functions $f:X\times L\rightarrow A$ and $g:Y\times L\rightarrow B$ such that $p(a,b|x,y)=\expec\delta_{f(x,\Lambda),a}\delta_{f(y,\Lambda),b}$. A correlation is \emph{quantum} if there exist finite-dimensional Hilbert spaces $H_A$ and $H_B$, POVMs $\{P^x_a\}_{a\in A}\subseteq B(H_A)$ $\{Q^y_b\}_{b\in B}\subseteq B(H_B)$ for all $x\in X$ and $y\in Y$, and a state $\ket{\psi}\in H_A\otimes H_B$ such that $p(a,b|x,y)=\braket{\psi}{P^x_a\otimes Q^y_b}{\psi}$. A correlation is \emph{quantum approximate} if it is the limit of a sequence of quantum correlations. A correlation is \emph{quantum commuting} if there exists a Hilbert spaces $H$, POVMs $\{P^x_a\}_{a\in A}\subseteq  B(H)$ $\{Q^y_b\}_{b\in B}\subseteq  B(H)$ for all $x\in X$ and $y\in Y$, and a state $\ket{\psi}\in H$ such that $[P^x_a,Q^y_b]=0$ and $p(a,b|x,y)=\braket{\psi}{P^x_aQ^y_b}{\psi}$. A correlation is \emph{tracial} if there exists a von Neumann algebra $\mc{M}$, a POVM $\{P^x_a\}_{a\in A\cup B}\subseteq\mc{M}$ for all $x\in X\cup Y$, and a tracial state $\tau:\mc{M}\rightarrow\C$ such that $p(a,b|x,y)=\tau(P^x_aP^y_b)$. We denote sets of correlations by $C_Q(A,B|X,Y)$, where the symbol $Q=d$ for deterministic, $Q=c$ for classical, $Q=q$ for quantum, $Q=qa$ for quantum approximate, $Q=qc$ for quantum commuting, and $Q=t$ for tracial. We do not specify the question and answer sets when clear from context. Write also $C_{t,Q}(A,B|X,Y)=C_t(A,B|X,Y)\cap C_Q(A,B|X,Y)$. In this case, we have that $\mc{M}=\C$ if $p$ is deterministic, $\mc{M}$ is commutative if $p$ is classical, $\mc{M}$ is finite-dimensional if $p$ is quantum, $\mc{M}$ is Connes-embeddable if $p$ is quantum approximate, and $C_{t,qc}(A,B|X,Y)=C_t(A,B|X,Y)$. The \emph{value} of $\ttt{G}$ with respect to a set of correlations $C_Q$ is $\mfk{v}_Q(\ttt{G})=\sup_{p\in C_Q(A,B|X,Y)}\mfk{v}(\ttt{G},p)$. Note that $\mfk{v}_d(\ttt{G})=\mfk{v}_c(\ttt{G})$, $\mfk{v}_{t,d}(\ttt{G})=\mfk{v}_{t,c}(\ttt{G})$, $\mfk{v}_q(\ttt{G})=\mfk{v}_{qa}(\ttt{G})$, and $\mfk{v}_{t,q}(\ttt{G})=\mfk{v}_{t,qa}(\ttt{G})$.

A nonlocal game is \emph{synchronous} if $X=Y$, $A=B$, and $V(a,b|x,x)=0$ whenever $a\neq b$. Then, if $p$ is perfect for the game, it must also be synchronous. This implies that $p$ is tracial~\cite{PSSTW16}. Further, near-perfect quantum correlations can be well-approximated by tracial correlations~\cite{Vid22,MdlS23}, in the sense that if $\mfk{v}_q(\ttt{G})\geq 1-\varepsilon$, then $\mfk{v}_{t,q}(\ttt{G})\geq1-O((\varepsilon/C)^{1/4})$, under the assumption that the question distribution of $\ttt{G}$ satisfies $\pi(x,y)=\pi(y,x)$ and $\pi(x,x)\geq C\sum_{y\neq x}\pi(x,y)$; the same holds for quantum approximate and quantum commuting correlations. As such, we can always restrict to working with tracial correlations by restricting the games to synchronous games, which may be done by symmetrising the probability distribution and adding consistency checks.

There are several ways to express constraint satisfaction problems as nonlocal games. Let $A$ and $B$ be relational structures over the same signature $\sigma$. For a probability distribution $\pi:(\cup_{R\in\sigma}\{R\}\times R^A)^2\rightarrow[0,1]$, the \emph{constraint-constraint game} $\ttt{G}_{c-c}(A,B,\pi)$ is the synchronous nonlocal game with question sets $\bigcup_{R\in\sigma}\{R\}\times R^A$, answer sets $\bigcup_{R\in\sigma}B^{\ar(R)}$, question distribution $\pi$, and predicate $V(\mathbf{a},\mathbf{b}|(R,\mathbf{x}),(S,\mathbf{y}))=1$ iff $\mathbf{a}\in R^B$, $\mathbf{b}\in S^B$, and $a_i=b_j$ whenever $x_i=y_j$. For a probability distribution $\pi:\cup_{R\in\sigma}\{R\}\times R^A\rightarrow[0,1]$, the \emph{constraint-variable game} $\ttt{G}_{c-v}(A,B,\pi)$ is the nonlocal game with question sets $\bigcup_{R\in\sigma}\{R\}\times R^A$ and $A$, answer sets $\bigcup_{R\in\sigma}B^{\ar(R)}$ and $B$, question distribution $((R,\mathbf{x}),y)\mapsto\sum_{i:\,x_i=y}\frac{\pi(R,\mathbf{x})}{\ar(R)}$, and predicate $V(\mathbf{a},b|(R,\mathbf{x}),y)=1$ iff $\mathbf{a}\in R^B$ and $a_i=b$ when $x_i=y$. A third way of presenting the constraint problem as a nonlocal game arises in the $2$-CSP case, \textit{i.e.} $\ar(R)=2$ for all $R\in\sigma$. For a probability distribution $\pi:\cup_{R\in\sigma}\{R\}\times R^A\rightarrow[0,1]$, the \emph{assignment nonlocal game} $\ttt{G}_a(A,B,\pi)$ is the nonlocal game with question sets $A$, answer sets $B$, question distribution $(x,y)\mapsto\sum_{R\in\sigma:\,(x,y)\in R^A}\pi(R,(x,y))$, and predicate $V(a,b|x,y)=1$ iff $(a,b)\in R^B$ for all $R\in\sigma$ such that $(x,y)\in R^A$. If the probability distribution has full support, perfect deterministic tracial correlations for all three games correspond to relational structure homomorphisms $A\rightarrow B$. Further, there is a perfect classical correlation iff there is a homomorphism $A\rightarrow B$.

A tracial correlation is \emph{oracularisable} if $[P^x_a,P^y_b]=0$ whenever $\pi(x,y)>0$~\cite{JNV+21}. The \emph{oracularisation} of a nonlocal game is the game where Alice is asked both her and Bob's questions and Bob is asked one of the two questions; the players win if Alice responds with a correct answer and Bob responds with answer consistent with Alice's. This presentation of the oracularisation renders it as a constraint-variable game, but there are other equivalent presentations, for example as a constraint-constraint game~\cite{MS24}. The oracularisation transformation of a game is always sound, and it is also complete (with respect to tracial correlations) if the optimal correlation was oracularisable.

\subsection{Weighted algebras}

A \emph{weighted algebra} is a pair $(\mc{A},\mu)$ where $\mc{A}$ is a $\ast$-algebra and $\mu:\mc{A}\rightarrow\R_{\geq 0}$ is a finitely-supported function, called the \emph{weight function}. The \emph{defect polynomial} of a weighted algebra is
\begin{align*}
    D(\mc{A},\mu)=\sum_{a\in\mc{A}}\mu(a)a^\ast a.
\end{align*}
Given a tracial state $\tau:\mc{A}\rightarrow\C$, the \emph{defect} of $\tau$ is $$\defect(\tau;\mu)=\tau(D(\mc{A},\mu))=\sum_{a\in\mc{A}}\mu(a)\norm{a}_\tau^2.$$
We denote it $\defect(\tau)$ if the weight function is clear from context. Weighted algebras capture approximate representations of $\mc{A}/\ang{\supp(\mu)}$ and $\defect(\tau)$ captures how much the GNS representation of the tracial state $\tau$ deviates from an exact representation of this algebra.

Since we consider only the representations of weighted algebras induced by tracial states, the natural notion of order is weaker than that induced by the sums of squares. We say $a,b\in\mc{A}$ are \emph{cyclically equivalent} if there exist $c_1,d_1,\ldots,c_k,d_k\in\mc{A}$ such that $a-b=\sum_i[c_i,d_i]$; we denote this by $a\cyceq b$. Write $a\gtrsim b$ if there exist $s_1,\ldots,s_k\in\mc{A}$ such that $a-b\cyceq\sum_is_i^\ast s_i$.

A \emph{$C$-homomorphism} $\alpha:(\mc{A},\mu)\rightarrow(\mc{B},\nu)$ is a $\ast$-homomorphism $\alpha:\mc{A}\rightarrow\mc{B}$ such that $\alpha(D(\mc{A},\mu))\lesssim C\cdot D(\mc{B},\nu)$. If there is a $C$-homomorphism $\alpha:(\mc{A},\mu)\rightarrow(\mc{B},\nu)$, then for every tracial state $\tau$ on $\mc{B}$, there exists a tracial state $\tau'=\tau\circ\alpha$ on $\mc{A}$ such that $\defect(\tau')\leq C\defect(\tau)$.

The following inequality is often useful in the context of weighted algebras: for $a_1,\ldots,a_k\in\mc{A}$, $\omega$ a primitive $k$-th root of unity, and writing $\abs*{a}^2=a^\ast a$,
\begin{align*}
    \abs[\Big]{\sum_{i=1}^ka_i}^2\leq\sum_{j=0}^{k-1}\abs[\Big]{\sum_{i=1}^k\omega^{(i-1)j}a_i}^2=\sum_{j=0}^{k-1}\sum_{i,i'=1}^k\omega^{(i-i')j}a_{i'}^\ast a_i=k\sum_{i=1}^k\abs*{a_i}^2.
\end{align*}
This is tight, by taking for example $a_i=1$.

We use the weighted algebra formalism to capture approximate quantum homomorphisms between relational structures, as studied in~\cite{MS24,CM24}. In the following, $A$ and $B$ are relational structures over a signature $\sigma$.

Let $\pi$ be a probability distribution on $\parens*{\bigcup_{R\in\sigma}\{R\}\times R^A}^2$. The \emph{constraint-constraint algebra} is the weighted algebra $\Mor^{c-c}_\pi(A,B)$ generated by PVMs $\{\Phi^{R,\mathbf{a}}_{\mathbf{b}}\}_{\mathbf{b}\in R^B}$ for all $R\in\sigma$ and $\mathbf{a}\in R^A$, equipped with the weight function
\begin{align*}
    \mu_{c-c,\pi}(\Phi^{R,\mathbf{a}}_{\mathbf{b}}\Phi^{R',\mathbf{a}'}_{\mathbf{b}'})=\pi((R,\mathbf{a}),(R',\mathbf{a}'))
\end{align*}
if there exist $i,j\in[\ar(R)]$ such that $a_i=a_j'$ but $b_i\neq b_j'$, and $0$ on all other elements. Denote the defect polynomial $D^{c-c}_\pi(A,B)=D(\Mor^{c-c}_\pi(A,B),\mu_{c-c,\pi})$.

Let $\pi$ be a probability distribution on $\bigcup_{R\in\sigma}\{R\}\times R^A$. The \emph{constraint-variable algebra} is the weighted algebra $\Mor^{c-v}_\pi(A,B)$ generated by PVMs $\{\Phi^{R,\mathbf{a}}_{\mathbf{b}}\}_{\mathbf{b}\in R^B}$ for all $R\in\sigma$ and $\mathbf{a}\in R^A$ and $\{p^a_b\}_{b\in B}$ for all $a\in A$, and equipped with the weight function
\begin{align*}
    \mu_{c-v,\pi}(\Phi^{R,\mathbf{a}}_{\mathbf{b}}(1-p^{a_i}_{b_i}))=\frac{\pi(R,\mathbf{a})}{\ar(R)},
\end{align*}
for all $R\in\sigma$, $\mathbf{a}\in R^A$, and $i\in[\ar(R)]$, and $0$ on all other elements. Denote the defect polynomial $D^{c-v}_\pi(A,B)=D(\Mor^{c-v}_\pi(A,B),\mu_{c-v,\pi})$.

Let $\pi$ be a probability distribution on $\bigcup_{R\in\sigma}\{R\}\times R^A$. The \emph{assignment algebra} is the weighted algebra $\Mor^a_{\pi}(A,B)$ generated by PVMs
$\{p^a_b\}_{b\in B}$ for all $a\in A$, and equipped with the weight function
\begin{align*}
    \mu_{a,\pi}(p^{a_1}_{b_1}\cdots p^{a_{\ar(R)}}_{b_{\ar(R)}})=\pi(R,\mathbf{a}),
\end{align*}
for all $R\in\sigma$, $\mathbf{a}\in R^A$, and $\mathbf{b}\notin R^B$, and $0$ on all other elements. Denote the defect polynomial $D^{a}_\pi(A,B)=D(\Mor^{a}_\pi(A,B),\mu_{a,\pi})$.

In $\Mor^{c-c}_\pi(A,B)$ or $\Mor^{c-v}_\pi(A,B)$, write $\Phi^{R,\mathbf{a},i}_b=\sum_{\mathbf{b}\in R^B:\,b_i=b}\Phi^{R,\mathbf{a}}_{\mathbf{b}}$.

There is a bijective correspondence between  tracial correlations $p$ for $\ttt{G}_{c-c}(A,B,\pi)$ and tracial states $\tau$ on $\Mor^{c-c}_\pi(A,B)$ such that $\mfk{v}(\ttt{G}_{c-c}(A,B,\pi),p)=1-\defect(\tau)$, deterministic correlations map to characters , quantum correlations map to finite-dimensional traces, and quantum approximate correlations map to Connes-embeddable traces. There is an analogous bijective correspondence for the constraint-variable model and, if $\ar(R)=2$ for all $R\in\sigma$, for the assignment model.

\subsection{Complexity theory}

A \emph{Turing machine} is a tuple $M=(k,\Sigma,S,\delta)$, where $k\in\Z_{\geq 0}$, called the number of work tapes; $\Sigma$ is a finite set, called the set of symbols, containing a blank symbol $\Box$; $S$ is a finite set, called the set of states, containing the start state $\ttt{start}$ and the halt state $\ttt{halt}$; and $\delta:S\times\Sigma^{k+2}\rightarrow S\times\Sigma^{k+2}\times\{-1,0,1\}^{k+2}$ is a function, called the transition function. A Turing machine describes an abstract computation as follows. The machine consists of $k+2$ tapes (infinite strips where one symbol can be placed for every integer) of which one is the input tape, one is the output tape, and the remainder are work tapes; one head on each tape that can read and write on one position at a time; and a register that can store one of the states. The machine starts with a blank symbol on every position, except for the input tape where a word $x$ in $\Sigma$ can be input, and the register in the $\ttt{start}$ state. Then, at each step of the computation, the head on each tape $i$ reads the symbol $\sigma_i$ it is on, and the machine computes $\delta(s,\sigma_1,\ldots,\sigma_{k+2})=(s',\sigma_1',\ldots,\sigma_{k+2}',d_1,\ldots,d_{k+2})$, where $s$ is the state in the register. Then, the machine updates the state in the register to $s'$, writes $\sigma_i'$ at each head position $i$, and moves each head $i$ in the direction given by $d_i$. If the state in the register ever becomes $\ttt{halt}$, then the machine halts. The output of the Turing machine is the word on the output tape, omitting blank symbols, when the computation halts, and is denoted $M(x)$. A \emph{probabilistic Turing machine} is a Turing machine with an additional tape that starts filled with uniformly random symbols. This provides an abstract presentation of randomised models of computation.

The \emph{runtime} of a Turing machine on input $x$ is the number of computation steps before reaching the halt state. Since it takes a finite amount of information to specify a Turing machine, they can be represented as bit strings; importantly, there is a \emph{universal Turing machine} that can take a bit string representation of a Turing machine as input and produce the same output as would the inputted Turing machine. A function $f$ is \emph{computable} if there is a Turing machine $M$ such that $f(x)=M(x)$ for all $x$. A function $f$ is \emph{$T$-time computable}, where $T:\N\rightarrow\N$, if there is a Turing machine $M$ such that $f(x)=M(x)$ and the runtime of $M$ on input $x$ is less than $T(|x|)$. A function $f$ is \emph{polynomial-time computable} if it is $T$-time computable with $T(n)=Cn^c$ for some $C,c\geq 0$.

A \emph{language} is a subset $L\subseteq\{0,1\}^\ast$, where $\Sigma^\ast$ denotes the set of words in $\Sigma$. A Turing machine $M$ \emph{decides} a language $L$ if $M(x)=1$ if $x\in L$ and $M(x)=0$ if $x\notin L$. A \emph{promise problem} $P$ is a pair of disjoint subsets $Y(P),N(P)\subseteq\{0,1\}^\ast$. The elements of $Y(P)\cup N(P)$ are called the \emph{instances} of $P$; the elements of $Y(P)$ are the \emph{yes instances} and the elements of $N(P)$ are the \emph{no instances}. We write $x\in P$ to mean $x\in Y(P)\cup N(P)$. Every language $L$ corresponds to the promise problem $(L,L^c)$, so we will work exclusively with promise problems in the following. A Turing machine $M$ decides $P$ if $M(x)=1$ for all $x\in Y(P)$ and $M(x)=0$ for all $x\in N(P)$. Note that it does not matter what $M$ outputs for $x$ in neither $Y(P)$ nor $N(P)$ --- in fact $M$ may not even halt. A \emph{complexity class} is a collection of promise problems.

$\tsf{P}$ is the class of all problems that can be decided in polynomial time. $\tsf{NP}$ is the class of all problems $P$ such that there exists a Turing machine $M$ such that for all $x\in Y(P)$, there exists $y\in\{0,1\}^\ast$ such that $M(x,y)=1$ and for all $x\in N(P)$ and all $y\in\{0,1\}^\ast$, $M(x,y)=0$; and $M$ halts in polynomial time. $\tsf{EXP}$ is the class of all problems that can be decided in exponential time, \textit{i.e.} $T(n)=C\cdot 2^{kn^c}$ for some $C,c,k\geq 0$. $\tsf{R}$ is the class of all problems that can be decided with a Turing machine; $P\in\tsf{R}$ is called \emph{decidable}, and $P\notin\tsf{R}$ is called \emph{undecidable}. $\tsf{RE}$ is the class of all problems $P$ such that there exists a Turing machine $M$ where if $x\in Y(P)$, $M$ halts on input $x$, and if $x\in N(P)$, $M$ does not halt on input $x$. $\tsf{MIP}^\ast_{c,s}$ is the class of all problems $P$ such that for all $x\in Y(P)\cap N(P)$, there exists a nonlocal game $\ttt{G}(x)$ where the questions can be sampled in polynomial time using a probabilistic Turing machine and the predicate function can be computed in polynomial time, such that if $x\in Y(P)$, $\mfk{v}_q(\ttt{G}(x))\geq c$ and if $x\in N(P)$, $\mfk{v}_q(\ttt{G}(x))<s$. We write $\tsf{MIP}^\ast=\tsf{MIP}^\ast_{1,\frac{1}{2}}$. $\tsf{MIP}_{c,s}$ is the same except that the quantum value is replaced by the classical value. For both $\tsf{MIP}^\ast$ and $\tsf{MIP}$, the original definitions were more general, but the above ones have been shown to be sufficient~\cite{JNV+21,BFL91}. Hence, $\tsf{MIP}^\ast$ and $\tsf{MIP}$ capture the complexity of quantum and classical strategies for succinctly-presented nonlocal games.

Let $P_1$ and $P_2$ be promise problems. A \emph{(Karp) reduction from $P_1$ to $P_2$} is a computable function $f$ such that for all $x\in Y(P_1)$, $f(x)\in Y(P_2)$ and for all $x\in N(P_1)$, $f(x)\in N(P_2)$. The reduction is \emph{polynomial-time} if $f$ is polynomial-time computable. A problem $P$ is \emph{hard} for a complexity class $\tsf{C}$, or $\tsf{C}$\emph{-hard}, if every problem in $\tsf{C}$ reduces to $P$. Note that the time complexity stipulated for the reduction depends on the complexity class: for $\tsf{NP}$, $\tsf{MIP}$, or $\tsf{MIP}^*$ it is polynomial, while for $\tsf{RE}$ it just has to be computable. We say $P$ is $\tsf{C}$\emph{-complete} if it is $\tsf{C}$-hard and $P\in\tsf{C}$.

The \emph{halting problem} is the problem where the yes instances are the Turing machines that halt on empty input and the no instances are the Turing machines that do not. The halting problem is $\tsf{RE}$-complete. The equality $\tsf{MIP}^\ast=\tsf{RE}$ is due to the fact that there is a polynomial-time reduction from the halting problem to (succinctly-presented) nonlocal games~\cite{JNV+21}.

\subsection{Constraint satisfaction problems}

Given a relational structure $A$ over $\sigma$, the \emph{constraint satisfaction problem (CSP)} over $A$ is the problem $\CSP(A)$ with instances $\ob(\mc{R}(\sigma))$, where $B\in Y(\CSP(A))$ if there exists a homomorphism $B\rightarrow A$ (\textit{i.e.} $\mor(B,A)$ is nonempty) and $B\in N(\CSP(A))$ otherwise. Due to the \emph{CSP dichotomy theorem}, we know that $\CSP(A)$ is either contained in $\tsf{P}$ or $\tsf{NP}$-complete~\cite{Bul17,Zhu17}. Given $1\geq c\geq s\geq 0$, the \emph{$(c,s)$-gapped CSP} induced by $A$ is the promise problem $\CSP(A)_{c,s}$ where the yes instances are relational structures $B$ such that there is a character $\chi$ on $\Mor^a_{\mbb{u}}(B,A)$ such that $\defect(\chi)\leq 1-c$, and the no instances are relational structures $B$ such that for all characters $\chi$ on $\Mor^a_{\mbb{u}}(B,A)$, $\defect(\chi)> 1-s$. In the constraint system picture, this corresponds to there existing an assignment satisfying a $c$ fraction of the constraints in the yes case, and no assignment satisfying an $s$ fraction or more of the constraints in the no case. Note also that $\CSP(A)_{1,1}=\CSP(A)$. It follows from the PCP theorem~\cite{ALMSS98,Din07} that either $\CSP(A)_{1,s}\in\tsf{P}$ for all $s$, or there exists $s<1$ such that $\CSP(A)_{1,s}$ is $\tsf{NP}$-complete. However, $s$ in general depends on $A$, and for any $A$ there is a large enough gap that $\CSP(A)_{1,s}$ becomes easy. The \emph{$(c,s)$-gapped succinct CSP} induced by $A$ is the promise problem $\SuccinctCSP(A)_{c,s}$ where the instances are probabilistic Turing machines that sample from a probability distribution $\pi$ on $\bigcup_{R\in\sigma}\{R\}\times R^B$ where $B$ is a relational structure over $\sigma$ (note that since $B$ is presented succinctly, it is exponential-size in the length of the description of $M$); the yes instances are those such that there is a character $\chi$ on $\Mor^a_{\pi}(B,A)$ such that $\defect(\chi)\leq 1-c$ and the no instances are those such that for all characters $\chi$ on $\Mor^a_{\pi}(B,A)$, $\defect(\chi)> 1-s$.

Entangled CSPs can be presented similarly, except that there is a dependence on the way the constraint satisfaction problem is expressed as a weighted algebra. Let $w\in\{c-c,c-v,a\}$ The \emph{$(c,s)$-gapped $w$-entangled CSP} induced by $A$ is the promise problem $\CSP_{w}(A)^\ast_{c,s}$ where the yes instances are relational structures $B$ such that there exists a Connes-embeddable trace $\tau$ on $\Mor^{w}_{\mbb{u}}(B,A)$ such that $\defect(\tau)\leq1-c$, and the no instances are relational structures $B$ such that for all Connes-embeddable traces $\tau$ on $\Mor^{w}_{\mbb{u}}(B,A)$, $\defect(\tau)>1-s$. Succinct CSPs also extend to the entangled setting, and they are often arise more naturally in the context of nonlocal games. Connes-embeddable traces arise here as they constitute limits of finite-dimensional traces, which encode quantum correlations. The \emph{$(c,s)$-gapped succinct $w$-entangled CSP} induced by $A$ is the promise problem $\SuccinctCSP_{w}(A)^\ast_{c,s}$ where the instances are probabilistic Turing machines that sample from a probability distribution $\pi$ on $\bigcup_{R\in\sigma}\{R\}\times R^B$ (or $\parens*{\bigcup_{R\in\sigma}\{R\}\times R^B}^2$ if $w=c-c$) where $B$ is a relational structure over $\sigma$; the yes instances are those such that there exists a Connes-embeddable trace $\tau$ on $\Mor^{w}_{\pi}(B,A)$ such that $\defect(\tau)\leq1-c$, and the no instances are relational structures $B$ such that for all Connes-embeddable traces $\tau$ on $\Mor^{w}_{\pi}(B,A)$, $\defect(\tau)>1-s$.

A relational structure $A$ is \emph{boolean} if it has alphabet size $2$. A relational structure $A$ is \emph{two-variable falsifiable (TVF)} if for every $R\in\sigma$ and every pair $1\leq i<j\leq\ar(R)$, there exist $a,b\in A$ such that if $\mathbf{a}\in A^{\ar(R)}$ with $a_i=a$ and $a_j=b$, then $\mathbf{a}\notin R^A$. It was shown in~\cite{CM24} that if $A$ is not TVF and $\CSP(A)$ is $\tsf{NP}$-complete, then for $w\in\{c-c,c-v\}$ there exists $s<1$ such that there is a polynomial-time reduction from the halting problem to $\SuccinctCSP_{w}(A)^\ast_{1,s}$~\cite[Theorem 4.14 \& Corollary 4.15]{CM24}; and if $A$ is TVF, boolean, and $\tsf{NP}$-complete, or $A=K_3=C_3$, then for $w\in\{c-c,c-v,a\}$ there exists $s<1$ such that there is a polynomial-time reduction from the halting problem to $\SuccinctCSP_{w}(A)^\ast_{1,s}$~\cite[Theorem 4.14 \& Corollary 4.16]{CM24}. Note that the cases differ in that, in the first case, we do not necessarily have hardness of the $a$-entangled CSP. This also implies that in the above cases, $\CSP_w(A)^\ast_{1,s}$ is undecidable, as there is an exponential-time reduction from the halting problem.

\section{Quantum relational structure homomorphisms}\label{sec:quantum-homs}

In this section, we introduce spaces of quantum homomorphisms between relational structures and show their basic properties. In \cref{sec:q-hom-spaces}, we present the quantum spaces of relational structure homomorphisms. In \cref{sec:cocat}, we construct a category-theoretic framework for quantum relational structure homomorphisms, and use it to show their basic properties. In \cref{sec:q-groups}, we show every relational structure admits a quantum automorphism group and an oracular quantum automorphism group.

\subsection{Quantum spaces from relational structure homomorphisms}\label{sec:q-hom-spaces}

First, we present the quantum homomorphisms of relational structures in terms of quantum spaces. This provides a unified perspective which can be used to capture all the different classes of (tracial) quantum strategies.

\begin{definition}\label{def:quantum-endomorphism-monoid}
    Let $A$ and $B$ be finite sets. The \emph{quantum space of transformations from $A$ to $B$} is the quantum space $T^+_{A,B}$ with $C^\ast$-algebra of continuous functions $C(T^+_{A,B})$ generated by $p_{ab}$ for $a\in A$ and $b\in B$ subject to the relations $p_{ab}^\ast=p_{ab}^2=p_{ab}$ for all $a\in A$ and $b\in B$, and $\sum_{b\in B}p_{ab}=1$ for all $a\in A$. Write $T^+_{m,n}=T^+_{[m],[n]}$ and $T^+_{A}=T^+_{A,A}$.
    
    Let $A$ and $B$ be relational structures over the same signature $\sigma$. The \emph{quantum space of homomorphisms from $A$ to $B$} is the quantum subspace $\mor^+(A,B)$ of $T^+_{A,B}$ via the relations $p_{a_1b_1}p_{a_2b_2}\cdots p_{a_{\ar(R)}b_{\ar(R)}}=0$ for all $R\in\sigma$, $\mathbf{a}\in R^A$, and $\mathbf{b}\notin R^B$. The \emph{quantum endomorphism monoid of $A$} is $\enmo^+(A)=\mor^+(A,A)$. Write the algebras of functions $\Mor^+(A,B)=C(\mor^+(A,B))$ and $\End^+(A)=C(\enmo^+(A))$.

    The \emph{oracular quantum space of homomorphisms from $A$ to $B$} is the quantum subspace $\mor^{o+}(A,B)$ of $\mor^+(A,B)$ via the relations $[p_{a_ib},p_{a_jb'}]=0$ for all $R\in\sigma$, $\mathbf{a}\in R^A$, $i,j\in[\ar(R)]$, and $b,b'\in B$. The \emph{oracular quantum endomorphism monoid} is $\enmo^{o+}(A)=\mor^{o+}(A,A)$.  Write as above $\Mor^{o+}(A,B)=C(\mor^{o+}(A,B))$ and $\End^{o+}(A)=C(\enmo^{o+}(A))$.

    We say a relational structure $A$ has \emph{only classical endomorphisms} if $\enmo^+(A)=\enmo(A)$; we say that \emph{only classical oracular endomorphisms} if $\enmo^{o+}(A)=\enmo(A)$.
\end{definition}

Representations of quantum spaces of homomorphisms into finite von Neumann algebras correspond to perfect tracial strategies for CSP nonlocal games: the oracular quantum space corresponds to the constraint-constraint and constraint-variable games, and the non-oracular space corresponds to the assignment game.

It follows from the definitions that, under the assumption that the probability distributions have large enough support (see \cref{lem:perfect-strats} for the specific statement),
\begin{align*}
    &C_u^\ast(\Mor^{c-c}_{\pi'}(A,B)/\ang*{\supp\mu_{c-c,{\pi'}}})=C_u^\ast(\Mor^{c-v}_\pi(A,B)/\ang*{\supp\mu_{c-v,\pi}})=\Mor^{o+}(A,B),\\
    &C_u^\ast(\Mor^{a}_\pi(A,B)/\ang*{\supp\mu_{a,\pi}})=\Mor^{+}(A,B).
\end{align*}

We can define a \emph{coproduct} $\Delta:C(T^+_n)\rightarrow C(T^+_n)\otimes C(T^+_n)$ and a \emph{counit} $\epsilon:C(T^+_n)\rightarrow\C$ on the generators as $\Delta(p_{ij})=\sum_{k=1}^n p_{ik}\otimes p_{jk}$ and $\epsilon(p_{ij})=\delta_{i,j}$, respectively. We will see in the following section that $\Delta$ and $\epsilon$ extend to coproduct and counit $\ast$-homomorphisms on $C(T^+_n)$. These definitions also give rise to a coproduct and a counit on $\End^+(A)$ and $\End^{o+}(A)$ for any relational structure $A$.

\subsection{Category-theoretic unification} \label{sec:cocat}

Relational structure homomorphisms are easily presented as morphisms of the category~$\mc{R}(\sigma)$ (see \cref{sec:defs-rel-structs}); we show that there is a similar unified way of presenting the quantum homomorphism spaces. This is useful as it allows us to show general statements about quantum homomorphisms that hold for all classes of quantum strategies. In particular, we require the notion of composition of quantum homomorphisms for our main no-go theorem, which arises naturally in this framework.

Previous work~\cite{ABdSZ17} has shown that quantum homomorphisms on finite-dimensional Hilbert spaces can be presented as Kleisli morphisms of a monad on the category of relational structures. We use a different approach that is more in line with the notion of quantum spaces, and allows use to work with quantum homomorphisms in a wider variety of settings than finite-dimensional Hilbert spaces. To do so, we make use of a slightly modified notion of a cocategory in the sense of Bichon~\cite{Bic14}, which has previously been used to study, \textit{e.g.}, quantum analogues of groupoids~\cite{Bic14} and quantum isomorphisms and metric spaces~\cite{BCE+20,Eif20}.

\begin{definition}
    A \emph{$C^\ast$-cocategory} $\mc{C}^+$ consists of a class of objects $\ob(\mc{C}^+)$ and for every pair $a,b\in\ob(\mc{C}^+)$ a quantum space $\mor_{\mc{C}}^+(a,b)$, such that for all $a,b,c\in\ob(\mc{C}^+)$, there exists a $\ast$-homomorphism $\Delta_{a,c}^b:\Mor^+(a,c)\rightarrow\Mor^+(a,b)\otimes\Mor^+(b,c)$ called \emph{cocomposition}, where $\Mor^+(a,b)\coloneqq C(\mor^+(a,b))$, satisfying the coassociativity conditions
    $$(\Delta^b_{a,c}\otimes\id_{\Mor^+(c,d)})\Delta^c_{a,d}=(\id_{\Mor^+(a,b)}\otimes\Delta_{b,d}^c)\Delta_{a,d}^b;$$
    and for all $a\in\ob(\mc{C}^+)$, there exists a $\ast$-homomorphism $\epsilon_a:\Mor^+(a,a)\rightarrow\C$ called \emph{coidentity}, satisfying the counit conditions
    \begin{align*}
        &(\epsilon_a\otimes\id_{\Mor^+(a,b)})\Delta_{a,b}^a=(\id_{\Mor^+(a,b)}\otimes\epsilon_b)\Delta_{a,b}^b=\id_{\Mor^+(a,b)}.
    \end{align*}
    Note that reference to $\mc{C}^+$ is dropped when clear from context.
\end{definition}

Every category $\mc{C}$ where $\mor_{\mc{C}}(a,b)$ is a compact Hausdorff space can be expressed as a $C^\ast$-cocategory. In this, $\Mor^+(a,b)$ is the space of continuous functions $\mor(a,b)\rightarrow\C$, and we define the cocomposition $\Delta_{a,c}^b:C(\mor(a,c))\rightarrow C(\mor(a,b)\times\mor(b,c))$ as $\Delta_{a,c}^b(\phi)(f,g)=\phi(g\circ f)$ and the coidentity $\epsilon_a:C(\mor(a,a))\rightarrow\C$ as $\epsilon_a(\phi)=\phi(\id_A)$.

\begin{definition}
    Let $\mc{C}^+$ be a $C^\ast$-cocategory, and let $Q$ be a class of $C^\ast$-algebras closed under the minimal tensor product and containing $\C$. The \emph{$Q$-morphism category of $\mc{C}^+$} is $\mc{C}^Q$ with objects $\ob(\mc{C}^Q)=\ob(\mc{C}^+)$ and morphisms
    $$\mor_{\mc{C}^Q}(a,b)=\set*{\pi:\Mor^+(a,b)\rightarrow\mc{A}}{\pi\text{ is a $\ast$-homomorphism and }\mc{A}\in Q},$$
    writing $\mor^Q(a,b)=\mor_{\mc{C}^Q}(a,b)$ when $\mc{C}^+$ is clear from context, with composition defined for $\pi\in\mor^Q(a,b)$ and $\varphi\in\mor^Q(b,c)$ as
    $$\varphi\circ\pi=(\pi\otimes\varphi)\Delta_{a,c}^b.$$
\end{definition}
We write $\pi:a\overset{Q}{\rightarrow}b$ to mean $\pi\in\mor^Q(a,b)$.

We make use of the following classes of $C^\ast$-algebras: $d=\{\C\}$ gives the deterministic morphisms, $c=\set*{C(X)}{X\text{ is a compact space}}$ gives the classical morphisms, $q=\set*{B(\C^d)}{d\in\N}$ gives the quantum morphisms, $qa=\set*{\mc{A}}{\mc{A}\hookrightarrow\mc{R}^\omega}$ gives the quantum approximate morphisms, and $qc=\set*{\mc{A}}{\exists\,\tau:\mc{A}\rightarrow\C\text{ tracial state}}$ gives the quantum commuting morphisms.

\begin{remark}
    Note that each of the classes defined above is closed under the spatial tensor product. For $Q=d$, $\C\otimes\C=\C$ by definition. For $Q=c$, this fact follows from the isomorphism $C(X)\otimes C(Y)\cong C(X\times Y)$. For $Q=q$, the fact follows from the isomorphism $B(\C^{d_1})\otimes B(\C^{d_2})\cong B(\C^{d_1\cdot d_2})$. For $Q=qc$, for any tracial states $\tau_1:\mc{A}_1\rightarrow\C$ and $\tau_2:\mc{A}_2\rightarrow\C$, the tracial state on the algebraic tensor product $\mc{A}_1\odot\mc{A}_2$ extends to a tracial state $\tau_1\otimes\tau_2:\mc{A}_1\otimes\mc{A}_2\rightarrow\C$. Finally, for $Q=qa$, let $\mc{A}_1,\mc{A}_2\in qa$. So, there exist injections $\iota_i:\mc{A}_i\rightarrow\mc{R}^\omega$. $\mc{R}^\omega\subseteq B(H)$ for some Hilbert space $H$, so $\iota_i$ is a faithful representation of $\mc{A}_i$. As such, $\iota_1\odot\iota_2:\mc{A}_1\odot\mc{A}_2\rightarrow B(H\otimes H)$ extends to a faithful representation $\iota_1\otimes\iota_2$ of $\mc{A}_1\otimes\mc{A}_2$. $(\iota_1\otimes\iota_2)(\mc{A}_1\otimes\mc{A}_2)$ is the closure of $\iota_1(\mc{A}_1)\odot\iota_2(\mc{A}_2)$ with respect to the operator norm, so it is contained within the closure of $\mc{R}^\omega\odot\mc{R}^\omega$ with respect to the operator norm, which is in turn contained within the closure in the strong operator topology, $\mc{R}^\omega\bar{\otimes}\mc{R}^\omega\cong\mc{R}^\omega$. Hence, $\iota_1\otimes\iota_2$ can be corestricted to an injective $\ast$-homomorphism $\mc{A}_1\otimes\mc{A}_2\rightarrow\mc{R}^\omega$, giving that $\mc{A}_1\otimes\mc{A}_2\in qa$.
\end{remark}

\begin{lemma}
    Let $\mc{C}^+$ be a $C^\ast$-cocategory and $Q$ be a class of $C^\ast$-algebras that contains $\C$ and is closed under the minimal tensor product. Then $\mc{C}^Q$ is well-defined as a category.
\end{lemma}

\begin{proof}
    First, note that composition is well-defined. Let $\pi\in\mor^Q(a,b)$ and $\varphi\in\mor^Q(b,c)$. Then, there exist $\mc{A},\mc{B}\in Q$ such that $\pi:\Mor^+(a,b)\rightarrow\mc{A}$ and $\varphi:\Mor^+(a,b)\rightarrow\mc{B}$ are $\ast$-homomorphisms. As such, $\varphi\circ\pi:\Mor^Q(a,c)\rightarrow\mc{A}\otimes\mc{B}$ is an element of $\mor^Q(a,c)$.
    
    Next, we need to show that the composition is associative. Let $\pi\in\mor^Q(a,b)$, $\varphi\in\mor^Q(b,c)$, and $\psi\in\mor^Q(c,d)$. Then, by coassociativity,
    \begin{align*}
        \psi\circ(\varphi\circ\pi)&=((\varphi\circ\pi)\otimes\psi)\Delta_{a,d}^c\\
        &=(\pi\otimes\varphi\otimes\psi)(\Delta_{a,c}^b\otimes\id)\Delta_{a,d}^c\\
        &=(\pi\otimes\varphi\otimes\psi)(\id\otimes\Delta_{b,d}^c)\Delta_{a,d}^b\\
        &=(\pi\otimes(\psi\circ\varphi))\Delta_{a,d}^b\\
        &=(\psi\circ\varphi)\circ\pi.
    \end{align*}
    Finally, we need to show the existence of an identity for each $a\in\ob(\mc{C}^Q)$. We claim $\epsilon_a:\Mor^+(a,a)\rightarrow\C$ is this element. In fact, for any $\pi\in\mor^Q(a,b)$, we get by the counit property that
    \begin{align*}
        &\pi\circ\epsilon_a=(\epsilon_a\otimes\pi)\Delta_{a,b}^a=\pi\\
        &\epsilon_b\circ\pi=(\pi\otimes\epsilon_b)\Delta_{a,b}^b=\pi.\qedhere
    \end{align*}
\end{proof}

Now, we can use the above objects to capture quantum homomorphisms between relational structures.

\begin{definition}
    The \emph{$C^\ast$-cocategory of relational structures} over a signature $\sigma$ is the $C^\ast$-cocategory $\mc{R}^+(\sigma)$ where the objects are the relational structures, for $A,B\in\ob(\mc{R}^+(\sigma))$, $\mor^+(A,B)$ is the quantum space of homomorphisms from $A$ to $B$, the cocomposition is defined via $\Delta_{A,C}^B(p_{ac})=\sum_{b\in B}p_{ab}\otimes p_{bc}$, and the coidentity is defined by $\epsilon_A(p_{aa'})=\delta_{a,a'}$.

    The \emph{oracular $C^\ast$-cocategory of relational structures} over $\sigma$ is the $C^\ast$-cocategory $\mc{R}^{o+}(\sigma)$ defined identically except that $\mor^+_{\mc{R}^o(\sigma)}(A,B)=\mor^{o+}(A,B)$ is the oracular quantum space of homomorphisms from $A$ to $B$.
\end{definition}

\begin{lemma}
    $\mc{R}^+(\sigma)$ and $\mc{R}^{o+}(\sigma)$ are well-defined as $C^\ast$-cocategories.
\end{lemma}

\begin{proof}
    To show this, we need that the cocomposition and coidentity extend to $\ast$-homomorphisms that satisfy the coproduct and counit relations. First, the relations of $C(T_{A,C}^+)$ are $p_{ac}=p_{ac}^\ast$ for all $a\in A$ and $c\in C$; $p_{ac}^2=p_{ac}$ for all $a\in A$ and $c\in C$; and $\sum_{c\in C}p_{ac}=1$ for all $a\in A$. Note that as $C(T_{A,C}^+)$ is a $C^\ast$-algebra, it follows that $p_{ac}p_{ac'}=\delta_{c,c'}p_{ac}$. These are preserved under $\Delta_{A,C}^B$ as
    \begin{align*}
        &\Delta_{A,C}^B(p_{ac})^\ast=\parens[\Big]{\sum_{b\in B}p_{ab}\otimes p_{bc}}^\ast=\sum_{b\in B}p_{ab}^\ast\otimes p_{bc}^\ast=\sum_{b\in B}p_{ab}\otimes p_{bc}=\Delta_{A,C}^B(p_{ac}),\\
        &\Delta_{A,C}^B(p_{ac})^2=\sum_{b,b'\in B}p_{ab}p_{ab'}\otimes p_{bc}p_{b'c}=\sum_{b,b'\in B}\delta_{b,b'}p_{ab}\otimes p_{bc}p_{b'c}=\sum_{b\in B}p_{ab}\otimes p_{bc}=\Delta_{A,C}^B(p_{ac}),\\
        &\sum_{c\in C}\Delta_{A,C}^B(p_{ac})=\sum_{\substack{b\in B\\c\in C}}p_{ab}\otimes p_{bc}=\sum_{b\in B}p_{ab}\otimes 1=1\otimes 1=\Delta_{A,C}^B(1).
    \end{align*}
    Now, let $R\in\sigma$ and let $\mathbf{a}\in R^A$ and $\mathbf{c}\notin R^C$. Writing $n=\ar(R)$, the relation $p_{a_1c_1}\cdots p_{a_nc_n}=0$ is preserved under $\Delta_{A,C}^B$ as
    \begin{align*}
        \Delta_{A,C}^B(p_{a_1c_1})\cdots\Delta_{A,C}^B(p_{a_nc_n})&=\sum_{b_1,\ldots,b_n\in B}p_{a_1b_1}\cdots p_{a_nb_n}\otimes p_{b_1c_1}\cdots p_{b_nc_n}\\
        &=\sum_{\mathbf{b}\in R^B}p_{a_1b_1}\cdots p_{a_nb_n}\otimes p_{b_1c_1}\cdots p_{b_nc_n}=0,
    \end{align*}
    as $\mathbf{a}\in R^A$. Therefore, $\Delta_{A,C}^B$ extends to a $\ast$-homomorphism. Next, we need coassociativity of the the cocomposition. Since the $\Delta_{A,C}^B$ are $\ast$-homomorphisms, we need only show this property on the generators. In fact,
    \begin{align*}
        (\Delta_{A,C}^B\otimes\id)\Delta_{A,D}^C(p_{ad})&=\sum_{c\in C}\Delta_{A,C}^B(p_{ac})\otimes p_{cd}\\
        &=\sum_{\substack{b\in B\\c\in C}}p_{ab}\otimes p_{bc}\otimes p_{cd}\\
        &=\sum_{b\in B} p_{ab}\otimes\Delta_{B,D}^C(p_{bd})\\
        &=(\id\otimes\Delta_{B,D}^C)\Delta_{A,D}^B(p_{ad}).
    \end{align*}
    We also need to show that the coidentity is a $\ast$-homomorphism. We approach this the same way:
    \begin{align*}
        &\epsilon_A(p_{aa'})^\ast=\delta_{a,a'}^\ast=\delta_{a,a'}=\epsilon_A(p_{aa'}),\\
        &\epsilon_A(p_{aa'})^2=\delta_{a,a'}^2=\epsilon_A(p_{aa'}),\\
        &\sum_{a'\in A}\epsilon_A(p_{aa'})=\sum_{a'\in A}\delta_{a,a'}=1=\epsilon_A(1).
    \end{align*}
    For the remaining relations, suppose $R\in\sigma$, $\mathbf{a}\in R^A$ and $\mathbf{a}'\notin R^A$. Writing $n=\ar(R)$, we find that
    \begin{align*}
        \epsilon_A(p_{a_1a_1'})\cdots\epsilon_A(p_{a_na_n'})=\delta_{a_1,a_1'}\cdots\delta_{a_n,a_n'}=\delta_{\mathbf{a},\mathbf{a}'}=0,
    \end{align*}
    as $\mathbf{a}$ and $\mathbf{a}'$ cannot be equal. Hence $\epsilon_A$ is a $\ast$-homomorphism. As such, to show it also satisfies the counit relations, it suffices to show it on the generators. In fact,
    \begin{align*}
        &(\epsilon_A\otimes\id)\Delta_{A,B}^A(p_{ab})=\sum_{a'\in A}\epsilon_A(p_{aa'})\otimes p_{a'b}=\sum_{a'\in A}\delta_{a,a'}p_{a'b}=p_{ab},\\
        &(\id\otimes\epsilon_B)\Delta_{A,B}^B(p_{ab})=\sum_{b'\in B}p_{ab'}\otimes\epsilon_B(p_{b'b})=\sum_{b'\in B}\delta_{b',b}p_{ab'}=p_{ab}.
    \end{align*}
    Hence $\mc{R}^+(\sigma)$ is well-defined as a $C^\ast$-cocategory

    To show that $\mc{R}^{o+}(\sigma)$ is also well-defined, it suffices to show that the cocomposition and the coidentity preserve the additional commutation relations $[p_{a_ic},p_{a_jc'}]=0$ for all $R\in\sigma$, $\mathbf{a}\in R^A$, $1\leq i,j\leq\ar(R)$, and $c,c'\in C$. In fact, writing $n=\ar(R)$ and assuming without loss of generality that $i<j$, we have that
    \begin{align*}
        [\Delta_{A,C}^B(p_{a_ic}),\Delta_{A,C}^B(p_{a_jc'})]&=\sum_{b,b'\in B}[p_{a_ib}\otimes p_{bc},p_{a_jb'}\otimes p_{b'c'}]\\
        &=\sum_{b,b'\in B}p_{a_ib}p_{a_jb'}\otimes[p_{bc},p_{b'c'}]\\
        &=\sum_{b_1,\ldots,b_n\in B}p_{a_1b_1}\cdots p_{a_nb_n}\otimes[p_{b_ic},p_{b_jc'}]\\
        &=\sum_{\mathbf{b}\in R^B}p_{a_1b_1}\cdots p_{a_nb_n}\otimes[p_{b_ic},p_{b_jc'}]=0,
    \end{align*}
    by the commutation relation on $\Mor^+(B,C)$. It is clear that $[\epsilon_A(p_{a_ia}),\epsilon_A(p_{a_ja'})]=0$, as the codomain of $\epsilon_A$ is a commutative algebra. Hence $\mc{R}^{o+}(\sigma)$ is well-defined as well.
\end{proof}

To close this section, we see how this unified picture captures the classical and quantum notions of morphisms between relational structures. First, the $d$-morphism category is essentially the same as the classical category of relational structures.

\begin{lemma}
    $\mc{R}^d(\sigma)\cong\mc{R}^{od}(\sigma)\cong\mc{R}(\sigma)$.
\end{lemma}

\begin{proof}
    First, note that for a $\ast$-homomorphism $\pi:\Mor^+(A,B)\rightarrow\C$, every commutator in $\Mor^+(A,B)$ is in the kernel of $\pi$, and hence it factors through $\Mor^{o+}(A,B)$. As such, we get a bijection between $\mor^d(A,B)$ and $\mor^{od}(A,B)$. In the following, we show the isomorphism $\mc{R}^d(\sigma)\cong\mc{R}(\sigma)$ --- by the bijections above, an identical argument can be used to show the other isomorphism $\mc{R}^{od}(\sigma)\cong\mc{R}(\sigma)$.
    
    Next, $\{\pi(p_{ab})\}_{b\in B}$ is a PVM in $\C$. So $\pi(p_{ab})=0$ or $\pi(p_{ab})=1$, and for each $a$ there is a unique $b$ such that $\pi(p_{ab})=1$. This induces a unique function $f:A\rightarrow B$ such that $\pi(p_{ab})=\delta_{b,f(a)}$. We define the functors $F:\mc{R}^d(\sigma)\rightarrow\mc{R}(\sigma)$ and $G:\mc{R}(\sigma)\rightarrow\mc{R}^d(\sigma)$ on objects as $F(A)=G(A)=A$ and on morphisms as $F(\pi)=f$ and $G(f)$ as the extension of a map $p_{ab}\mapsto\delta_{b,f(a)}$ to an $\ast$-homomorphism. It is clear that these are inverses of each other if they are well-defined. It remains to show that they are in fact functors. To show that $G$ is a functor, we need to first show that for all $f\in\mor(A,B)$, $G(f)\in\mor^d(A,B)$. This follows because $G(f)$ preserves the relations of $\Mor^+(A,B)$: let $a\in A$, $b,b'\in B$, $R\in\sigma$, $\mathbf{a}\in R^A$, and $\mathbf{b}\notin R^B$, and writing~$n=\ar(R)$,
    \begin{align*}
        &G(f)(p_{ab})^\ast=\delta_{b,f(a)}=G(f)(p_{ab}),\\
        &G(f)(p_{ab})^2=\delta_{b,f(a)}^2=G(f)(p_{ab}),\\
        &\sum_{b\in B}G(f)(p_{ab})=\sum_{b\in B}\delta_{b,f(a)}=1,\\
        &G(f)(p_{a_1b_1})\cdots G(f)(p_{a_nb_n})=\delta_{b_1,f(a_1)}\cdots\delta_{b_n,f(a_n)}=\delta_{\mathbf{b},f(\mathbf{a})}=0,
    \end{align*}
    as $f(\mathbf{a})\in R^B$ since $f$ is a homomorphism. Next, we need that $G$ preserves compositions. Let $f\in\mor(A,B)$ and $g\in\mor(B,C)$. Then, we have that
    \begin{align*}
        (G(g)\circ G(f))(p_{ac})&=(G(f)\otimes G(g))\Delta_{A,C}^B(p_{ac})\\
        &=\sum_{b\in B} G(f)(p_{ab})\otimes G(g)(p_{bc})\\
        &=\sum_{b\in B}\delta_{b,f(a)}\delta_{c,g(b)}\\
        &=\delta_{c,g(f(a))}=G(g\circ f)(p_{ac}).
    \end{align*}
    As the maps are $\ast$-homomorphisms, this implies that $G(g)\circ G(f)=G(g\circ f)$ as wanted. Finally, $G$ preserves the identity, as $G(\id_A)(p_{aa'})=\delta_{a,a'}=\epsilon_A(p_{aa'})$, so $G$ is a functor.

    We also need to show that $F$ is a functor. Let $\pi\in\mor^d(A,B)$ and suppose that $F(\pi)\notin\mor(A,B)$. Then, there exists $R\in\sigma$ and $\mathbf{a}\in R^A$ such that $F(\pi)(\mathbf{a})\notin R^B$. This means that
    \begin{align*}
        0=\pi(p_{a_1\,F(\pi)(a_1)}\cdots p_{a_n\,F(\pi)(a_n)})=\delta_{F(\pi)(a_1),F(\pi)(a_1)}\cdots\delta_{F(\pi)(a_n),F(\pi)(a_n)}=1,
    \end{align*}
    a contradiction. Hence, $F(\pi)$ is a homomorphism. Next, we need that $F$ preserves compositions. Let $\pi\in\mor^d(A,B)$ and $\varphi\in\mor^d(B,C)$, so $F(\varphi\circ\pi)(a)$ is such that $(\varphi\circ\pi)(p_{ac})=\delta_{c,F(\varphi\circ\pi)(a)}$. But
    $$(\varphi\circ\pi)(p_{ac})=(\pi\otimes\varphi)\Delta_{A,C}^B(p_{ac})=\sum_{b\in B}\pi(p_{ab})\otimes\varphi(p_{bc})=\sum_{b\in B}\delta_{b,F(\pi)(a)}\delta_{c,F(\varphi)(b)}=\delta_{c,F(\varphi)\circ F(\pi)(b)},$$
    so $F(\varphi\circ\pi)=F(\varphi)\circ F(\pi)$. Finally,
    $F(\epsilon_A)(a)=a$, so it is the identity. Hence, $F$ is a functor as well, finishing the proof.
\end{proof}

Next, perfect strategies for the nonlocal game realisations of CSPs can be captured by $Q$-morphism categories.

\begin{lemma}\label{lem:perfect-strats}
    Let $A$ and $B$ be relational structures over $\sigma$, and let $Q\in\{d,c,q,qa,qc\}$. 
    \begin{enumerate}[(i)]
        \item Let $\pi$ be a probability distribution on $\bigcup_{R\in\sigma}\{R\}\times R^A$  with full support. Then, $\mor^{oQ}(A,B)\neq\varnothing$ if and only if there exists $p\in C_{t,Q}$ such that $\mfk{v}(\ttt{G}_{c-v}(A,B,\pi),p)=1$.
        
        \item Let $\pi$ be a probability distribution on $\parens*{\bigcup_{R\in\sigma}\{R\}\times R^A}^2$ such that, for all $a\in A$, the graph $G_a$ is connected, where $G_a$ is the graph with vertices $(R,\mathbf{a})$ for all $R\in\sigma$ and $\mathbf{a}\in R^A$ where there exists $i\in[\ar(R)]$ such that $a_i=a$; and edges $(R,\mathbf{a})\sim_{G_a}(R',\mathbf{a}')$ iff $\pi((R,\mathbf{a}),(R',\mathbf{a}'))>0$. Then, $\mor^{oQ}(A,B)\neq\varnothing$ if and only if there exists $p\in C_{t,Q}$ such that $\mfk{v}(\ttt{G}_{c-c}(A,B,\pi),p)=1$.
        
        \item Suppose $\ar(R)=2$ for all $R\in\sigma$ and let $\pi$ be a probability distribution on $\bigcup_{R\in\sigma}\{R\}\times R^A$  with full support. Then, $\mor^Q(A,B)\neq\varnothing$ if and only if there exists $p\in C_{t,Q}$ such that $\mfk{v}(\ttt{G}_a(A,B,\pi),p)=1$.
    \end{enumerate}
\end{lemma}

Part (i) for $Q=q$ was shown in \cite{ABdSZ17}.

\begin{proof}
    All the parts of this result are proved in a similar way. We show part (ii) with $Q=qa$ here. First, if $\pi\in\mor^{oqa}(A,B)$, then $\pi:\Mor^+(A,B)\rightarrow\mc{A}$ for a $C^\ast$-algebra $\mc{A}$ that embeds into $\mc{R}^\omega$. Hence, there is a tracial state $\tau$ on $\mc{A}$. This gives rise to a tracial quantum approximate correlation $p(\mathbf{b},\mathbf{b}'|(R,\mathbf{a}),(R',\mathbf{a}'))=\tau(p_{a_1b_1}\cdots p_{a_{\ar(R)}b_{\ar(R)}}p_{a_1'b_1'}\cdots p_{a_{\ar(R)}'b_{\ar(R)}'})$. By construction, we have that $p(\mathbf{b},\mathbf{b}'|(R,\mathbf{a}),(R',\mathbf{a}'))=0$ if $\mathbf{b}\notin R^B$ or $\mathbf{b}'\notin (R')^B$; and if $a_i=a_j'$ and $b_i\neq b_j'$ for some $i,j$, then $p(\mathbf{b},\mathbf{b}'|(R,\mathbf{a}),(R',\mathbf{a}'))=\tau(p_{a_1b_1}\cdots p_{a_{\ar(R)}b_{\ar(R)}}p_{a_ib_i}p_{a_j'b_j'}p_{a_1'b_1'}\cdots p_{a_{\ar(R)}'b_{\ar(R)}'})=0$. From here, it follows that
    \begin{align*}
        \mfk{v}(\ttt{G}_{c-c}(A,B,\pi),p)&=\hspace{-0.2cm}\sum_{(R,\mathbf{a}),(R',\mathbf{a}'),\mathbf{b},\mathbf{b}'}\hspace{-0.2cm}\pi((R,\mathbf{a}),(R',\mathbf{a}'))V(\mathbf{b},\mathbf{b}'|(R,\mathbf{a}),(R',\mathbf{a}'))p(\mathbf{b},\mathbf{b}'|(R,\mathbf{a}),(R',\mathbf{a}'))\\
        &=\sum_{\substack{R,R'\in\sigma,\mathbf{a}\in R^A,\mathbf{a}'\in(R')^A\\\mathbf{b}\in R^B,\mathbf{b}'\in(R')^B\\a_i=a_j'\Rightarrow b_i=b_j'}}\pi((R,\mathbf{a}),(R',\mathbf{a}'))p(\mathbf{b},\mathbf{b}'|(R,\mathbf{a}),(R',\mathbf{a}'))\\
        &=\sum_{R,R'\in\sigma,\mathbf{a}\in R^A,\mathbf{a}'\in(R')^A}\pi((R,\mathbf{a}),(R',\mathbf{a}'))=1.
    \end{align*}
    Conversely, let $p\in C_{t,qa}$ be a perfect correlation for $\ttt{G}_{c-c}(A,B,\pi)$. First, there exist PVMs $\{P^{(R,\mathbf{a})}_{\mathbf{b}}\}_{\mathbf{b}\in R^B}\subseteq\mc{R}^\omega$ for all $R\in\sigma$ and $\mathbf{a}\in R^A$ such that $p(\mathbf{b},\mathbf{b}'|(R,\mathbf{a}),(R',\mathbf{a}'))=\tau(P^{(R,\mathbf{a})}_{\mathbf{b}}P^{(R',\mathbf{a}')}_{\mathbf{b}'})$. Next, since $p$ is perfect, we have that $p(\mathbf{b},\mathbf{b}'|(R,\mathbf{a}),(R',\mathbf{a}'))=0$ whenever there exists $i,j$ such that $a_i=a_j'$ but $b_i\neq b_j'$. Write $P^{(R,\mathbf{a}),i}_b=\sum_{\mathbf{b}:\,b_i=b}P^{(R,\mathbf{a})}_{\mathbf{b}}$. Then, if $a_i=a_j'$ and $\pi((R,\mathbf{a}),(R',\mathbf{a}'))>0$, $\tau(P^{(R,\mathbf{a}),i}_bP^{(R',\mathbf{a}'),j}_{b'})=0$ for all $b\neq b'$. By faithfulness of the trace $P^{(R,\mathbf{a}),i}_bP^{(R',\mathbf{a}'),j}_{b'}P^{(R,\mathbf{a}),i}_b=0$ and therefore $P^{(R,\mathbf{a}),i}_bP^{(R',\mathbf{a}'),j}_{b'}=0$ for all $b'\neq b$, which implies $$P^{(R,\mathbf{a}),i}_b=\sum_{b'\in B}P^{(R,\mathbf{a}),i}_bP^{(R',\mathbf{a}'),j}_{b'}=P^{(R,\mathbf{a}),i}_bP^{(R',\mathbf{a}'),j}_b=\sum_{b'}P^{(R,\mathbf{a}),i}_{b'}P^{(R',\mathbf{a}'),j}_{b}=P^{(R',\mathbf{a}'),j}_{b}.$$ As such, for any $a\in A$, the connectedness of $G_a$ implies that $P^{(R,\mathbf{a}),i}_b=P^{(R',\mathbf{a}'),j}_b$ whenever $a_i=a_j'=a$, allowing us to define $P^a_b=P^{(R,\mathbf{a}),i}_b$. Then, the map $p_{ab}\mapsto P^a_b$ extends to a $\ast$-homomorphism $\Mor^{o+}(A,B)\rightarrow\mc{R}^\omega$ as the $\{P^a_b\}_{b\in B}$ are PVMs, $P^{a_1}_{b_1}\cdots P^{a_{\ar(R)}}_{b_{\ar(R)}}=P^{(R,\mathbf{a})}_{\mathbf{b}}$, and $[P^{a_i}_b,P^{a_j}_{b'}]=0$ if $\mathbf{a}\in R^A$. Hence $\mor^{oqa}(A,B)\neq\varnothing$.
\end{proof}

\subsection{Quantum groups from relational structures}\label{sec:q-groups}

As an aside, we show that the quantum spaces of morphisms of relational structures may be used to define new quantum groups. This may be of independent interest, but is not required for the main results and can be skipped.

\begin{definition}
    Let $A$ be a relational structure. The \emph{quantum automorphism group of $A$} is the quantum subgroup $\aut^+(A)$ of the quantum permutation group $S_A^+$ with respect to the relations $p_{a_1a_1'}\cdots p_{a_{\ar(R)}a_{\ar(R)}'}=p_{a_1'a_1}\cdots p_{a_{\ar(R)}'a_{\ar(R)}}=0$ for all $R\in\sigma$, $\mathbf{a}\in R^A$, and $\mathbf{a}'\notin R^A$.

    The \emph{oracular quantum automorphism group of $A$} is the quantum subgroup $\aut^{o+}(A)$ of $\aut^+(A)$ with respect to the relations $[p_{a_ia},p_{a_ja'}]=[p_{aa_i},p_{a'a_j}]=0$ for all $a,a'\in A$, $R\in\sigma$, and $\mathbf{a}\in R^A$.
\end{definition}

\begin{lemma}
    $\aut^+(A)$ and $\aut^{o+}(A)$ are compact matrix quantum groups.
\end{lemma}

\begin{proof}
    We need first to show that the coproduct $\Delta$ of $S_A^+$ preserves the relations of $\Aut^+(A)=C(\aut^+(A))$ and $\Aut^{o+}(A)=C(\aut^{o+}(A))$. In fact, we know that $\Delta$ preserves the relations of $\End^+(A)$ and $\End^{o+}(A)$, that is $p_{a_1a_1'}\cdots p_{a_{\ar(R)}a_{\ar(R)}'}=0$ and $[p_{a_ia},p_{a_ja'}]=0$. To show that it preserves the other relations consider the transpose map $\tau:\Aut^+(A)\rightarrow\Aut^+(A)$ given by $\tau(p_{aa'})=p_{a'a}$ and the tensor swap map $\sigma:\Aut^+(A)^{\otimes 2}\rightarrow\Aut^+(A)^{\otimes 2}$ given by $\sigma(x\otimes y)=y\otimes x$. $\sigma$ is a standard $\ast$-homomorphism, and $\tau$ is also an $\ast$-homomorphism as the set of relations is invariant under this transpose. $\tau$ exchanges the relations of $\End^+(A)$, which we know $\Delta$ satisfies, with the new relations that we need to show it satisfies, and we have also that
    $$\Delta\tau(p_{aa'})=\sum_{b\in A}p_{a'b}\otimes p_{ba}=(\tau\otimes\tau)\sigma\Delta(p_{aa'}).$$
    Hence, $\Delta$ preserves the new relations of $\Aut^+(A)$ as
    \begin{align*}
        \Delta(p_{a_1'a_1})\cdots \Delta(p_{a_{\ar(R)}'a_{\ar(R)}})&=\Delta\tau(p_{a_1a_1'})\cdots \Delta\tau(p_{a_{\ar(R)}a_{\ar(R)}'})\\
        &=\sigma(\tau\otimes\tau)(\Delta(p_{a_1a_1'})\cdots \Delta(p_{a_{\ar(R)}a_{\ar(R)}'}))\\
        &=0.
    \end{align*}
    The same holds in the same way for the relations $[p_{aa_i},p_{a'a_j}]=0$ of $\Aut^{o+}(A)$.

    To finish, we need that $u=(p_{aa'})$ and $u^T$ are invertible in $\mbb{M}_{|A|}(\Aut^+(A))$. As for any quantum subgroup of $S_{A}^+$, they are unitary, which we show explicitly here:
    \begin{align*}
        &(u^\ast u)_{aa'}=\sum_{b\in A}p^\ast_{ba}p_{ba'}=\sum_{b\in A}\delta_{a,a'}p_{ba}=\delta_{a,a'}1,\\
        &(uu^\ast)_{aa'}=\sum_{b\in A}p_{ab}p_{a'b}^\ast=\sum_{b\in A}\delta_{a,a'}p_{ab}=\delta_{a,a'}1,
    \end{align*}
    so $u$ is unitary. As $u^T=u^\ast$, it is also unitary. The same holds in the same way for $\aut^{o+}(A)$.
\end{proof}

\section{Commutativity gadgets}\label{sec:comm-gadgets}

In this section, we formally define commutativity gadgets, study their properties --- notably their relation to the complexity of entangled CSPs --- and prove a no-go theorem relating the existence of non-classical endomorphisms to the nonexistence of commutativity gagdets. Commutativity gadgets allow classical reductions between CSPs to be lifted to the setting of their entangled analogues, by forcing variables to commute and hence allowing them to express classical constraints. To achieve this, a commutativity gadget needs to satisfy two fundamental properties: for any two variables, there is a quantum homomorphism that assigns the two variables to any two classical values; and for any quantum homorphism, the PVMs corresponding to these variables commute. We formalise this as follows.

\begin{definition}
    Let $A$ be a relational structure over a signature $\sigma$. A \emph{commutativity gadget} for $A$ is a relational structure $G$ over $\sigma$ along with $x,y\in G$, called the distinguished variables, such that
    \begin{enumerate}[(i)]
        \item For all $a,b\in A$, there exists $\pi_{a,b}\in\mor^{qa}(G,A)$ such that $\pi_{a,b}(p_{xa})=\pi_{a,b}(p_{yb})=1$,
        \item For all $\pi\in\mor^{qa}(G,A)$ and $a,b\in A$, $[\pi(p_{xa}),\pi(p_{yb})]=0$.
    \end{enumerate}
    We call $(G,x,y)$ an \emph{oracular commutativity gadget} if the same properties holds with respect to $\mor^{oqa}(G,A)$.
\end{definition}

One could also make an analogous definition for commuting operator homomorphisms, but we focus in this work on quantum (approximate) homomorphisms.

The remainder of this section is organised as follows. In \cref{sec:robust}, we investigate robust commutativity gadgets, that allow us to construct gapped reductions. In \cref{sec:rel-oracular}, we give conditions under which oracular and non-oracular commutativity gadgets may be related. In \cref{sec:comm-gadgets-complexity}, we show how the existence of robust commutativity gadgets gives rise to gapped hardness of entangled CSPs. In \cref{sec:nogo}, we prove that non-classical endomorphisms imply nonexistence of a commutativity gadget.

\subsection{Robust commutativity gadgets}\label{sec:robust}

To capture approximate notions of satisfiability, we need to work with stronger notions of commutativity gadgets that allow us to provide guarantees on the defects of weighted algebras.

\begin{definition}
    A \emph{$c-c$-robust commutativity gadget} for $A$ over $\sigma$ is an oracular commutativity gadget $(G,x,y)$ such that for all $\varepsilon>0$ there exists $\delta>0$ where if $\defect(\tau)<\delta$ for a finite-dimensional tracial state on $\Mor^{c-c}_{\mbb{u}}(G,A)$, then
    $$\frac{1}{m^2}\sum_{\substack{R\in\sigma,\mathbf{x}\in R^G,i\in[\ar(R)]:\,x_i=x\\S\in\sigma,\mathbf{y}\in S^G,j\in[\ar(R)]:\,y_j=y\\a,b\in A}}\norm{[\Phi^{R,\mathbf{x},i}_a,\Phi^{S,\mathbf{y},j}_b]}_\tau^2<\varepsilon,$$
    where $m=\sum_{R\in\sigma}|R^A|$ so $\mbb{u}((R,\mathbf{x}),(S,\mathbf{y}))=\frac{1}{m^2}$.

    A \emph{$c-v$-robust commutativity gadget} for $A$ over $\sigma$ is an oracular commutativity gadget $(G,x,y)$ such that for all $\varepsilon>0$ there exists $\delta>0$ where if $\defect(\tau)<\delta$ for a finite-dimensional tracial state on $\Mor^{c-v}_{\mbb{u}}(G,A)$, then
    $$\sum_{a,b\in A}\norm{[p^x_a,p^y_b]}_\tau^2<\varepsilon.$$

    An \emph{$a$-robust commutativity gadget} for $A$ over $\sigma$ is an commutativity gadget $(G,x,y)$ such that for all $\varepsilon>0$ there exists $\delta>0$ where if $\defect(\tau)<\delta$ for a finite-dimensional tracial state on $\Mor^{a}_{\mbb{u}}(G,A)$, then
    $$\sum_{a,b\in A}\norm{[p^x_a,p^y_b]}_\tau^2<\varepsilon.$$
\end{definition}

\begin{lemma}
    $(G,x,y)$ is a $c-c$-robust commutativity gadget if and only if it is a $c-v$-robust commutativity gadget.
\end{lemma}

\begin{proof}
     The inclusion map is a $4L$-homomorphism $\alpha:\Mor^{c-c}_{\mbb{u}}(G,A)\rightarrow\Mor^{c-v}_{\mbb{u}}(G,A)$, where $L=\max_{R\in\sigma}\ar(R)$, and there exists an $m$-homomorphism $\beta:\Mor^{c-v}_{\mbb{u}}(G,A)\rightarrow\Mor^{c-c}_{\mbb{u}}(G,A)$, which is identity on the $\Phi^{R,\mathbf{a}}_{\mathbf{b}}$~\cite{CM24}. As such, to show the equivalence of the robust gadgets, it suffices to relate the values of $\frac{1}{m^2}\sum_{\substack{R\in\sigma,\mathbf{x}\in R^G,i\in[\ar(R)]:\,x_i=x\\S\in\sigma,\mathbf{y}\in S^G,j\in[\ar(R)]:\,y_j=y\\a,b\in A}}\norm{[\Phi^{R,\mathbf{x},i}_a,\Phi^{S,\mathbf{y},j}_b]}_\tau^2$ and $\sum_{a,b\in A}\norm{[p^x_a,p^y_b]}_\tau^2$ for a tracial state $\tau$ on $\Mor^{c-v}_\pi(G,A)$. First,
    \begin{align*}
        \sum_{a,b\in A}\norm{[\Phi^{R,\mathbf{x},i}_a,\Phi^{s,\mathbf{y},j}_b]}_\tau^2&=\sum_{a,b\in A}\norm{[p^{x_i}_a,p^{y_j}_b]+[\Phi^{R,\mathbf{x},i}_a-p^{x_i}_a,p^{y_j}_b]+[\Phi^{R,\mathbf{x},i}_a,\Phi^{S,\mathbf{y},j}_b-p^{y_j}_b]}_\tau^2\\
        &\leq 5\sum_{a,b\in A}\parens*{\norm{[p^{x_i}_a,p^{y_j}_b]}_\tau^2+2\norm{(\Phi^{R,\mathbf{x},i}_a-p^{x_i}_a)p^{y_j}_b}_\tau^2+2\norm{(\Phi^{S,\mathbf{y},j}_b-p^{y_j}_b)\Phi^{R,\mathbf{x},i}_a}_\tau^2}\\
        &=5\sum_{a,b\in A}\norm{[p^{x_i}_a,p^{y_j}_b]}_\tau^2+10\sum_{a\in A}\parens*{\norm{\Phi^{R,\mathbf{x},i}_a-p^{x_i}_a}_\tau^2+\norm{\Phi^{S,\mathbf{y},j}_a-p^{y_j}_a}_\tau^2}.
    \end{align*}
    Next,
    \begin{align*}
        \sum_{a\in A}\norm{\Phi^{R,\mathbf{x},i}_a-p^{x_i}_a}_\tau^2&=\sum_{a\in A}\parens*{\tau(p^{x_i}_a-\Phi^{R,\mathbf{x},i}_a)+2\tau(\Phi^{R,\mathbf{x},i}_a(1-p^{x_i}_a))}\\
        &=\tau(1-1)+2\sum_{\mathbf{a}\in R^A}\tau(\Phi^{R,\mathbf{x}}_{\mathbf{a}}(1-p^{x_i}_{a_i}))\\
        &=2\sum_{\mathbf{a}\in R^A}\norm{\Phi^{R,\mathbf{x}}_{\mathbf{a}}(1-p^{x_i}_{a_i})}_\tau^2.
    \end{align*}
    Hence, putting these together
    \begin{align*}
        \frac{1}{m^2}\sum_{\substack{R,\mathbf{x},i:\,x_i=x\\S,\mathbf{y},j:\,y_j=y\\a,b\in A}}\norm{[\Phi^{R,\mathbf{x},i}_a,\Phi^{S,\mathbf{y},j}_b]}_\tau^2&\leq\frac{5}{m^2}\sum_{\substack{R,\mathbf{x},i:\,x_i=x\\S,\mathbf{y},j:\,y_j=y\\a,b\in A}}\norm{[p^x_a,p^y_b]}_\tau^2+\frac{20}{m^2}\sum_{\substack{R,\mathbf{x},i:\,x_i=x\\S,\mathbf{y},j:\,y_j=y\\\mathbf{a}\in R^A}}\norm{\Phi^{R,\mathbf{x}}_{\mathbf{a}}(1-p^{x_i}_{a_i})}_\tau^2\\
        &\qquad\qquad+\frac{20}{m^2}\sum_{\substack{R,\mathbf{x},i:\,x_i=x\\S,\mathbf{y},j:\,y_j=y\\\mathbf{b}\in S^B}}\norm{\Phi^{S,\mathbf{y}}_{\mathbf{b}}(1-p^{y_j}_{b_j})}_\tau^2\\
        &\leq5L^2\sum_{a,b\in A}\norm{[p^x_a,p^y_b]}_\tau^2+\frac{40}{m}\sum_{\substack{R,\mathbf{x},i\\\mathbf{a}\in R^A}}\norm{\Phi^{R,\mathbf{x}}_{\mathbf{a}}(1-p^{x_i}_{a_i})}_\tau^2\\
        &\leq5L^2\sum_{a,b\in A}\norm{[p^x_a,p^y_b]}_\tau^2+40L\defect(\tau).
    \end{align*}
    On the other hand, by a similar argument,
    \begin{align*}
        \sum_{a,b\in A}\norm{[p^x_a,p^y_b]}_\tau^2&\leq5\sum_{a,b\in A}\norm{[\Phi^{R,\mathbf{x},i}_a,\Phi^{S,\mathbf{y},j}_b]}_\tau^2+10\sum_{a\in A}\parens*{\norm{\Phi^{R,\mathbf{x},i}_a-p^{x_i}_a}_\tau^2+\norm{\Phi^{S,\mathbf{y},j}_a-p^{y_j}_a}_\tau^2}\\
        &\leq 5\sum_{\substack{R,\mathbf{x},i:\,x_i=x\\S,\mathbf{y},j:\,y_j=y\\a,b\in A}}\norm{[\Phi^{R,\mathbf{x},i}_a,\Phi^{S,\mathbf{y},j}_b]}_\tau^2+40mL\defect(\tau).
    \end{align*}

    Now, suppose $(G,x,y)$ is a $c-c$-robust commutativity gadget. Let $\varepsilon>0$. Then, there exists $\delta>0$ such that for any trace $\tau$ on $\Mor^{c-c}_\pi(G,A)$ with $\defect(\tau)<\delta$, $\frac{1}{m^2}\sum_{\substack{R,\mathbf{x},i:\,x_i=x\\S,\mathbf{y},j:\,y_j=y\\a,b\in A}}\norm{[\Phi^{R,\mathbf{x},i}_a,\Phi^{S,\mathbf{y},j}_b]}_\tau^2<\frac{\varepsilon}{10m^2}$. Now, let $\delta'=\min\set*{\frac{\varepsilon}{80mL},\frac{\delta}{4L}}$. Let $\tau$ be a trace on $\Mor^{c-v}_\pi(G,A)$ such that $\defect(\tau)<\delta'$. Then, $\defect(\tau\circ\alpha)\leq 4L\defect(\tau)<4L\delta'\leq\delta$, so $\frac{1}{m^2}\sum_{\substack{R,\mathbf{x},i:\,x_i=x\\S,\mathbf{y},j:\,y_j=y\\a,b\in A}}\norm{[\Phi^{R,\mathbf{x},i}_a,\Phi^{S,\mathbf{y},j}_b]}_{\tau\circ\alpha}^2<\frac{\varepsilon}{10m^2}$. By the inequalities above and the fact that $\alpha$ acts as identity,
    \begin{align*}
        \sum_{a,b\in A}\norm{[p^x_a,p^y_b]}_\tau^2<5\frac{\varepsilon}{10}+40mL\defect(\tau)<\varepsilon,
    \end{align*}
    so $(G,x,y)$ is a $c-v$-robust commutativity gadget.

    Conversely, suppose that $(G,x,y)$ is a $c-v$-robust commutativity gadget. Let $\varepsilon>0$. There exists $\delta>0$ such that for any tracial state $\tau$ on $\Mor^{c-v}_\pi(G,A)$ such that $\defect(\tau)<\delta$, $\sum_{a,b\in A}\norm{[p^x_a,p^y_b]}_\tau^2<\frac{\varepsilon}{10}$. Let $\delta'=\min\set*{\frac{\varepsilon}{80L^3},\frac{\delta}{m}}$. Suppose $\tau$ is a tracial state on $\Mor^{c-c}_\pi(G,A)$ such that $\defect(\tau)<\delta'$. Then $\defect(\tau\circ\beta)< m\delta'\leq\delta$, so $\sum_{a,b\in A}\norm{[p^x_a,p^y_b]}_{\tau\circ\beta}^2\leq\frac{\varepsilon}{10L^2}$. Then, by the inequalities above,
    \begin{align*}
        \frac{1}{m^2}\sum_{\substack{R,\mathbf{x},i:\,x_i=x\\S,\mathbf{y},j:\,y_j=y\\a,b\in A}}\norm{[\Phi^{R,\mathbf{x},i}_a,\Phi^{S,\mathbf{y},j}_b]}_\tau^2&=\frac{1}{m^2}\sum_{\substack{R,\mathbf{x},i:\,x_i=x\\S,\mathbf{y},j:\,y_j=y\\a,b\in A}}\norm{[\Phi^{R,\mathbf{x},i}_a,\Phi^{S,\mathbf{y},j}_b]}_{\tau\circ\beta}^2<5L^2\frac{\varepsilon}{10L^2}+40L\defect(\tau)<\varepsilon,
    \end{align*}
    so $(G,x,y)$ is a $c-c$-robust commutativity gadget.
\end{proof}

Generally, as the above proof demonstrates, working with robust commutativity gadgets can be quite unwieldy. However, we can instead work with a stronger notion.

\begin{definition}
    Let $A$ be a relational structure over a signature $\sigma$. An \emph{algebraic commutativity gadget} for $A$ is a commutativity gadget $(G,x,y)$ such that for all $a,b\in A$, $[p_{xa},p_{yb}]=0$ in $\Mor^+(G,A)$. An \emph{oracular algebraic commutativity gadget} is a oracular commutativity gadget $(G,x,y)$ such that the same commutation relations hold in $\Mor^{o+}(G,A)$.
\end{definition}

\begin{lemma}
    Suppose $(G,x,y)$ is an algebraic commutativity gadget for $A$. Then it is an $a$-robust commutativity gadget.
\end{lemma}

\begin{proof}
    Since $[p_{xa},p_{yb}]=0$ in $\Mor^+(G,A)$ for each $a,b\in A$, we know that $[p^x_{a},p^y_{b}]$ is an infinitesimal element of the quotient of $\Mor^a_{\mbb{u}}(G,A)$ by $\ang*{\supp\mu_{a,\mbb{u}}}$, which is the ideal generated by the set of relations $\scr{R}=\set*{p^{x_1}_{a_1}\cdots p^{x_{\ar(R)}}_{a_{\ar(R)}}}{R\in\sigma,\mathbf{x}\in R^G,\mathbf{a}\notin R^A}$. Let $\varepsilon>0$. Then, $\abs*{[p_{xa},p_{yb}]}^2\leq\frac{\varepsilon}{2|A|^2}$ in $\Mor^a_{\mbb{u}}(G,A)/\ang*{\supp\mu_{a,\mbb{u}}}$, so $\sum_{a,b\in A}\abs*{[p_{xa},p_{yb}]}^2\leq\frac{\varepsilon}{2}$. Therefore, there exists a positive element $p\in\Mor^a_{\mbb{u}}(G,A)$ such that $\frac{\varepsilon}{2}-\sum_{a,b\in A}\abs*{[p_{xa},p_{yb}]}^2-p\in\ang*{\supp\mu_{a,\mbb{u}}}$. Then, we have that there exist $N\in\N$, $r_{i}\in \scr{R}$, and $v_{i},w_{i}\in\Mor^a_{\mbb{u}}(G,A)$ for $i\in[N]$ such that $\sum_{a,b\in A}\abs*{[p_{xa},p_{yb}]}^2+p-\frac{\varepsilon}{2}=\sum_{i=1}^Nv_ir_iw_i$. Since $\Mor^{a}_{\mbb{u}}(G,A)$ is archimedean with respect to the $\ast$-positive cone of sums-of-squares, there exist $\alpha_{i}\in \mbb{R}_{\geq 0}$ such that $|w_iv_i|^2\leq \alpha_{i}$. Now, let $\tau$ be a tracial state on $\Mor^a_\pi(G,A)$. Then,
    \begin{align*}
        \sum_{a,b\in A}\norm{[p^x_{a},p^y_{b}]}_\tau^2&=\frac{\varepsilon}{2}-\tau(p)+\sum_{i=1}^N\tau(v_ir_iw_i)\\
        &\leq\frac{\varepsilon}{2}+\sum_{i=1}^N\sqrt{\alpha_i}\norm{r_i}_\tau\\
        &\leq\frac{\varepsilon}{2}+\sum_{i=1}^N\sqrt{m\alpha_i\defect(\tau)},
    \end{align*}
    where $m=\sum_{R\in\sigma}|R^G|$. Write $C=\sum_{i=1}^N\sqrt{m\alpha_i}$. Now, let $\delta=\parens*{\frac{\varepsilon}{2C}}^2$. By the above, if $\tau$ is a tracial state such that $\defect(\tau)<\delta$, then $\sum_{a,b}\norm{[p^x_{a},p^y_{b}]}_\tau^2<\varepsilon$.
\end{proof}

\begin{lemma}
    Suppose $(G,x,y)$ is an oracular algebraic commutativity gadget for $A$. Then it is an $c-v$-robust commutativity gadget.
\end{lemma}

We omit the proof as it follows along exactly the same lines as the previous lemma.

\begin{question}
    Do there exist (oracular) commutativity gadgets or robust commutativity gadgets that are not algebraic?
\end{question}

\subsection{Relating oracular and non-oracular commutativity gadgets} \label{sec:rel-oracular}

\begin{lemma}\label{lem:classical-to-oracular}
    Suppose $(G,x,y)$ is a commutativity gadget (\textit{resp.}  $a$-robust commutativity gadget, algebraic commutativity gadget) for $A$ such that $\pi_{a,b}\in\mor^c(G,A)$ for all $a,b\in A$. Then, $(G,x,y)$ is an oracular commutativity gadget (\textit{resp.} $c-v$-robust commutativity gadget, oracular algebraic commutativity gadget).
\end{lemma}

This does not necessarily capture all possible examples of commutativity gadget, but it does capture all the examples we know.

\begin{question}
    Are there relational structures that admit a commutativity gadget but not one satisfying the conditions of \cref{lem:classical-to-oracular}? Are there relational structures that admit a commutativity gadget but no oracular commutativity gadget?
\end{question}

\begin{proof}[Proof of \cref{lem:classical-to-oracular}]
    First, if $(G,x,y)$ is a commutativity gadget with $\pi_{a,b}$ classical, then $\pi_{a,b}$ factors through $\Mor^{o+}(G,A)$. Hence, $(G,x,y)$ satisfies property (i) of an oracular commutativity gadget. Next, to show property (ii), suppose that $\pi\in\mor^{oqa}(G,A)$. Then, by precomposing with the quotient map $\Mor^+(G,A)\rightarrow\Mor^{o+}(G,A)$, this induces a representation of $\Mor^{+}(G,A)$ and by the commutativity gadget property, we have that $[\pi(p_{xa}),\pi(p_{yb})]=0$ for all $a,b\in A$. The same proof with $\pi$ replaced by the identity also shows that if $(G,x,y)$ is an algebraic commutativity gadget under the same assumption, then it is also an oracular algebraic commutativity gadget.

    Now, suppose $(G,x,y)$ is an $a$-robust commutativity gadget. Let $\tau$ be a tracial state on $\Mor^{c-v}_{\mbb{u}}(G,A)$. By precomposing with the inclusion map $\Mor_{\mbb{u}}^a(G,A)\rightarrow\Mor^{c-v}_{\mbb{u}}(G,A)$, $\tau$ induces a tracial state on $\Mor_{\mbb{u}}^a(G,A)$, which we refer to as $\tau'$. We have that
    \begin{align*}
        \defect(\tau')&=\sum_{R\in\sigma,\mathbf{a}\in R^A,\mathbf{b}\notin R^B}{\mbb{u}}(R,\mathbf{a})\norm{p^{a_1}_{b_1}\cdots p^{a_{\ar(R)}}_{b_{\ar(R)}}}_\tau^2\\
        &=\sum_{R\in\sigma,\mathbf{a}\in R^A,\mathbf{b}\notin R^B}{\mbb{u}}(R,\mathbf{a})\big\|(p^{a_1}_{b_1}-\Phi^{R,\mathbf{a},1}_{b_1})p^{a_2}_{b_2}\cdots p^{a_{\ar(R)}}_{b_{\ar(R)}}+\Phi^{R,\mathbf{a},1}_{b_1}(p^{a_2}_{b_2}-\Phi^{R,\mathbf{a},2}_{b_2})p^{a_3}_{b_3}\cdots p^{a_{\ar(R)}}_{b_{\ar(R)}}\\
        &\hspace{4.5cm}+\ldots+\Phi^{R,\mathbf{a},1}_{b_1}\cdots\Phi^{R,\mathbf{a},\ar(R)-1}_{b_{\ar(R)-1}}(p^{a_{\ar(R)}}_{b_{\ar(R)}}-\Phi^{R,\mathbf{a},\ar(R)}_{b_{\ar(R)}})+\Phi^{R,\mathbf{a}}_{\mathbf{b}}\big\|_\tau^2\\
        &\leq \sum_{R\in\sigma,\mathbf{a}\in R^A,b\in B}{\mbb{u}}(R,\mathbf{a})\ar(R)\sum_{i=1}^{\ar(R)}\norm{p^{a_i}_{b}-\Phi^{R,\mathbf{a},i}_{b}}_\tau^2\\
        &=2\sum_{R\in\sigma,\mathbf{a}\in R^A,\mathbf{b}\in R^B,i\in[\ar(R)]}{\mbb{u}}(R,\mathbf{a})\ar(R)\norm{\Phi^{R,\mathbf{a}}_{\mathbf{b}}(1-p^{a_i}_{b_i})}_\tau^2\leq2L^2\defect(\tau),
    \end{align*}
    where $L=\max_{R\in\sigma}\ar(R)$. Now, let $\varepsilon>0$. As $(G,x,y)$ is an $a$-robust commutativity gadget, there exists $\delta>0$ such that if $\varphi$ is a tracial state on $\Mor^a_{\mbb{u}}(G,A)$ such that $\defect(\varphi)<\delta$, then $\sum_{a,b\in A}\norm{[p^x_{a},p^y_{b}]}_\varphi^2<\varepsilon$. Suppose that $\tau$ is a tracial state on $\Mor^{c-v}_{\mbb{u}}(G,A)$ such that $\defect(\tau)<\frac{\delta}{2L^2}$. By the above, $\tau'$ is a tracial state on $\Mor^{a}_{\mbb{u}}(G,A)$ such that $\defect(\tau')<\delta$, and hence $\sum_{a,b\in A}\norm{[p^x_{a},p^y_{b}]}_\tau^2=\sum_{a,b\in A}\norm{[p^x_{a},p^y_{b}]}_{\tau'}^2<\varepsilon$. So $(G,x,y)$ is a $c-v$-robust commutativity gadget.
\end{proof}

\begin{definition}\label{def:oracularisable}
    A relational structure $A$ over a signature $\sigma$ is called \emph{oracularisable} if $\mor^+(B,A)=\mor^{o+}(B,A)$ for all $B$ over $\sigma$.
\end{definition}

\begin{lemma}
    Let $A$ be an oracularisable relational structure. Then $(G,x,y)$ is an oracular (algebraic) commutativity gadget iff $(G,x,y)$ is a (algebraic) commutativity gadget.
\end{lemma}

\begin{proof}
    The proof follows directly from the fact that every $\ast$-representation of $\Mor^+(G,A)$ is a $\ast$-representation of $\Mor^{o+}(G,A)$ and vice versa.
\end{proof}

\subsection{Commutativity gadgets and complexity}
\label{sec:comm-gadgets-complexity}

In this section, we relate the existence of commutativity gadgets to undecidability of entangled CSPs. In particular, we show that if a CSP is classically $\tsf{NP}$-complete and admits a commutativity gadget, then the entangled version is undecidable. This result was shown in the special case of boolean CSPs in~\cite{CM24}.

\begin{theorem}\label{thm:cm24}
    Suppose that $A$ is a relational structure such that $\CSP(A)$ is $\tsf{NP}$-complete.
    \begin{itemize}
        \item If $A$ has an $a$-robust commutativity gadget, then there exists $s<1$ such that there is a polynomial-time reduction from the halting problem to $\SuccinctCSP_a(A)_{1,s}^\ast$.

        \item If $A$ has a $c-v$-robust commutativity gadget, then there exists $s<1$ such that there is a polynomial-time reduction from the halting problem to $\SuccinctCSP_{c-v}(A)_{1,s}^\ast$.

        \item If $A$ has a $c-c$-robust commutativity gadget, then there exists $s<1$ such that there is a polynomial-time reduction from the halting problem to $\SuccinctCSP_{c-c}(A)_{1,s}^\ast$.
    \end{itemize}
\end{theorem}

\begin{proof}
    We show the case that $A$ has a $c-v$-robust commutativity gadget here. The other two cases are analogous, using instead of the constraint-variable algebra the constraint-constraint algebra in the case of the $c-c$-robust commutativity gadget and the $a+comm$ algebra (see~\cite{CM24} for the definition) in the case of the $a$-robust commutativity gadget.

    Let $\sigma'=\sigma\cup\{E\}$ with $\ar(E)=2$; and let $A'$ be the relational structure with alphabet $A$, and constraints $R^{A'}=R^A$ for $R\in\sigma$ and $E^{A'}=A\times A$. Due to the addition of the relation $E$, $A'$ is a non-TVF relational structure so, by~\cite{CM24}, there exists $s'<1$ such that there is a polynomial-time reduction from the halting problem to $\SuccinctCSP_{c-v}(A')_{1,s'}^\ast$. Let $(G,x,y)$ be the $c-v$-robust commutativity gadget, and let $\delta:[0,1]\rightarrow[0,1]$ be a convex function such that $\delta(\varepsilon)\leq\varepsilon$ and if $\defect(\tau)<\delta(\varepsilon)$ for a finite-dimensional tracial state on $\Mor_{\mbb{u}}^{c-v}(G,A)$, then $\sum_{a,b\in A}\norm{[p^x_a,p^y_b]}_\tau<\varepsilon$ (we know such a modulus function $\delta$ exists by taking any modulus function $\delta_0$ and letting $\delta$ be the function whose epigraph is the convex hull of the epigraph of $\varepsilon\mapsto\min\{\delta_0(\varepsilon),\varepsilon\}$). Let $m=\sum_{R\in\sigma}|R^G|$.
    
    Now, suppose $M'$ is an instance of $\SuccinctCSP_{c-v}(A')_{1,s'}^\ast$, that samples from a probability distribution $\pi'$ on $\bigcup_{R\in\sigma'}\{R\}\times R^{B'}$ for a relational structure $B'$ over $\sigma'$. Now, let $B$ be the relational structure over $\sigma$ with alphabet $B'\cup\set*{z^{a,b}}{(a,b)\in E^{B'},z\in G}/\sim$, where $\sim$ is the equivalence relation given by $x^{a,b}\sim a$ and $y^{a,b}\sim b$ for all $(a,b)\in E^{B'}$, which we use implicitly; and relations $R^B=R^{B'}\cup\set*{\mathbf{z}^{a,b}}{R\in\sigma,\,\mathbf{z}\in R^G,\,(a,b)\in E^{B'}}$ (where $\mathbf{z}^{a,b}=(z_1^{a,b},\ldots,z_{\ar(R)}^{a,b})$). Then, define the probability distribution $\pi$ on $\bigcup_{R\in\sigma}\{R\}\times R^B$ as $\pi(R,\mathbf{x})=\pi'(R,\mathbf{x})$ for $R\in\sigma$, and $\pi(R,\mathbf{z}^{a,b})=\frac{\pi(E,(a,b))}{m}$ for all $(a,b)\in E$, $R\in\sigma$, and $\mathbf{z}\in R^G$. Therefore, there exists an instance $M$ of $\SuccinctCSP_{c-v}(A)_{1,s}$ that samples $\pi$.

    Next, any representation $\varphi'\in\mor^{oqa}(B',A')$ induces a representation $\varphi\in\mor^{oqa}(B,A)$. To do so, take $\varphi(p_{ba})=\varphi'(p_{ba})\otimes 1$ and $\varphi(p_{z^{b_1,b_2}a})=\sum_{a_1,a_2\in A}\varphi'(p_{b_1a_1}p_{b_2a_2})\otimes\pi_{a_1,a_2}(p_{za})$, where we assume without loss of generality that the $\pi_{a_1,a_2}$ all act on the same Hilbert space. This is well-defined on the generators as, for $(b_1,b_2)\in E^{B'}$, 
    \begin{align*}
    \varphi(p_{x^{b_1,b_2}a})&=\sum_{a_1,a_2\in A}\varphi'(p_{b_1a_1}p_{b_2a_2})\otimes\pi_{a_1,a_2}(p_{xa})\\
    &=\sum_{a_1,a_2\in A}\varphi'(p_{b_1a_1}p_{b_2a_2})\otimes\delta_{a_1,a}1\\
    &=\sum_{a_2\in A}\varphi'(p_{b_1a}p_{b_2a_2})\otimes1\\
    &=\varphi'(p_{b_1a})\otimes 1=\varphi(p_{b_1a}),
    \end{align*}and identically $\varphi'(p_{y^{b_1,b_2}a})=\varphi(p_{b_2a})$. $\varphi$ also extends to a representation of $\Mor^+(B,A)$, as it satisfies all the relations. First, $\varphi(p_{ba})$ and $\varphi(p_{z^{b_1,b_2}a})$ are projections by construction, and
    \begin{align*}
        &\sum_{a\in A}\varphi(p_{ba})=\sum_{a\in A}\varphi'(p_{ba})\otimes 1=1\\
        &\sum_{a\in A}\varphi'(p_{z^{b_1,b_2}a})=\sum_{a,a_1,a_2\in A}\varphi'(p_{b_1a_1}p_{b_2a_2})\otimes\pi_{a_1,a_2}(p_{xa})=\sum_{a_1,a_2\in A}\varphi'(p_{b_1a_1}p_{b_2a_2})\otimes1=1.
    \end{align*}
    Next, note that by construction, if $\mathbf{b}\in R^{B}$, either $\mathbf{b}=\mathbf{z}^{b,b'}$ for some $(b,b')\in E^{B'}$, or $b_i\in B'$ for all $i$. For both cases, taking $\mathbf{a}\notin R^A$, we get, writing $n=\ar(R)$, $\varphi(p_{b_1a_1})\cdots\varphi(p_{b_na_n})=\varphi'(p_{b_1a_1})\cdots\varphi'(p_{b_na_n})\otimes 1=0$ and
    \begin{align*}
        \varphi(p_{z^{b,b'}_1a_1})\cdots\varphi(p_{z^{b,b'}_na_n})&=\hspace{-0.5cm}\sum_{a^1_1,a^2_1,\ldots,a^1_n,a^2_n\in A}\hspace{-0.5cm}\varphi'(p_{ba^1_{1}}p_{b'a^2_1}\cdots p_{ba^1_n}p_{b'a^2_n})\otimes\pi_{a^1_1,a^2_1}(p_{z_1a_1})\cdots\pi_{a^1_n,a^2_n}(p_{z_na_n})\\
        &=\sum_{a,a'\in A}\varphi'(p_{ba}p_{b'a'})\otimes\pi_{a,a'}(p_{z_1a_1}\cdots p_{z_na_n})=0.
    \end{align*}
    Finally, for any $a,a'\in A$ and $1\leq i,j\leq n$, $[\varphi(p_{b_ia}),\varphi(p_{b_ja'})]=[\varphi(p_{b_ia}),\varphi(p_{b_ja'})]\otimes 1=0$ and
    \begin{align*}
        [\varphi(p_{z^{b,b'}_ia}),\varphi(p_{z^{b,b'}_ja'})]&=\sum_{a_1,a_2,a_3,a_4\in A}[\varphi'(p_{ba_1}p_{b'a_2})\otimes\pi_{a_1,a_2}(p_{z_ia}),\varphi'(p_{ba_3}p_{b'a_4})\otimes\pi_{a_3,a_4}(p_{z_ia'})]\\
        &=\sum_{a_1,a_2\in A}\varphi'(p_{ba_1}p_{b'a_2})\otimes[\pi_{a_1,a_2}(p_{z_ia}),\pi_{a_1,a_2}(p_{z_ja'})]=0.
    \end{align*}
    Hence, we have that if $M'$ is a yes instance of $\SuccinctCSP_{c-v}(A')_{1,s'}$, it is mapped to a yes instance of $\SuccinctCSP_{c-v}(A)_{1,s}$ (for any $s$).
    
    Finally, let $\tau$ be a finite-dimensional trace on $\Mor^{c-v}_\pi(B,A)$. This has a finite-dimensional GNS representation $(\varphi,\ket{\psi})$. For every $(b_1,b_2)\in E^{B'}$, let $$\varepsilon_{b_1,b_2}=1-\sum_{a_1,a_2\in A}\braket{\psi}{\varphi(p^{b_1}_{a_1}p^{b_2}_{a_2}p^{b_1}_{a_1})^2}{\psi}=\sum_{a_1,a_2\in A}\tau(p^{b_1}_{a_1}p^{b_2}_{a_2}p^{b_1}_{a_1}p^{b_2}_{a_2}).$$
    Due to \cite{dlS22}, there exists a PVM $\{\Pi^{b_1,b_2}_{a_1,a_2}\}_{a_1,a_2\in A}$ such that $$\sum_{a_1,a_2\in A}\braket{\psi}{(\varphi(p^{b_1}_{a_1}p^{b_2}_{a_2}p^{b_1}_{a_1})-\Pi^{b_1,b_2}_{a_2,a_2})^2}{\psi}\leq9\varepsilon_{b_1,b_2}.$$
    Take $\varphi'$ to be the representation of $\Mor^{c-v}_{\pi'}(B',A')$ induced by $\varphi'(p^b_a)=\varphi(p^b_a)$, $\varphi'(\Phi^{R,\mathbf{b}}_{\mathbf{a}})=\varphi(\Phi^{R,\mathbf{b}}_{\mathbf{a}})$ for $R\in\sigma$ and $\varphi'(\Phi^{E,b_1,b_2}_{a_1a_2})=\Pi^{b_1,b_2}_{a_1,a_2}$. This induces a trace $\tau'$ on $\Mor^{c-v}_{\pi'}(B',A')$ as $\tau'(x)=\braket{\psi}{\varphi'(x)}{\psi}$. Then, we have that
    \begin{align*}
        \defect(\tau')&=\hspace{-5mm}\sum_{\substack{R\in\sigma,i\in\ar(R)\\\mathbf{b}\in R^{B'},\mathbf{a}\in R^A}}\frac{\pi(R,\mathbf{b})}{\ar(R)}\norm{\Phi^{R,\mathbf{b}}_{\mathbf{a}}(1-p^{b_i}_{a_i})}_\tau^2+\hspace{-5mm}\sum_{\substack{(b_1,b_2)\in E^{B'}\\(a_1,a_2)\in A^2}}\hspace{-5mm}\frac{\pi(E,b_1,b_2)}{2}\braket{\psi}{\Pi^{b_1,b_2}_{a_1,a_2}((1-\varphi(p^{b_1}_{a_1}))+(1-\varphi(p^{b_2}_{a_2})))}{\psi},
    \end{align*}
    Since the state $\ket{\psi}$ is tracial on $\im\varphi'$, we have that for $i=1,2$,
    \begin{align*}
        \sum_{a_1,a_2}&\braket{\psi}{\Pi^{b_1,b_2}_{a_1,a_2}(1-\varphi(p^{b_i}_{a_i}))}{\psi}=\sum_{a_1,a_2}\braket{\psi}{\abs*{\Pi^{b_1,b_2}_{a_1,a_2}(1-\varphi(p^{b_i}_{a_i}))}^2}{\psi}\\
        &\leq2\sum_{a_1,a_2}\parens*{\braket{\psi}{\abs*{\varphi(p^{b_1}_{a_1}p^{b_2}_{a_2}p^{b_1}_{a_1})(1-\varphi(p^{b_i}_{a_i}))}^2}{\psi}+\braket{\psi}{(\varphi(p^{b_1}_{a_1}p^{b_2}_{a_2}p^{b_1}_{a_1})-\Pi^{b_1,b_2}_{a_2,a_2})^2}{\psi}}\\
        &\leq2\sum_{a_1,a_2}\norm{p^{b_1}_{a_1}p^{b_2}_{a_2}p^{b_1}_{a_1}(1-p^{b_i}_{a_i})}_\tau^2+18\varepsilon_{b_1,b_2}
    \end{align*}
    If $i=1$, the first term is $0$. Else, we have that
    \begin{align*}
        \sum_{a_1,a_2}\norm{p^{b_1}_{a_1}p^{b_2}_{a_2}p^{b_1}_{a_1}(1-p^{b_2}_{a_2})}_\tau^2&=\sum_{a_1,a_2}\norm{p^{b_1}_{a_1}p^{b_2}_{a_2}(p^{b_2}_{a_2}p^{b_1}_{a_1}-p^{b_1}_{a_1}p^{b_2}_{a_2})}_\tau^2\leq\sum_{a_1,a_2}\norm{[p^{b_1}_{a_1},p^{b_2}_{a_2}]}_\tau^2.
    \end{align*}
    Since
    \begin{align*}
        \sum_{a_1,a_2}\norm{[p^{b_1}_{a_1},p^{b_2}_{a_2}]}_\tau^2&=\sum_{a_1,a_2}\tau((p^{b_1}_{a_1}p^{b_2}_{a_2}-p^{b_2}_{a_2}p^{b_1}_{a_1})(p^{b_2}_{a_2}p^{b_1}_{a_1}-p^{b_1}_{a_1}p^{b_2}_{a_2}))\\
        &=2\sum_{a_1,a_2}\parens*{\tau(p^{b_1}_{a_1}p^{b_2}_{a_2})-\tau(p^{b_1}_{a_1}p^{b_2}_{a_2}p^{b_1}_{a_1}p^{b_2}_{a_2})}\\
        &=2\varepsilon_{b_1,b_2}
    \end{align*}
    we can collect terms and get that
    \begin{align*}
        \sum_{a_1,a_2}\braket{\psi}{\Pi^{b_1,b_2}_{a_1,a_2}((1-\varphi(p^{b_1}_{a_1}))+(1-\varphi(p^{b_2}_{a_2})))}{\psi}&\leq20\sum_{a_1,a_2}\norm{[p^{b_1}_{a_1},p^{b_2}_{a_2}]}_\tau^2.
    \end{align*}
    Hence, the defect is bounded as
    \begin{align*}
        \defect(\tau')\leq\sum_{\substack{R\in\sigma,i\in\ar(R)\\\mathbf{b}\in R^{B'},\mathbf{a}\in R^A}}\frac{\pi(R,\mathbf{b})}{\ar(R)}\norm{\Phi^{R,\mathbf{b}}_{\mathbf{a}}(1-p^{b_i}_{a_i})}_\tau^2+10\sum_{\substack{(b_1,b_2)\in E^{B'}\\(a_1,a_2)\in A^2}}\pi(E,b_1,b_2)\norm{[p^{b_1}_{a_1},p^{b_2}_{a_2}]}_\tau^2.
    \end{align*}
    By the commutativity gadget property, $$\sum_{a_1,a_2}\norm{[p^{b_1}_{a_1},p^{b_2}_{a_2}]}_\tau^2\leq\delta^{-1}\parens[\Bigg]{\frac{1}{m}\sum_{\substack{R\in\sigma,i\in\ar(R)\\\mathbf{z}\in R^G,\mathbf{a}\in R^A}}\frac{1}{\ar(R)}\norm{\Phi^{R,\mathbf{z}^{b_1,b_2}}_{\mathbf{a}}(1-p^{z^{b_1,b_2}_i}_{a_i})}_\tau^2}.$$
    Therefore, using Jensen's inequality,
    \begin{align*}
        \defect(\tau')&\leq\sum_{\substack{R\in\sigma,i\in\ar(R)\\\mathbf{b}\in R^{B'},\mathbf{a}\in R^A}}\frac{\pi(R,\mathbf{b})}{\ar(R)}\norm{\Phi^{R,\mathbf{b}}_{\mathbf{a}}(1-p^{b_i}_{a_i})}_\tau^2\\
        &\qquad+10\sum_{(b_1,b_2)\in E^{B'}}\pi(E,b_1,b_2)\delta^{-1}\parens[\Bigg]{\frac{1}{m}\sum_{\substack{R\in\sigma,i\in\ar(R)\\\mathbf{z}\in R^G,\mathbf{a}\in R^A}}\frac{1}{\ar(R)}\norm{\Phi^{R,\mathbf{z}^{b_1,b_2}}_{\mathbf{a}}(1-p^{z^{b_1,b_2}_i}_{a_i})}_\tau^2}\\
        &\leq10\delta^{-1}\parens[\Bigg]{\sum_{\substack{R\in\sigma,i\in\ar(R)\\\mathbf{b}\in R^{B},\mathbf{a}\in R^A}}\frac{\pi(R,\mathbf{b})}{\ar(R)}\norm{\Phi^{R,\mathbf{b}}_{\mathbf{a}}(1-p^{b_i}_{a_i})}_\tau^2}=10\delta^{-1}(\defect(\tau)).
    \end{align*}
    Using this, we know that if $M'$ is a no instance of $\SuccinctCSP_{c-v}(A')^\ast_{1,s'}$, then $\defect(\tau')>1-s'$ for any trace $\tau'$ on $\Mor_{\pi'}^{c-v}(B',A')$. But then, we must have that $\defect(\tau)>\delta(\frac{1-s'}{10})$ for any trace $\tau$ on $\Mor_\pi^{c-v}(B,A)$, by the above upper bound. Hence, taking $s=1-\delta(\frac{1-s'}{10})$, $M$ is a no instance of $\SuccinctCSP_{c-v}(A)_{1,s}^\ast$, completing the proof.
\end{proof}

\subsection{A no-go theorem}\label{sec:nogo}

\begin{theorem}\label{thm:nogo}
    Let $A$ be a relational structure over a signature $\sigma$
    \begin{enumerate}[(i)]
        \item If $\enmo^{qa}(A)\neq\enmo^{c}(A)$, then there is no commutativity gadget for $A$.

        \item If $\enmo^{oqa}(A)\neq\enmo^{c}(A)$, then there is no oracular commutativity gadget for $A$.
    \end{enumerate}
\end{theorem}

\begin{proof}
    By the hypothesis of (i), there exist $\pi_0\in\enmo^{qa}(A)$ and $a_1,b_1,a_2,b_2\in A$ such that $\pi_0(p_{a_1b_1})$ and $\pi_0(p_{a_2b_2})$ do not commute. Now suppose, for a contradiction, that $A$ admits a commutativity gadget $(G,x,y)$. In particular, this implies that there exists $\pi_{a_1,a_2}\in\mor^{qa}(G,A)$ such that $\pi_{a_1,a_2}(p_{xa_1})=\pi_{a_1,a_2}(p_{ya_2})=1$. Then, composing the morphisms gives $\pi:=\pi_0\circ\pi_{a_1,a_2}\in\mor^{qa}(G,A)$. This satisfies
    \begin{align*}
        &\pi(p_{xb_1})=\sum_{a\in A}\pi_{a_1,a_2}(p_{xa})\otimes\pi_0(p_{ab_1})=1\otimes\pi_0(p_{a_1b_1}),\\
        &\pi(p_{yb_2})=\sum_{a\in A}\pi_{a_1,a_2}(p_{ya})\otimes\pi_0(p_{ab_2})=1\otimes\pi_0(p_{a_2b_2}).
    \end{align*}
    Hence, by construction, $\pi(p_{xb_1})$ and $\pi(p_{yb_2})$ do not commute, giving the wanted contradiction.

    For (ii), note that the same proof holds with $qa$ replaced with $oqa$.
\end{proof}

\section{Examples and constructions} \label{sec:examples}

Graphs provide a wide variety of well-studied examples of CSPs. In this section, we study existence and nonexistence (using \cref{thm:nogo}) of commutativity gadgets for a variety of graph CSPs.

In \cref{sec:k-col}, we show that $k$-colouring does not admit a commutativity gadget for $k\geq 4$, but that it does in the oracular model. In \cref{sec:schmidt}, we prove a criterion that allows us to find non-classical endomorphisms, and use this to find examples of graphs with no commutativity gadget. In \cref{sec:classical-endo}, we study the quantum endomorphism monoids of some families of graphs, and show that odd cycles and odd graphs admit only classical endomorphisms. In \cref{sec:graph-oracularisability}, we show that graph CSPs are oracularisable if and only if they do not admit a four-cycle.

\subsection{The case of \texorpdfstring{$k$}{\textit{k}}-colouring}\label{sec:k-col}

The problem of $k$-colouring a graph is the CSP induced by the complete graph $K_k$.

\begin{proposition}
    For $k\geq 4$, the complete graph $K_k$ does not admit a commutativity gadget.
\end{proposition}

\begin{proof}
    By definition, $\aut^+(K_k)=S^+_k$. It is a standard result in the theory of quantum groups that $S_k^+$ is noncommutative for $k\geq 4$~\cite{Wan98}. Further, we can explicitly show that in fact $\End^+(K_k)$ has a finite-dimensional noncommutative representation. It is clear that the following map extends to a noncommutative representation $\pi:\End^+(A)\rightarrow B(\C^2)$: $\pi(p_{00})=\pi(p_{11})=\ketbra{0}$, $\pi(p_{01})=\pi(p_{10})=\ketbra{1}$, $\pi(p_{22})=\pi(p_{33})=\ketbra{+}$, $\pi(p_{23})=\pi(p_{32})=\ketbra{-}$, and $\pi(p_{ij})=0$ otherwise. As such $\enmo^{qa}(K_k)\neq\enmo^c(K_k)$, so by \cref{thm:nogo}, there is no commutativity gadget for $K_k$.
\end{proof}

However, this no-go result does not extend to the oracular model.

\begin{proposition}\label{prop:k-col-yesgo}
    For all $k\in\N$, the complement of the cycle graph $(\overline{C}_{2k},0,1)$ is an oracular algebraic commutativity gadget for $K_k$.
\end{proposition}

Note that for $k=3$, $(\overline{C}_{6},0,1)$ is Ji's triangular prism gadget~\cite{Ji13}. See \cref{fig:gadgetk4} for a visualisation.

\begin{figure}
    \centering
    \begin{tikzpicture}[3d view={30}{10}]
        \draw (0,0,-2) -- (2.31,0,-2) -- (1.15,0,0) -- (0,0,-2) (0,8,-2) -- (2.31,8,-2) -- (1.15,8,0) -- (0,8,-2) (0,0,-2) -- (0,8,-2) (2.31,0,-2) -- (2.31,8,-2) (1.15,0,0) -- (1.15,8,0);
        \fill (1.15,0,0) circle (3pt) node[above left] {$0$};
        \fill (1.15,8,0) circle (3pt) node[above right] {$3$};
        \fill (0,0,-2) circle (3pt) node[below left] {$2$};
        \fill (0,8,-2) circle (3pt) node[above left] {$5$};
        \fill (2.31,0,-2) circle (3pt) node[below right] {$4$};
        \fill (2.31,8,-2) circle (3pt) node[below right] {$1$};
    \end{tikzpicture}
    \begin{tikzpicture}[3d view={30}{10}]
        \draw (0,0,0) -- (0,0,-2) -- (2,0,-2) -- (2,0,0) -- (0,0,0) -- (2,0,-2) (2,0,0) -- (0,0,-2) (0,8,0) -- (0,8,-2) -- (2,8,-2) -- (2,8,0) -- (0,8,0) -- (2,8,-2) (2,8,0) -- (0,8,-2) (0,8,0) -- (0,0,-2) -- (0,8,-2) -- (2,0,-2) -- (2,8,-2) -- (2,0,0) -- (2,8,0) -- (0,0,0) -- (0,8,0);
        \fill (0,0,0) circle (3pt) node[above left] {$0$};
        \fill (0,8,0) circle (3pt) node[above right] {$3$};
        \fill (0,0,-2) circle (3pt) node[below left] {$6$};
        \fill (0,8,-2) circle (3pt) node[below] {$1$};
        \fill (2,0,-2) circle (3pt) node[below right] {$4$};
        \fill (2,8,-2) circle (3pt) node[below right] {$7$};
        \fill (2,0,0) circle (3pt) node[above] {$2$};
        \fill (2,8,0) circle (3pt) node[below right] {$5$};
    \end{tikzpicture}
\caption{The graphs $\overline{C}_6$ and $\overline{C}_8$ presented as prisms}
\label{fig:gadgetk4}
\end{figure}
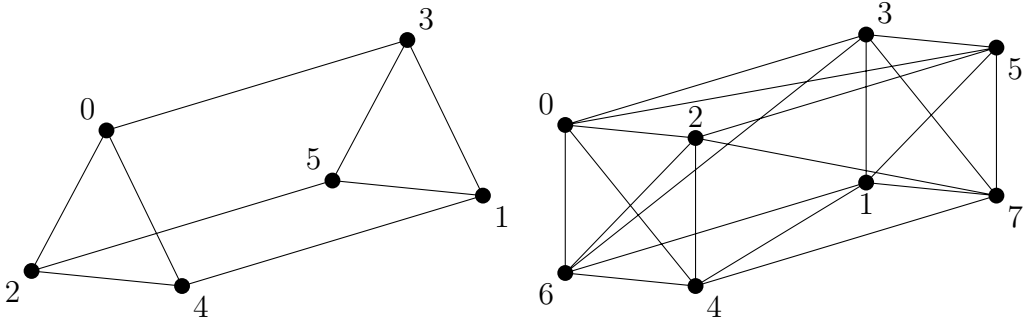

\begin{lemma}\label{lem:comm-core}
    $\End^{o+}(K_k)$ is commutative and $\sum_{a\in K_k}p_{ab}=1$ for all $b\in\Z_k$.
\end{lemma}

\begin{proof}
    Every pair of distinct vertices in $K_k$ is connected by an edge, hence by construction every pair of generators in $\End^{o+}(K_k)$ commute. So $\End^{o+}(K_k)$ is a commutative $C^\ast$ algebra, proving the first part of the lemma.

    Next, consider the following:
    \begin{align*}
        \sum_{b\in K_k}\parens*{1-\sum_{a\in K_k}p_{ab}}^2&=\sum_{b\in K_k}\parens*{1-2\sum_{a\in K_k}p_{ab}+\sum_{a,a'\in K_k}p_{ab}p_{a'b}}\\
        &=k+\sum_{b\in K_k}\parens*{-\sum_{a\in K_k}p_{ab}+\sum_{a\sim_{K_k} a'}p_{ab}p_{a'b}}\\
        &=k-\sum_{a\in K_k}\sum_{b\in K_k}p_{ab}=k-k=0.
    \end{align*}
    As $C^\ast$-algebras are hereditary, $\sum_{a\in K_k}p_{ab}=1$ for all $b$.
\end{proof}

\begin{lemma}\label{lem:ratcheting}
    In $\Mor^{o+}(\overline{C}_{2k}, K_k)$, $[p_{(2a)b},p_{(2a+1)c}]=[p_{0b},p_{1c}]$ for all $a,b,c\in\Z_k$.
\end{lemma}

\begin{proof}
    For all $i\neq j$, $2i+1\sim_{\overline{C}_{2k}}2j+1$. As such, there is a $\ast$-homomorphism $\End^{o+}(K_k)\rightarrow \Mor^{o+}(\overline{C}_{2k},K_k)$ defined on the generators as $p_{ab}\mapsto p_{(2a+1)b}$. Then, by \cref{lem:comm-core}, the $p_{(2a+1)b}$ commute and $\sum_{a\in\Z_k}p_{(2a+1)b}=1$. This implies that $\sum_{a\in\Z_k}[p_{0b},p_{(2a+1)c}]=[p_{0b},1]=0$. For all $a\neq 0,-1$, $0\sim_{\overline{C}_{2k}}2a+1$, giving $[p_{0b},p_{(2a+1)c}]=0$. Hence, $[p_{0b},p_{1c}]=-[p_{0b},p_{(-1)c}]$.

    In the same way, for all $i\neq j$, $2i\sim_{\overline{C}_{2k}}2j$, so $\sum_{a\in\Z_k}p_{(2a)b}=1$. Taking the commutator with $p_{(-1)c}$ gives $\sum_{a\in\Z_k}[p_{(2a)b},p_{(-1)c}]=0$. As $2a\sim_{\overline{C}_{2k}}-1$ for all $a\neq0,-1$, this implies that $[p_{0b},p_{(-1)c}]=-[p_{(-2)b},p_{(-1)c}]$. Putting these together gives $[p_{0b},p_{1c}]=[p_{(-2)b},p_{(-1)c}]$. Proceeding inductively gives the result.
\end{proof}

\begin{proof}[Proof of \cref{prop:k-col-yesgo}]
    We need to show that $(\overline{C}_{2k},0,1)$ satisfies the two conditions of an oracular algebraic commutativity gadget. For the first condition, consider the maps $f,g:\overline{C}_{2k}\rightarrow K_k$ defined as $f(2a)=f(2a+1)=a$ and $g(2a)=g(2a-1)=a$. These are graph homomorphisms as the vertices of $\overline{C}_{2k}$ that map to the same vertex of $K_k$ are not connected by an edge. Now, let $\sigma_a$ be a permutation of $\Z_k$ that maps $0$ to $a$, and for $a\neq b$ let $\sigma_{a,b}$ be a permutation of $\Z_k$ that maps $0$ to $a$ and $1$ to $b$. Now, for any $a,b\in K_k$, we take $\pi_{a,b}$ as the classical homomorphism $\sigma_a\circ f$ if $a=b$, and $\sigma_{a,b}\circ g$ if $a\neq b$.

    To show the second condition, we need that $[p_{0b},p_{1c}]=0$ for all $b,c\in K_k$. In fact,
    \begin{align*}
        [p_{0b},p_{1c}]&=\sum_{a\in\Z_k}p_{(2a)b}[p_{0b},p_{1c}]\\
        &=\sum_{a\in\Z_k}p_{(2a)b}[p_{(2a+2)b},p_{(2a+3)c}]\\
        &=\sum_{a\in\Z_k}[p_{(2a)b}p_{(2a+2)b},p_{(2a+3)c}]\\
        &=\sum_{a\in\Z_k}[0,p_{(2a-1)c}]=0,
    \end{align*}
    where the first line follows from $\sum_{a\in\Z_k}p_{(2a)b}=1$, the second line follows from \cref{lem:ratcheting}, the third line follows from the fact that $2a\sim_{\overline{C}_{2k}}2a+3$ and hence $p_{(2a)b}$ commutes with $p_{(2a+3)c}$, and the last line follows from the $k$-colouring condition.
\end{proof}

\begin{corollary}
    For every $k\geq 3$, there exists $s<1$ such that there is a polynomial-time reduction from the halting problem to $\SuccinctCSP_{c-c}(K_k)_{1,s}^\ast$ and $\SuccinctCSP_{c-v}(K_k)_{1,s}^\ast$.
\end{corollary}

The corollary follows by using \cref{prop:k-col-yesgo} in \cref{thm:cm24}.

\begin{question}
    Is non-oracular entangled $k$-colouring also undecidable for all $k$?
\end{question}

\subsection{Constructing graphs with no commutativity gadget}\label{sec:schmidt}

In this section, we work out a technique for showing that graphs have nonclassical quantum endomorphisms, and use this to construct small examples of graphs admitting no commutativity gadget. To do so, we generalise a result of Schmidt~\cite{Sch20}, used to show that quantum automorphism groups of graphs are nonclassical.

\begin{theorem}[Schmidt criterion \cite{Sch20}]
    Let $G$ be a graph and let $\sigma,\tau\in\aut(G)\backslash\{\id\}$.
    \begin{enumerate}
        \item If $\sigma$ and $\tau$ have disjoint supports, then $\aut^+(G)\neq\aut(G)$.

        \item If $\sigma$ and $\tau$ have disconnected supports, then $\aut^{o+}(G)\neq\aut(G)$.
    \end{enumerate}
\end{theorem}

We extend this by considering graph endomorphisms that are not permutations, which will give rise to more general quantum endomorphisms.

\begin{definition}
    Let $X$ be a set and let $f:X\rightarrow X$ be a function. The \emph{support of $f$} is the set $\supp(f)=\set*{x\in X}{f(x)\neq x}$.
\end{definition}

This agrees with the notion of support for permutations, but not the notion for real- or complex-valued functions, where the support is the set of points that are not mapped to $0$. The results of~\cite{Sch20} hold generally for permutations, but extending them to non-invertible functions requires an additional assumption on them.

\begin{definition}
    Let $f,g\in\enmo(G)$. We say $f$ and $g$ are \emph{weakly adjacency congruent (WAC)} if, for all $x\in\supp(f)$ and $y\in\supp(g)$ such that $x\sim_G y$, it holds that $f(x)\sim_G g(y)$. This induces a relation on $G$, denoted $f\approx_Gg$.
\end{definition}

The WAC relation is a weaker form of the adjacency relation on the exponential graph $G^G$, where $f\sim_{G^G} g$ if $x\sim_G y$ implies $f(x)\sim_G g(y)$. We work out some properties of the WAC relation.

\begin{lemma} Let $G$ be a graph.\label{lem:wac-props}
    \begin{enumerate}[(i)]
        \item The WAC relation is reflexive and symmetric.
        \item For all $f\in\enmo(G)$, $\id\approx_G f$.
        \item Let $f,g\in\enmo(G)$ with disjoint supports. Suppose that for all $x\sim_G y$ such that $x,g(y)\in\supp(f)$ and $y,f(x)\in\supp(g)$, it holds that $f(x)\sim_G g(y)$. Then, $f\approx_G g$.
        \item If $f\approx_G g$, then $f^i\approx_Gg^j$ for all $i,j\in\N$.
    \end{enumerate}
\end{lemma}

\begin{proof}
    \begin{enumerate}[(i)]
        \item Let $x,y\in\supp(f)$. If $x\sim_G y$, then $f(x)\sim_Gf(y)$, since $f$ is an endomorphism. Hence $f\approx_Gf$. Also, the definition of $f\approx_G g$ is symmetric, so it implies $g\approx_G f$.

        \item $\id\approx_G f$ vacuously, as the support of $\id$ is empty.

        \item Suppose that $f$ and $g$ satisfy the condition in the statement. Let $x\sim_G y$ such that $x\in\supp(f)$ and $y\in\supp(g)$. We need to show that $f(x)\sim_Gg(y)$. Consider first the case that $g(y)$ is a fixed point of $f$. As the supports of $f$ and $g$ are disjoint, we must have that $x$ is a fixed point of $g$. As $g$ is an endomorphism, this implies that $g(y)\sim_Gg(x)=x$. Next, as $f$ is an endomorphism, $g(y)=f(g(y))\sim_Gf(x)$. The case that $f(x)$ is a fixed point of $g$ follows by symmetry. Finally, the case that $g(y)\in\supp(f)$ and $f(x)\in\supp(g)$ follows by hypothesis. Hence $f\approx_Gg$.

        \item Suppose, for a contradiction, that $f\approx_G g$ but $f^i\not\approx_G g^j$ for some $(i,j)$. Then, there exist $x\sim_G y$ such that $x\in\supp(f)$, $y\in\supp(g)$, and $f^i(x)\nsim_G g^j(y)$. We may assume that $(i,j)$ is minimal in the sense that $f^{i'}\approx_Gg^{j'}$ for all $(i',j')$ such that $(i'\leq i\land j'<j)\lor(i'<i\land j'\leq j)$. We may also assume that $f^{i'}(x)\in\supp(f)$ for all $i'<i$, since if $f^{i'}(x)$ is a fixed point, then $f^{i'}(x)=f^{i'+1}(x)=\ldots=f^i(x)$, so we may replace $i$ with $i'$. In the same way, we may assume $g^{j'}(x)\in\supp(g)$ for all $j'<j$. With this, since $f^{i-1}\approx_G g^{j-1}$, we have $f^{i-1}(x)\sim_G g^{j-1}(y)$. Since we also have $f^{i-1}(x)\in\supp(f)$ and $g^{j-1}(y)\in\supp(g)$, this implies that $f^i(x)\sim_G g^j(y)$, since $f$ and $g$ are WAC, giving the wanted contradiction.\qedhere
    \end{enumerate}
\end{proof}

\begin{lemma}\label{lem:wac-classes}
    Suppose that $f,g\in\enmo(G)$ with disjoint supports.
    \begin{enumerate}[(i)]
        \item If $f$ and $g$ commute, then $f\approx_G g$.
        \item If $f$ and $g$ have disconnected supports, then $f\approx_Gg$.
    \end{enumerate}
\end{lemma}

In particular, (i) implies that a pair of permutations with disjoint supports are WAC.

\begin{proof}
    \begin{enumerate}[(i)]
        \item Let $x\sim_G y$ such that $x\in\supp(f)$ and $y\in\supp(g)$. Then, as the supports of $f$ and $g$ are disjoint, $x$ is a fixed point of $g$ and $y$ is a fixed point of $f$, giving $g(x)=x$ and $f(y)=y$. Also, $g\circ f=f\circ g$ is an endomorphism, so $g(f(x))\sim_G g(f(y))$. Putting these together with the commutation assumption gives
        $$f(x)=f(g(x))=g(f(x))\sim_Gg(f(y))=g(y).$$
        \item In this case, if $x\sim_Gy$, we cannot simultaneously have $x\in\supp(f)$ and $y\in\supp(g)$. Thus we vacuously get $f\approx_G g$.\qedhere
    \end{enumerate}
\end{proof}

\begin{proposition}[Schmidt criterion for endomorphisms]\label{thm:endo-schmidt}
    Suppose that $f,g\in\enmo(G)$ with nonempty supports.
    \begin{enumerate}[(i)]
        \item If $f$ and $g$ have disjoint supports and are WAC, then $\enmo^{q}(G)\neq\enmo^c(G)$.
        \item If $f$ and $g$ have disconnected supports, then $\enmo^{oq}(G)\neq\enmo^c(G)$.
    \end{enumerate}
\end{proposition}

\begin{proof}
    For (i), consider the $C^\ast$-algebra $$C^\ast(\Z_2\ast\Z_2)=\gen*{p_0,p_1,q_0,q_1}{p_i^2=p_i^\ast=p_i,q_i^2=q_i^\ast=q_i,p_0+p_1=q_0+q_1=1}.$$
    Let $\varphi:\End^+(G)\rightarrow C^\ast(\Z_2\ast\Z_2)$ be the $\ast$-homomorphism defined on the generators as $\varphi(p_{xy})=\delta_{x,y}(p_0+q_0-1)+\delta_{f(x),y}p_1+\delta_{g(x),y}q_1$. We want to show that this is a well-defined $\ast$-homomorphism by showing that the $\varphi(p_{xy})$ satisfy the relations of $\End^+(G)$. If $x$ is a fixed point of $f$, the $\varphi(p_{xy})=\delta_{x,y}q_0+\delta_{g(x),y}q_1$, which is a projection. On the other hand, if $x$ is not a fixed point of $f$, then it must be a fixed point of $g$, so $\varphi(p_{xy})=\delta_{x,y}p_0+\delta_{f(x),y}p_1$. Also,
    $$\sum_{y\in G}\varphi(p_{xy})=(p_0+q_0-1)+p_1+q_1=1.$$
    So $\{\varphi(p_{xy})\}_{y\in G}$ is a PVM for each $x$. For the remaining relations, let $x\sim_G x'$ and $y\nsim_Gy'$. We need to show that $\varphi(p_{xy})\varphi(p_{x'y'})=0$. Since $f$ and $g$ are endomorphisms, $f(x)\sim_G f(x')$ and $g(x)\sim_G g(x')$. First, suppose that $x$ and $x'$ are fixed points of $f$. Then,
    $$\varphi(p_{xy})\varphi(p_{x'y'})=(\delta_{x,y}q_0+\delta_{g(x),y}q_1)(\delta_{x',y'}q_0+\delta_{g(x'),y'}q_1)=\delta_{x,y}\delta_{x',y'}q_0+\delta_{g(x),y}\delta_{g(x'),y'}q_1=0,$$
    as $x\sim_G x'$ and $g(x)\sim_Gg(x')$ but $y\nsim_Gy'$. Next, if $x,x'\in\supp(f)$, then they are fixed points of $g$, so we have by symmetry that $\varphi(p_{xy})\varphi(p_{x'y'})=0$. Lastly, suppose that $x\in\supp(f)$ and $x'$ is a fixed point of $f$ (the case that $x$ is a fixed point of $f$ and $x'\in\supp(f)$ will follow by symmetry). Then,
    $$\varphi(p_{xy})\varphi(p_{x'y'})=\delta_{x,y}\delta_{x',y'}p_0q_0+\delta_{x,y}\delta_{g(x'),y'}p_0q_1+\delta_{f(x),y}\delta_{x',y'}p_1q_0+\delta_{f(x),y}\delta_{g(x'),y'}p_1q_1.$$
    The first term is $0$ as $x\sim_Gx'$ but $y\nsim_Gy'$. The second term is $0$ as $x=g(x)\sim_Gg(x')$, and similarly for the third term as $f(x)\sim_G f(x')=x'$. Finally, for the final term, if $x'$ is also a fixed point of $g$, then it is $0$ for the same reason as the third term; else $x'\in\supp(g)$, so by WAC property, $f(x)\sim_Gg(x')$. Hence $\varphi(p_{xy})\varphi(p_{x'y'})=0$, so $\varphi$ is a $\ast$-homomorphism.

    Composing $\varphi$ with a $\ast$-homomorphism to a finite-dimensional $C^\ast$-algebra gives an element of $\enmo^{q}(G)$. Since the supports of $f$ and $g$ are disjoint and nonempty, there exist distinct $x\in\supp(f)$ and $y\in\supp(g)$. Then, $\varphi(p_{xf(x)})=p_0$ and $\varphi(p_{yg(y)})=q_0$. By composing with the representation extending $p_0\mapsto\ketbra{0}$ and $q_0\mapsto\ketbra{+}$, we get a non-classical quantum endomorphism.

    For (ii), we want to show that $\varphi$ defined in the same way on $\End^{o+}(G)$ also extends to a $\ast$-homomorphism. Due to \cref{lem:wac-classes}(ii), $f$ and $g$ with disconnected supports are WAC, so all the relations but the commutativity relations hold by the proof of (i). To show these, let $x\sim_G x'$ and let $y,y'\in G$. First, if $x,x'\in\supp(f)$, $\varphi(p_{xy}),\varphi(p_{x'y'})\in\spn\{p_0,p_1\}$ so they commute. Similarly, if $x$ and $x'$ are fixed points of $f$, the $\varphi(p_{xy}),\varphi(p_{x'y'})\in\spn\{q_0,q_1\}$, so they commute. Finally, in the case that $x\in\supp(f)$ and $x'$ is a fixed point of $f$, we know that $x'$ must also be a fixed point of $g$, as it is adjacent to an element of $\supp(f)$. As such, $x'$ is a fixed point of both $f$ and $g$, so $\varphi(p_{x'y'})=\delta_{x',y'}$, which is central. Hence, $\varphi:\End^{o+}(G)\rightarrow C^\ast(\Z_2\ast\Z_2)$ is a $\ast$-homomorphism, and we can find that $\enmo^{oq}(G)\neq\enmo^c(G)$ by the same argument as for (i).
\end{proof}

\begin{example}
    \label{xmpl:Schmidt}
    \begin{enumerate}[(a)]
        \item\label{xmpl:diamond} The \emph{diamond graph} is the graph $D$ defined by $V(D)=[4]$ and $$E(D)=\set*{\{1,2\},\{2,3\},\{3,4\},\{4,1\},\{2,4\}}.$$ See also \cref{fig:diamond}. Consider the maps $f,g:[4]\rightarrow[4]$ defined as
        \begin{align*}
            f=\begin{pmatrix}1&2&3&4\\1&2&1&4\end{pmatrix}\qquad g=\begin{pmatrix}1&2&3&4\\3&2&3&4\end{pmatrix}.
        \end{align*}
        It is easy to check that these are endomorphisms of $D$ with disconnected supports. Hence by \cref{thm:endo-schmidt} and \cref{thm:nogo}, $D$ has no oracular or non-oracular commutativity gadget. This gives an example of a relational structure with alphabet size $4$ with no commutativity gagdets.

        \item Let $D' = C_6 + \{0,3\}$ be the graph obtained from $C_6$ by adding the edge $\{0,3\}$; see \cref{fig:diamond}.
        Consider the maps $f,g : \Z_6 \to \Z_6$ given by
        \begin{align*}
            f = \begin{pmatrix} 0 & 1 & 2 & 3 & 4 & 5 \\ 0 & 1 & 2 & 3 & 2 & 1 \end{pmatrix} \qquad g = \begin{pmatrix} 0 & 1 & 2 & 3 & 4 & 5 \\ 0 & 5 & 4 & 3 & 4 & 5 \end{pmatrix}.
        \end{align*}
        Again, it is easy to see that these are endomorphisms of $D'$ with disconnected supports.
        By \cref{thm:endo-schmidt} and \cref{thm:nogo}, $D'$ has no oracular or non-oracular commutativity gadget.
        But since $D'$ is $2$-connected and outerplanar with $|V(D')| \neq 4$, it follows from \cite[Corollary~5.4]{DFKRZ25} that $D'$ has no quantum symmetry; that is, $\aut^+(D') = \aut(D')$.
        This shows that even graphs without quantum symmetry can satisfy the Schmidt criterion for endomorphisms, and therefore not admit commutativity gadgets.
    \end{enumerate}
\end{example}

\begin{figure}
    \centering
    \begin{tikzpicture}[scale=1.5,vertex/.style={circle,fill,inner sep=1.3pt}]
        \draw (0,0) node[vertex] (v1) {} node[above] {$1$};
        \draw (1,-1) node[vertex] (v2) {} node[right] {$2$};
        \draw (0,-2) node[vertex] (v3) {} node[below] {$3$};
        \draw (-1,-1) node[vertex] (v4) {} node[left] {$4$};
        \draw (v2) -- (v3) -- (v4) -- (v1) -- (v2) -- (v4);
        \begin{scope}[xshift=4cm,yshift=-1cm]
            \draw (90:1cm) node[vertex] (v0) {} node[above] {$0$};
            \draw (30:1cm) node[vertex] (v1) {} node[above right] {$1$};
            \draw (-30:1cm) node[vertex] (v2) {} node[below right] {$2$};
            \draw (-90:1cm) node[vertex] (v3) {} node[below] {$3$};
            \draw (-150:1cm) node[vertex] (v4) {} node[below left] {$4$};
            \draw (150:1cm) node[vertex] (v5) {} node[above left] {$5$};
            \draw (v0) -- (v1) -- (v2) -- (v3) -- (v4) -- (v5) -- (v0) -- (v3);
        \end{scope}
    \end{tikzpicture}
\caption{The diamond graph $D$ (left) and the graph $D' = C_6 + \{0,3\}$ (right).}
\label{fig:diamond}
\end{figure}
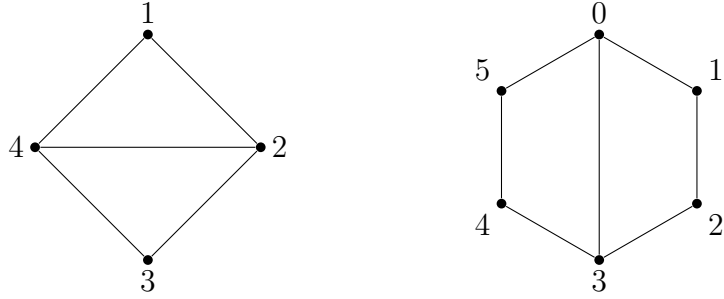

\cref{xmpl:Schmidt}\labelcref{xmpl:diamond} shows that there is a relational structure with $\tsf{NP}$-complete CSP on an alphabet of size $4$ that does not admit an oracular commutativity gadget. On the other hand, due to~\cite{CM24}, every $\tsf{NP}$-complete relational structure with alphabet size $2$ has an oracular commutativity gadget. However, the situation in the intermediate case of alphabet size $3$ is not clear. Since the only non-bipartite graph on $3$ vertices is the triangle $K_3$, every graph with $\tsf{NP}$-complete CSP on $3$ vertices has both a commutativity gadget and an oracular commutativity gadget (also, every the entangled CSP of a bipartite graph is already in $\tsf{P}$; see~\cite{BZ25} or \cref{sec:bipartite}). We do not know if this extends to non-graph relational structures with alphabet size $3$.

\begin{question}
    Is there an $\tsf{NP}$-complete relational structure with alphabet size $3$ that has no commutativity gadget or oracular commutativity gadget?
\end{question}

\subsection{Graphs with only classical endomorphisms}\label{sec:classical-endo}

In this section, we study the quantum endomorphism monoids of several well-known graphs. We find a variety of graphs with only classical endomorphisms, which therefore provide ideal candidates for the study of commutativity gadgets. We rely on results finding graphs that have no quantum symmetry, \textit{i.e.} $\aut^+(G)=\aut(G)$.

\begin{theorem}[\cite{BB07,Sch18,Sch20}]\label{thm:nosym}
    For $n\in\N$, the odd cycles $C_{2n+1}$, the Petersen graph $P$, and the odd graphs $O_n=K(2n-1,n-1)$ have no quantum symmetry.
\end{theorem}
Note that $O_2=C_3=K_3$ the triangle and $O_3=P$ the Petersen graph.

We use the idea of cores to relate nonexistence of quantum symmetries and transformations. Recall that a relational structure $A$ is a \emph{core} if $\aut(A)=\enmo(A)$.

\begin{definition}
    A relational structure $A$ is a \emph{quantum core} if $\aut^+(A)=\enmo^+(A)$ and an \emph{oracular quantum core} if $\aut^{o+}(A)=\enmo^{o+}(A)$.
\end{definition}

\begin{lemma}\label{lem:cores}
    Let $A$ be a relational structure.
    \begin{enumerate}[(i)]
        \item If $A$ is a quantum core, then $A$ is an oracular quantum core.
        \item If $A$ is an oracular quantum core, then $A$ is a core.
    \end{enumerate}
\end{lemma}

\begin{proof}
    \begin{enumerate}[(i)]
        \item Let $A$ be a quantum core. Since the generators of $\End^{o+}(A)$ satisfy all the relations of $\End^+(A)$, they also satisfy all the relations of $\Aut^+(A)$. It remains to show that they satisfy the additional relations of $\Aut^{o+}(A)$, namely $[p_{ba_i},p_{b'a_j}]=0$ for all $R\in\sigma$, $1\leq i\leq j\leq\ar(R)$, $\mathbf{a}\in R^A$, and $b,b'\in A$. The case $b=b'$ follows from the fact that $\{p_{ba}\}_{a\in A}$ is a PVM. Else, suppose first that there exists $S\in\sigma$, $i'\neq j'\in[\ar(S)]$, and $\mathbf{b}\in S^A$ such that $b=b_{i'}$ and $b'=b_{j'}$. Then, by the commutation relations of $\End^{o+}(A)$, $[p_{ba_i},p_{b'a_j}]=[p_{b_{i'}a_i},p_{b_{j'}a_j}]=0$. If not, for all $S\in\sigma$, $i'\neq j'\in[\ar(S)]$, and $\mathbf{b}\in A^{\ar(S)}$, if $b_{i'}=b$ and $b_{j'}=b'$, then $\mathbf{b}\notin S^A$. As such, using the relations of $\Aut^+(A)$,
        \begin{align*}
            p_{ba_i}p_{b'a_j}&=\hspace{-0.15cm}\sum_{\substack{b_1,\ldots,b_{i-1}\in A,\\b_{i+1},\ldots,b_{j-1}\in A,\\b_{j+1},\ldots,b_{\ar(R)}\in A}}\hspace{-0.15cm}p_{b_1a_1}\cdots p_{b_{i-1}a_{i-1}}p_{ba_i}p_{b_{i+1}a_{i+1}}\cdots p_{b_{j-1}a_{j-1}}p_{b'a_j}p_{b_{j+1}a_{j+1}}\cdots p_{b_{\ar(R)}a_{\ar(R)}}\\
            &=\sum_{\mathbf{b}\in A^{\ar(R)}:\,b_i=b,b_j=b'}p_{b_1a_1}\cdots p_{b_{\ar(R)}a_{\ar(R)}}\\
            &=\sum_{\mathbf{b}\in R^A:\,b_i=b,b_j=b'}p_{b_1a_1}\cdots p_{b_{\ar(R)}a_{\ar(R)}}=0,
        \end{align*}
        where the first line follows by $\sum_{b\in A}p_{ba}=1$, the second line by $p_{b_1a_1}\cdots p_{b_na_n}=0$ for all $\mathbf{b}\notin R^A$, and the last line as the sum is empty by the argument above. Then, we also have $p_{b'a_j}p_{ba_i}=(p_{ba_i}p_{b'a_j})^\ast=0$, so $p_{ba_i}$ and $p_{b'a_j}$ commute.

        \item $\End(A)=C(\enmo(A))$ is the quotient of $\End^{o+}(A)$ by the ideal $\gen{[p_{ab},p_{a'b'}]}{a,b,a',b'\in A}$ in $\End^{o+}(A)$. In the same way, $\Aut(A)=C(\aut(A))$ is the quotient of $\Aut^{o+}(A)$ by the ideal with the same generators $\gen{[p_{ab},p_{a'b'}]}{a,b,a',b'\in A}\subseteq\Aut^{o+}(A)$. But, since $A$ is an oracular quantum core, so $\End^{o+}(A)=\Aut^{o+}(A)$. Therefore, $\End(A)=\Aut(A)$, so by Gelfand duality $\enmo(A)=\aut(A)$.
    \end{enumerate}
\end{proof}

\begin{theorem}\label{thm:core-examples}
    Let $n\in\N$. $C_{2n+1}$, $O_n$, and $K(3n+2,n+1)$ are quantum cores.
\end{theorem}

\begin{corollary}\label{cor:core-classical}
    For $n\in \N$, $C_{2n+1}$ and $O_n$ have only classical endomorphisms.
\end{corollary}

\begin{proof}
    By \cref{thm:core-examples}, $\enmo^+(G)=\aut^+(G)$ for $G=C_{2n+1},O_n$. Further, by \cref{thm:nosym}, $\aut^+(G)=\aut(G)$. Finally, using \cref{lem:cores}, it follows that $\aut(G)=\enmo(G)$. Hence $\enmo^+(G)=\enmo(G)$ as wanted.
\end{proof}

We need some some lemmas to prove \cref{thm:core-examples}. First, we show the following useful lemma, which is a version of~\cite[Lemma 3.2.2]{Sch20} for general quantum morphisms.

\begin{lemma}\label{lem:we-use-this-one-all-the-time}
    Let $G$ and $H$ be graphs. Suppose that $u,u'\in G$ and $v,v'\in H$ be vertices such that there is a walk of length $\ell$ from $u$ to $u'$ but no walk of length $\ell$ from $v$ to $v'$. Then, $p_{uv}p_{u'v'}=0$ in $\Mor^+(G,H)$.
\end{lemma}

\begin{proof}
    Let $\mathbf{u}$ be the walk of length $\ell$ from $u$ to $u'$. Then,
    \begin{align*}
        p_{uv}p_{u'v'}&=\sum_{v_2,\ldots,v_{\ell}}p_{uv}p_{u_2v_2}\cdots p_{u_\ell v_\ell}p_{u'v'}\\
        &=\sum_{\substack{\mathbf{v}\in H^{\ell+1}\\v_1=v,\,v_{\ell+1}=v'\\v_i\sim_H v_{i+1}\;\forall\,i}}p_{u_1v_1}\cdots p_{u_{\ell+1}v_{\ell+1}}=0,
    \end{align*}
    as every term in the sum corresponds to a walk of length $\ell$ from $v$ to $v'$, and there are no such walks.
\end{proof}

Next, we need a general lemma about projections in $C^\ast$-algebras.

\begin{lemma}\label{lem:projector-sums}
    Let $A$ be a $C^\ast$-algebra, and $q_{ab}\in A$ be projections for $a,b\in[n]$. If $\sum_{a,b}q_{ab}=n$ and $q_{ab}q_{ab'}=0$ for all $a$ and $b\neq b'$, then $\sum_bq_{ab}=1$ for all $a$.
\end{lemma}

\begin{proof}
    Since $q_{ab}q_{ab'}=0$, $\parens*{\sum_bq_{ab}}^2=\sum_{b,b'}q_{ab}q_{ab'}=\sum_{b}q_{ab}$, so $\sum_bq_{ab}$ is a projection. Since every projection $\leq 1$, we get $\sum_{a,b}q_{ab}\leq n$, with equality iff $\sum_bq_{ab}=1$ for all $a$. 
\end{proof}

Before going to the proof of the theorem, we need some properties of the graphs we are working with.

\begin{lemma}\label{lem:odd}
    Let $n\geq 2$.
    \begin{enumerate}[(i)]
        \item Every pair of distinct vertices in $O_n$ is connected by a path of odd length $<2n-1$.

        \item Every pair of non-adjacent vertices in $O_n$ is connected by a path of even length $<2n-2$.
    \end{enumerate}
\end{lemma}

\begin{proof}
    \begin{enumerate}[(i)]
        \item We use the Kneser graph presentation $O_n=K(2n-1,n-1)$. Let $u,v\subseteq[2n-1]$ be distinct vertices of $O_n$. Let $U=u\backslash v$, $V=v\backslash u$, $I=u\cap v$. Since $|u|=|v|=n-1$,
        $$\abs*{[2n-1]\backslash(u\cup v)}=(2n-1)-(|u|+|v|-|I|)=1+|I|.$$
        Then, there exists a set $J\subseteq[2n-1]$ of cardinality $|I|$ disjoint from $u$ and $v$; denote the remaining element $z$. Fix an order on $I$, and write $I_{\leq k}$ for the first $k$ elements and $I_{>k}$ for the remaining elements; we do similar for $J$. Now, consider the tuple $(u,v_0,u_0,\ldots,v_{|I|-1},u_{|I|-1},v_{|I|})$, where $u_i=U\cup I_{>i+1}\cup J_{\leq i}\cup\{z\}$ and $v_i=V\cup J_{>i}\cup I_{\leq i}$. This is path from $u$ to $v$ as the vertices are all distinct, $v_i\cap u_i=u_i\cap u_{i+1}=\varnothing$, and $v_{|I|}=V\cup I=v$. The path has length $2|I|+1$, which is odd; and as $u\neq v$, $|I|\leq n-2$, giving $2|I|+1\leq 2n-3<2n-1$.

        \item We use the same notation as above. If $u=v$, they are connected by a path of length $2$. Else, consider the tuple $(u_0',v_0',\ldots,u'_{|U|-1},v'_{|U|-1},u'_{|U|})$, where $u_i'=U_{>i}\cup V_{\leq i}\cup I$ and $v_i'=V_{>i+1}\cup J\cup U_{\leq i}\cup\{z\}$. This is a path from $u$ to $v$, as the vertices are distinct, $u_i'\cap v_i'=v_{i}'\cap u_{i+1}'=\varnothing$, and $u'_{|U|}=V\cup I=v$. It has length $2|U|$, which is even; and as $u$ and $v$ are not adjacent, $|I|\geq 1$, giving $|U|\leq n-2$, and hence $2|U|\leq 2n-4<2n-2$.
    \end{enumerate}
\end{proof}

\begin{lemma}\label{lem:kneser}
    Let $n>2k$.
    \begin{enumerate}
        \item $K(n,k)$ has no cycles of length $3$ iff $n<3k$.
        \item Every pair of distinct vertices of $K(n,k)$ is connected by a walk of length $3$ iff $n\geq 3k-1$.
        \item Every pair of non-adjacent vertices of $K(n,k)$ is connected by a walk of length $2$ iff $n\geq 3k-1$.
    \end{enumerate}
\end{lemma}

\begin{proof}
    \begin{enumerate}[(i)]
        \item Suppose first that $n<3k$. Let $u,v\subseteq[n]$ be adjacent vertices of $K(n,k)$. Then, $\abs*{[n]\backslash(u\cup v)}=n-2k<k$. Therefore, there is no subset $w\subseteq[n]$ of size $k$ such that $u\cap w=v\cap w=\varnothing$, and hence no vertex of $K(n,k)$ adjacent to both $u$ and $v$.

        Conversely, suppose that $n\geq 3k$. Then, $K(n,k)$ contains the triangle formed by the vertices $[k]$, $[2k]\backslash[k]$, and $[3k]\backslash[2k]$.

        \item Suppose $n\geq 3k-1$. Let $u,v\subseteq[n]$ be distinct vertices of $K(n,k)$. If $u$ and $v$ are adjacent, then $(u,v,u,v)$ is a walk of length $3$ from $u$ to $v$. Otherwise, if $u$ and $v$ are not adjacent, $\abs*{u\cap v}\leq k-1$, so $\abs*{n\backslash(u\cup v)}=n-2k+\abs*{u\cap v}\geq k-1+\abs*{u\cap v}\geq 2\abs*{u\cap v}$. Hence there are disjoint subsets $S_1,S_2\subseteq[n]\backslash(u\cup v)$ of size $\abs*{u\cap v}$. Then, $(u,(v\backslash u)\cup S_1,(u\backslash v)\cup S_2,v)$ is a walk from $u$ to $v$.

        For the converse, suppose $n<3k-1$. Let $u=[k]$ and $v=[k+1]\backslash\{1\}$. Let $w$ and $w'$ be vertices of $K(n,k)$ such that $u$ is adjacent to $w$ and $v$ is adjacent to $w'$. In particular as $n\leq 3k-2$, $w\subseteq A=\{k+1,\ldots,3k-2\}$ and $w'\subseteq B=\{1,k+2,\ldots,3k-2\}$. We have $\abs*{w\cap B},\abs*{w'\cap A}\geq k-1$ and $\abs*{A\cap B}=2k-3$. If $w$ and $w'$ are adjacent, then $w\cap B$ and $w'\cap A$ are disjoint in $A\cap B$. But this is impossible $\abs*{w\cap B}+\abs*{w'\cap A}>\abs*{A\cap B}$, so $w\cap w'\neq\varnothing$. Hence, there is no path of length $3$ from $u$ to $v$.

        \item Suppose that $n\geq 3k-1$. Let $u,v\subseteq[n]$ be non-adjacent vertices of $K(n,k)$. Then, $\abs*{u\cap v}\geq 1$, so $\abs*{[n]\backslash(u\cup v)}=n-2k+\abs*{u\cap v}\geq n-2k+1\geq k$. Hence there exists a vertex of $K(n,k)$ adjacent to both $u$ and $v$.

        For the converse, suppose $n<3k-1$. Then, the vertices $[k]$ and $[2k-1]\backslash[k-1]$ are not adjacent, and there is no vertex adjacent to both of them.
    \end{enumerate}
\end{proof}

\begin{lemma}\label{lem:cyclic-walks}
    Suppose $G$ is a graph with odd girth $\ell$. Then, $G$ has odd walk girth $\ell$.
\end{lemma}

This does not hold for the even girth: in fact every graph with an edge has even walk girth $2$.

\begin{proof}
    Suppose $\mathbf{x}$ is a cyclic walk of odd length $m<\ell$. We want to extract a cycle of odd length from $\mathbf{x}$ by splitting it up wherever it intersects itself. If $\mathbf{x}$ is already a cyclic path, then we have a contradiction. If not, there exists $i<j$ such that $x_i=x_j$. As such, we can split $\mathbf{x}$ into two cyclic walks, $(x_i,\ldots,x_j)$ of length $j-i$ and $(x_1,\ldots,x_i,x_{j+1},\ldots,x_{m+1})$ of length $m-(j-i)$. Both of these have length $<m$ and one of them has odd length. We can repeat this argument on this walk of odd length recursively. As the length of the walk decreases with each step, we must eventually reach a cyclic path of odd length. This is a contradiction as noted above.
\end{proof}

\begin{proof}[Proof of \cref{thm:core-examples}]
    First, consider the case $G=C_{2n+1}$. The generators of $\End^+(C_{2n+1})$ satisfy the relations $p_{ab}p_{(a+1)b'}=0$ if $b'\neq b\pm 1$ (modulo $2n+1$). First, we want to show that for $a\neq a'$, $p_{ab}p_{a'b}=0$. There exist $k,k'\in[2n]$ such that $a'=a+k=a-k'$; and as $2n+1$ is odd, one of $k,k'$ is odd. Hence, there is a path of odd length $<2n+1$ from $a$ to $a'$. But, by \cref{lem:cyclic-walks}, there is no cyclic walk of odd length $<2n+1$ in $C_{2n+1}$, so $p_{ab}p_{a'b}=0$ by \cref{lem:we-use-this-one-all-the-time}. This also implies by \cref{lem:projector-sums} that $\sum_ap_{ab}=1$. To finish, we also need that $p_{ba}p_{b'(a+1)}=0$ for all $b'\neq b\pm 1$. In the same way as above, there is a walk of even length $<2n$ from $b$ to $b'$ in $C_{2n+1}$. However, there is no such walk from $a$ to $a+1$, giving the wanted relation.

    We proceed in the same way for the other two classes of graph. Consider the case $G=O_n$ with $n\geq 2$. Let $a,a'\in O_n$ be distinct. By \cref{lem:odd}, there is a walk of odd length $<2n-1$ from $a$ to $a'$. $O_n$ has no cycle of odd length $<2n-1$, so by \cref{lem:cyclic-walks} there is no cyclic walk of this length. Hence $p_{ab}p_{a'b}=0$. As before, this implies $\sum_{a}p_{ab}=1$. Finally, let $a\nsim_{O_n}a'$ and $b\sim_{O_n}b'$. Again using \cref{lem:odd}, there is a walk of even length $<2n-2$ from $a$ to $a'$. On the other hand, a walk of even length $<2n-2$ from $b$ to $b'$ would give rise to a cyclic walk of odd length $<2n-1$, which cannot exist. So we have $p_{ab}p_{a'b'}=0$.

    Finally, consider the case $G=K(3n+2,n+1)$. Let $a,a'$ be distinct vertices of $G$. Using \cref{lem:kneser}, there is a walk of length $3$ from $a$ to $a'$, but there is no cyclic walk of length $3$, so $p_{ab}p_{a'b}=0$. Again, this implies $\sum_ap_{ab}=1$. Finally, for $a\nsim_G a'$ and $b\sim_G b'$, there is a walk of length $2$ from $a$ to $a'$, but as there is not cycle of length $3$, there is no walk of length $2$ from $b$ to $b'$. Hence, $p_{ab}p_{a'b'}=0$.
\end{proof}

It follows from \cref{cor:core-classical} that \cref{thm:nogo} does not preclude the existence of a commutativity gadget for $C_{2n+1}$ or $O_n$. However, it seems challenging to construct a commutativity gadget for these graphs. In \cref{sec:an-attempt}, we show that a natural extension of Ji's triangular prism commutativity gadget~\cite{Ji13} is not a commutativity gadget for the cycle $C_{2n+1}$ with $n\geq 2$.

\begin{question}
    Does there exist a commutativity gadget for $C_{2n+1}$ or $O_n$? Are their corresponding entangled CSPs undecidable?
\end{question}

\subsection{Oracularisability for graphs}\label{sec:graph-oracularisability}

In this section, we study oracularisability, in the sense of \cref{def:oracularisable}, for graph CSPs; this has consequences for the structure of possible commutativity gadgets for $C_{2n+1}$ and $O_n$, discussed in the previous section.

\begin{lemma}
    Let $G$ be a graph. $G$ is oracularisable if and only if $\Mor^+(K_2,G)$ is commutative.
\end{lemma}

\begin{proof}
    If $G$ is oracularisable, then $\mor^{+}(K_2,G)=\mor^{o+}(K_2,G)$, so at the level of the algebras, $\Mor^{+}(K_2,G)=\Mor^{o+}(K_2,G)$. Since $\Mor^{o+}(K_2,G)$ is generated by two commuting PVMs, it is commutative, and hence so is $\Mor^{+}(K_2,G)$.

    For the converse, suppose that $H$ is a graph. Now, let $h\sim_Hh'$. Then, there is a $\ast$-homomorphism $\Mor^+(K_2,G)\rightarrow\Mor^+(H,G)$ given by $p_{0g}\mapsto p_{hg}$ and $p_{1g}\mapsto p_{h'g}$. Since $\Mor^+(K_2,G)=\Mor^{o+}(K_2,G)$, we have that $[p_{hg},p_{h'g'}]=0$ for all $g,g'\in G$. As such, all the relations of $\Mor^{o+}(H,G)$ are satisfied in $\Mor^+(H,G)$, so $\mor^+(H,G)=\mor^{o+}(H,G)$.
\end{proof}

\begin{proposition}
    Let $G$ be a graph. $G$ is oracularisable if and only if $G$ contains no cyclic path of length $4$.
\end{proposition}

\begin{proof}
    First suppose that $G$ has a cyclic path of length $4$; denote this $(a,b,c,d,a)$. Consider the $\ast$-homomorphism $\pi:\Mor^+(K_2,G)\rightarrow B(\C^2\otimes\C^2)$ defined as $\pi(p_{0a})=\ketbra{00}$, $\pi(p_{0b})=\ketbra{10}$, $\pi(p_{0c})=\ketbra{01}$, $\pi(p_{0d})=\ketbra{11}$, $\pi(p_{1a})=\ketbra{1+}$, $\pi(p_{1b})=\ketbra{0+}$, $\pi(p_{1c})=\ketbra{1-}$, $\pi(p_{1d})=\ketbra{0-}$, and $\pi(p_{gh})=0$ for all other generators. It is clear to see that these satisfy the relations of $\Mor^+(K_2,G)$, so it is well-defined. $\pi(p_{0a})$ and $\pi(p_{1b})$ do not commute, so by the lemma, $G$ is not oracularisable.

    For the converse, suppose that $G$ contains no cyclic path of length $4$. Let $u,v\in G$; we want to show that $[p_{0u},p_{1v}]=0$ in $\Mor^+(K_2,G)$. First, if $u\nsim_G v$, then $p_{0u}p_{1v}=0$, so $[p_{0u},p_{1v}]=0$. Now, suppose that $u\sim_G v$. Then,
    \begin{align*}
        p_{0u}p_{1v}=\sum_{w\in G}p_{0u}p_{1v}p_{0w}=\sum_{w\sim_G v}p_{0u}p_{1v}p_{0w}.
    \end{align*}
    We claim that $p_{0u}p_{1v}p_{0w}=0$ when $w\neq u$. In fact,
    \begin{align*}
        p_{0u}p_{1v}p_{0w}=p_{0u}\parens*{1-\sum_{v'\neq v}p_{1v'}}p_{0w}=p_{0u}p_{0w}-\sum_{v'\neq v}p_{0u}p_{1v'}p_{0w}.
    \end{align*}
    As $u\neq w$, then $p_{0u}p_{0w}=0$; and as there is no cyclic path of length $4$, we must either have $v'\nsim_G u$ or $v'\nsim_G w$, and hence $p_{0u}p_{1v'}p_{0w}=0$. Therefore, $p_{0u}p_{1v}p_{0w}=0$ when $w\neq u$, so $p_{0u}p_{1v}=p_{0u}p_{1v}p_{0u}$. As this is hermitian, $[p_{0u},p_{1v}]=0$ in this case. Hence, $\Mor^{+}(K_2,G)$ is commutative, so by the lemma, $G$ is oracularisable.
\end{proof}

\begin{corollary}
    For all $n$, $C_{2n+1}$ and $O_n$ are oracularisable.
\end{corollary}

Hence, $C_{2n+1}$ and $O_n$ have commutativity gadgets if and only if they have oracular commutativity gadgets.

\small
\bibliographystyle{bibtex/bst/alphaarxiv.bst}
\bibliography{bibtex/bib/full.bib,bibtex/bib/quantum.bib,bibtex/quantum_new.bib}

\normalsize
\appendix

\section{Disproving an extension of Ji's triangular prism gadget}\label{sec:an-attempt}

One way to express Ji's triangular prism commutativity gadget is as $C_3\Boxprod P_2$. We show in this appendix that for any $n>1$ and any $k$, $C_{2n+1}\Boxprod P_k$ cannot be a commutativity gadget for $C_{2n+1}$. First, we show that sufficient distance is needed between the distinguished vertices of a commutativity gadget.

\begin{lemma}\label{lem:distinguished-distance}
    Suppose $(G,x,y)$ is a commutativity gadget for a graph $A$. Let $a,b\in A$. If there is a walk of length $\ell$ from $x$ to $y$ in $G$, then there is a walk of length $\ell$ from $a$ to $b$ in $A$.
\end{lemma}

We use this lemma in the contrapositive: if there is a no walk of length $\ell$ between a pair of vertices in $A$, then can be no walk of length $\ell$ between $x$ and $y$.

\begin{proof}
    Suppose, for a contradiction, that there exists a walk of length $\ell$ from $x$ to $y$ but no walk of length $\ell$ from $a$ to $b$. Then, by \cref{lem:we-use-this-one-all-the-time}, we must have $p_{xa}p_{yb}=0$ in $\Mor^+(G,A)$. On the other hand, as $(G,x,y)$ is a commutativity gadget, there exists $\pi_{a,b}\in\mor^{qa}(G,A)$ such that $\pi_{a,b}(p_{xa})=\pi_{a,b}(p_{yb})=1$. As such, we find that $0=\pi_{a,b}(p_{xa}p_{yb})=1$, which is the wanted contradiction.
\end{proof}

\begin{corollary}\label{cor:c2n+1-distance}
    Let $(G,x,y)$ be a commutativity gadget for $C_{2n+1}$. Then, $d(x,y)\geq 2n$.
\end{corollary}

\begin{proof}
    In $C_{2n+1}$, there is no cyclic walk of odd length $<2n+1$. Hence, by \cref{lem:distinguished-distance}, there is no walk of odd length $\leq 2n-1$ from $x$ and $y$. Also, for the same reason, there is no walk of even length $<2n$ from $0$ to $1$, and therefore there is no walk of even length $\leq 2n-2$ from $x$ to $y$. As such, we have that there is no walk of length $<2n$ from $x$ to $y$ in $G$, giving the result.
\end{proof}

In $G\Boxprod H$, there is a walk of length $\ell$ from $(g,h)$ to $(g',h')$ iff there is $s\leq\ell$ such that there is a walk of length $s$ from $g$ to $g'$ in $G$ and a walk of length $\ell-s$ from $h$ to $h'$ in $H$. In particular, this implies that the diameter of $C_{2n+1}\Boxprod P_k$ is $n+k$, so if $k<n$, $C_{2n+1}\Boxprod P_k$ cannot be a commutativity gadget for $C_{2n+1}$. We now focus on the case that $k\geq n\geq 2$. We begin with some technical lemmata.

\begin{lemma}\label{lem:pk-lifting}
    Let $\pi$ be a $\ast$-representation of $\Mor^+(P_k,C_{2n+1})$. Then, $\pi$ induces a $\ast$-representation $\pi'$ of $\Mor^+(C_{2n+1}\Boxprod P_k,C_{2n+1})$ as $\pi'(p_{(a,s)b})=\pi(p_{s(a+b)})$. 
\end{lemma}

\begin{proof}
    We show that the relations of $\Mor^+(C_{2n+1}\Boxprod P_k,C_{2n+1})$ are preserved under $\pi'$. First, by construction, the $\pi'(p_{(a,s)b})$ are projections, and $\sum_{b\in\Z_{2n+1}}\pi'(p_{(a,s)b})=\sum_{b\in\Z_{2n+1}}\pi(p_{s(a+b)})=1$. Now, suppose that $(a,s)\sim_{C_{2n+1}\Boxprod P_k}(a',s')$ and $b\nsim_{C_{2n+1}} b'$. We have that either $a=a'$ and $s=s'\pm 1$, or $a=a'\pm 1$ and $s=s'$; and that $b\neq b'\pm 1$. In the former case,
    $$\pi'(p_{(a,s)b})\pi'(p_{(a',s')b'})=\pi(p_{s(b+a)}p_{s'(b'+a)})=0,$$
    as $b+a\neq b'+a\pm 1$. In the latter case,
    $$\pi'(p_{(a,s)b})\pi'(p_{(a',s')b'})=\pi(p_{s(b+a)}p_{s(b'+a')})=0,$$
    as $b+a\neq b'+a'$, since $a=a'\pm 1$ but $b\neq b'\pm 1$.
\end{proof}

\begin{lemma}\label{lem:pk-to-c2n+1}
    Let $s_0,t_0\in[k]$ such that $t_0\geq s_0+2$. Then, there exists a finite-dimensional representation of $\Mor^+(P_k,C_{2n+1})$ under which $p_{s_0s_0}$ and $p_{t_0t_0}$ do not commute.
\end{lemma}

\begin{proof}
    Consider the maps $f,g:P_k\rightarrow P_k$ defined as $f(s)=s+2$ for $s\leq s_0$ and $f(s)=s$ for $s>s_0$; and $g(s)=s$ for all $s< t_0$ and $g(s)=s-2$ for $s\geq t_0$. Then, $f$ and $g$ are endomorphisms of $P_k$ with disconnected supports, so by \cref{thm:endo-schmidt}, there exists a $\ast$-homomorphism $\pi:\Mor^+(P_k,P_k)\rightarrow B(\C^2)$. By the construction of \cref{thm:endo-schmidt}, this $\ast$-homomorphism is given by $\pi(p_{st})=\delta_{s,t}(\ketbra{0}+\ketbra{+}-1)+\delta_{f(s),t}\ketbra{1}+\delta_{g(s),t}\ketbra{-}$. In particular, $\pi(p_{s_0s_0})=\ketbra{0}$ and $\pi(p_{t_0t_0})=\ketbra{+}$. Now, let $h:P_k\rightarrow C_{2n+1}$ be the graph homomorphism $h(s)=s\;\mathrm{mod}\;2n+1$. This induces a representation $\pi_h:\Mor^+(P_k,C_{2n+1})\rightarrow \C$ as $\pi_h(p_{sa})=\delta_{a,h(s)}$. Composing $\pi$ and $\pi_h$ as $qa$-morphisms gives
    $$(\pi_h\circ\pi)(p_{sa})=\sum_{t\in[k]}\pi(p_{st})\otimes\pi_h(p_{ta})=\sum_{t:\,a=t\;\mathrm{mod}\;2n+1}\pi(p_{st}).$$
    Therefore, $(\pi_h\circ\pi)(p_{s_0s_0})=\pi(p_{s_0s_0})+\pi(p_{s_0(s_0+2n+1)})+\ldots=\ketbra{0}$ and $(\pi_h\circ\pi)(p_{t_0t_0})=\pi(p_{t_0t_0})+\pi(p_{t_0(t_0+2n+1)})+\ldots=\ketbra{+}$, which do not commute.
\end{proof}

\begin{theorem}
    Let $n\geq 2$. There do not exist $k\in\N$, and $(a_0,s_0),(b_0,t_0)\in C_{2n+1}\Boxprod P_k$ such that $(C_{2n+1}\Boxprod P_k,(a_0,s_0),(b_0,t_0))$ is a commutativity gadget for $C_{2n+1}$.
\end{theorem}

Note that $(C_{2n+1}\Boxprod P_{n+1},(0,1),(n,n+1))$ satisfies property (i) of a commutativity gadget, and for $n=1$ gives Ji's triangular prism gadget.

\begin{proof}
    Suppose for a contradiction that such a commutativity gadget exists. Without loss of generality, we may assume that $s_0\leq t_0$. First, note that $d((a_0,s_0),(b_0,t_0))=d(a_0,b_0)+t_0-s_0$, where $d(a,b)=\min_{i\in\Z}\abs*{b-a+i(2n+1)}$ is the distance in $C_{2n+1}$. Hence, by \cref{cor:c2n+1-distance}, we must have that $d(a_0,b_0)+t_0-s_0\geq 2n$. In particular, since $d(a_0,b_0)\leq n$, $t_0\geq s_0+n\geq s_0+2$. Then, using \cref{lem:pk-to-c2n+1}, there exists $\pi\in\mor^{qa}(P_k,C_{2n+1})$ such that $\pi(p_{s_0s_0})$ and $\pi(p_{t_0t_0})$ do not commute. Now, using \cref{lem:pk-lifting}, we can construct $\pi'\in\mor^{qa}(C_{2n+1}\Boxprod P_k,C_{2n+1})$ such that $\pi'(p_{(a,s)b})=\pi(p_{s(a+b)})$. As such, $\pi'(p_{(a_0,s_0)(s_0-a_0)})=\pi(p_{s_0s_0})$ and $\pi'(p_{(b_0,t_0)(t_0-b_0)})=\pi(p_{t_0t_0})$, which do not commute. Hence, $(C_{2n+1}\Boxprod P_k,(a_0,s_0),(b_0,t_0))$ does not satisfy property (ii) of a commutativity gadget.
\end{proof}

\section{Alternate proof that bipartite graphs remain easy}\label{sec:bipartite}

In this appendix, we show that, for graph CSPs, classically easy problems remain easy when upgraded to the entangled setting. This gives an alternate proof of~\cite[Corollary 2]{BZ25}.

\begin{theorem}[Hell-Ne\v{s}et\v{r}il \cite{HN90}]
    Let $G$ be a graph. If $G$ is bipartite, then $\CSP(G)\in\tsf{P}$. Otherwise $\CSP(G)$ is $\tsf{NP}$-complete.
\end{theorem}

We show the following, extending the first part of the Hell-Ne\v{s}et\v{r}il theorem to entangled CSPs.

\begin{theorem}\label{thm:bipartite-easy}
    Let $G$ be a bipartite graph. Then, $\CSP_w(G)_{1,s}^\ast\in\tsf{P}$ for $w\in\{c-c,c-v,a\}$ and all $s\leq 1$.
\end{theorem}

It follows from this that $\SuccinctCSP_w(G)_{1,s}^\ast\in\tsf{EXP}$, and is therefore decidable.

\begin{lemma}\label{lem:bipartite-empty}
    Let $G$ be such that $E(G)=\varnothing$. Then, $\Mor^+(H,G)\neq 0$ if and only if $E(H)=\varnothing$.
\end{lemma}

\begin{proof}
    First, suppose $E(H)=\varnothing$. Fix a vertex $v\in V(G)$. Define the map $f:H\rightarrow G$ by $f(u)=v$ for all $u\in H$. This is a graph homomorphism, so $\mor(H,G)\neq\varnothing$. Hence, $\Mor^+(H,G)\neq 0$.

    Conversely, suppose that $H$ has an edge $\{a,b\}\in E(H)$. Since $u\nsim_G v$ for all $u,v\in G$, $p_{au}p_{bv}=0$ in $\Mor^+(H,G)$. Therefore, $1=\sum_{u,v\in G}p_{au}p_{bv}=0$, so $\Mor^+(H,G)=0$.
\end{proof}

\begin{lemma}\label{lem:bipartite-nonempty}
    Let $G$ be a bipartite graph such that $E(G)\neq\varnothing$. Then, $\Mor^+(H,G)\neq 0$ if and only if $H$ is bipartite.
\end{lemma}

\begin{proof}
    First, suppose $H$ is bipartite. Then, let $A,B$ be the partition of $H$; and let $\{u,v\}$ be an edge of $G$. Define $f:H\rightarrow G$ by $f(x)=u$ if $x\in A$ and $f(x)=v$ if $x\in B$. As such, if $a\sim_Hb$, then $f(a)\sim_Gf(b)$ as $u\sim_Gv$. Hence, $f$ is a graph homomorphism. So, $\mor(H,G)$ is nonempty, and therefore $\Mor^+(H,G)\neq 0$.

    Conversely, suppose that $H$ is not bipartite. Therefore, there is a cyclic walk of odd length in $H$. In particular, there is $x,y\in H$ such that there is a walk of odd length and a walk of even length from $x$ to $y$. Let $A,B$ be the partition of $G$. Consider the terms of $1=\sum_{u,v\in V(G)}p_{xu}p_{yv}$. If $u,v\in A$ or $u,v\in B$, every walk from $u$ to $v$ has even length. But there is a walk of odd length from $x$ to $y$, so by \cref{lem:we-use-this-one-all-the-time}, $p_{xu}p_{yv}=0$. Similarly, if $u\in A\land v\in B$ or $u\in B\land v\in A$, every walk from $u$ to $v$ has odd length. But there is a walk of even length from $x$ to $y$, so $p_{xu}p_{yv}=0$. Hence $1=0$ in $\Mor^+(H,G)$, so $\Mor^+(H,G)=0$.
\end{proof}

\begin{proposition}\label{prop:bipartite-classical}
    Suppose $G$ is bipartite. $\Mor^+(H,G)\neq 0$ if and only if $\mor(H,G)\neq\varnothing$.
\end{proposition}

\begin{proof}
    Suppose $\Mor^+(H,G)\neq 0$. In the case $E(G)=\varnothing$, then $E(G)=\varnothing$ by \cref{lem:bipartite-empty}, so $\mor(H,G)\neq\varnothing$. In the other case $E(G)\neq\varnothing$, then $H$ is bipartite by \cref{lem:bipartite-nonempty}, so $\mor(H,G)\neq\varnothing$ again. Conversely, if $\mor(H,G)\neq\varnothing$, then there is a nontrivial representation of $\Mor^+(H,G)$.
\end{proof}

\begin{proof}[Proof of \cref{thm:bipartite-easy}]
    Let $H$ be a graph. By \cref{prop:bipartite-classical}, $H\in Y(\CSP_a(G)_{1,s}^\ast)$ iff $H\in Y(\CSP(G))$. Therefore, the polynomial-time algorithm used to decide $\CSP(G)$ also decides $\CSP_a(G)_{1,s}^\ast$. The same holds for $c-c$ and $c-v$. 
\end{proof}

\end{document}